%% file: main-clean.tex
\documentclass[11pt]{article}
\usepackage[T1]{fontenc}
\usepackage{etoolbox}
\usepackage{graphicx}
\usepackage{latexsym}
\usepackage{amssymb}
\usepackage{amsmath}
\usepackage{amsthm}
\usepackage{thmtools}
\usepackage{forloop}
\usepackage[paper=letterpaper,margin=1in]{geometry}
\usepackage[ruled,vlined,linesnumbered]{algorithm2e}
\usepackage{mleftright}\mleftright
\usepackage{paralist}
\usepackage{graphics}
\usepackage{xcolor}
\usepackage{nicefrac}
\usepackage{scalerel}
\usepackage{threeparttable}
\usepackage{multirow}
\usepackage{mathtools}
\usepackage{bm}
\usepackage{array}
\usepackage{tikz}\usetikzlibrary{decorations} 
\usepackage{nicematrix}
\usepackage{caption}
\usepackage{hyperref}
\usepackage{cleveref}

\SetKw{Continue}{continue}

\makeatletter
\patchcmd{\@algocf@start}
  {-1.5em}
  {0pt}
  {}{}
\makeatother
\SetCommentSty{}

\definecolor{darkgreen}{rgb}{0,0.5,0}
\hypersetup{
    unicode=false,            
    colorlinks=true,          
    linkcolor=blue!70!black,  
    citecolor=darkgreen,      
    filecolor=magenta,        
    urlcolor=blue!70!black    
}

\newtheorem{theorem}{Theorem}[section]
\newtheorem{lemma}[theorem]{Lemma}

\newtheorem{corollary}[theorem]{Corollary}
\newtheorem{definition}{Definition}[section]
\newtheorem{proposition}[theorem]{Proposition}
\newtheorem{observation}[theorem]{Observation}

\newtheorem{claim}[theorem]{Claim}

\newcommand{\defcal}[1]{\expandafter\newcommand\csname c#1\endcsname{{\mathcal{#1}}}}
\newcommand{\defbb}[1]{\expandafter\newcommand\csname b#1\endcsname{{\mathbb{#1}}}}
\newcommand{\defvec}[1]{\expandafter\newcommand\csname v#1\endcsname{{\mathbf{#1}}}}
\newcounter{calBbCounter}
\forLoop{1}{26}{calBbCounter}{
    \edef\letter{\alph{calBbCounter}}
    \edef\Letter{\Alph{calBbCounter}}
    \expandafter\defcal\Letter
		\expandafter\defbb\Letter
		\expandafter\defvec\letter
}

\newcommand{\eps}{\varepsilon}
\newcommand{\ie}{{\it i.e.}}

\newcommand{\nnR}{{\bR_{\geq 0}}}
\newcommand{\email}[1]{{\href{mailto:#1}{#1}}}
\newcommand{\characteristic}{{\mathbf{1}}}

\newcommand{\vone}{{\mathbf{1}}}

\newcommand{\inner}[2]{\left<#1, #2\right>}

\newcommand{\poly}{{\mathtt{Poly}}}
\newcommand{\polylog}{{\mathtt{Polylog}}}

\let\oldnl\nl
\newcommand{\nonl}{\renewcommand{\nl}{\let\nl\oldnl}}

\title{Semi-Streaming Algorithms for Submodular\\Maximization under Random Arrival Order}

\author{Niv Buchbinder\thanks{Dept. of Statistics and Operations Research, Tel-Aviv University, Israel. E-mail: \email{niv.buchbinder@gmail.com}} \and
Moran Feldman\thanks{Department of Computer Science, University of Haifa, Israel. Work done while visiting School of Mathematics, Queen Mary University of London. E-mail: \email{moranfe@cs.haifa.ac.il}} \and
Siyue Liu\thanks{Tepper School of Business, Carnegie Mellon University (CMU). E-mail: \email{siyueliu@andrew.cmu.edu}} \and
Sherry Sarkar\thanks{Carnegie Mellon University (CMU). E-mail: \email{sherrys@andrew.cmu.edu}}}

\begin{document}

\maketitle

\input{abstract.tex}

\thispagestyle{empty}
\pagenumbering{Alph}
\clearpage
\setcounter{page}{1}
\pagenumbering{arabic}

\input{intro.tex}

\input{preliminaries}

\input{algorithms-matroid.tex}

\input{single-pass-alg.tex}

\input{stream-matroid-alg.tex}

\input{k-System-Hardness.tex}



\bibliography{bib}
\bibliographystyle{alpha}

\end{document}

%% file: abstract.tex
\begin{abstract}

We study random order semi-streaming algorithms for submodular maximization under a wide range of combinatorial constraint classes, including matroids, matroid $p$-parity, $p$-exchange systems and $p$-systems. For most of these classes of constraints, our results are the first improvement over what is known to be achievable for adversarial order. For matroids, matching and $p$-matchoids, previous random order results were known, and we improve over some of these as well. In the case of matroids, our improved results show a separation between adversarial and random order semi-streaming algorithms, and exponentially improve the number of passes necessary for getting $1 - 1/e - \eps$ approximation for maximizing a monotone submodular function subject to a matroid constraint. We also prove a new hardness result showing a similar separation for $p$-systems. Our results are based on two new technical tools. One tool provides a general way to translate offline algorithms for many classes of constraints into random order semi-streaming algorithms. The other tool is a semi-streaming variant of a recently proposed offline algorithm for matroid constraints.

\medskip

\end{abstract}

%% file: intro.tex
\section{Introduction}
Submodular maximization is a central topic in combinatorial optimization, studied extensively over the past five decades since the seminal works of \cite{fisher1978analysis,nemhauser1978analysis,nemhauser1978best}. A typical problem in this field consists of a submodular objective function $f:2^{\cN}\to \bR$ over a ground set $\cN$ and a constraint specifying the set $\cI \subseteq 2^\cN$ of feasible solutions.\footnote{A set function $f\colon 2^\cN \to \bR$ is submodular if for every pair of sets $S\subseteq T \subseteq \cN$ and element $u\in \cN\setminus T$, it holds that $f(u \mid S) \geq f(u \mid T)$, where $f(u \mid S) \triangleq f(S\cup \{u\}) - f(S)$ is the \emph{marginal gain} of adding $u$ to $S$.} The goal of the problem is to find a set $S \in \cI$ that (approximately) maximizes $f$ among all the sets of $\cI$. Naturally, the approximability of this problem is significantly affected by the properties of the constraint, and accordingly, many results have been obtained for the above problem for particular classes of constraints (see \cite{nemhauser1978analysis,calinescu2011maximizing,chekuri2014submodular,filmus2014monotone,sviridenko2017optimal,buchbinder2024constrained,buchbinder2024deterministic,buchbinder2025extending} for a few examples of such results). To make this more concrete, let us introduce some details. It is customary to identify the constraint with the pair $(\cN, \cI)$. The constraint is called an \emph{independence system} if $\cI$ is non-empty and \emph{down-closed} (i.e., if $A\subseteq B$ and $B\in \cI$, then $A\in \cI$). Independence systems are a very broad class of constraints, and already capture constraints that are hard to approximate even for linear objectives (e.g., maximum independent set in graphs). Thus, much of the literature focuses on structured subclasses of independence systems that do admit good approximation ratios. A well-known example of such a subclass is the class of matroids, which captures a diverse set of scenarios such as linear independence of vectors, forests in graphs and laminar families of bounds. A more general subclass of independence systems is the class of \emph{$p$-systems} (parametrized by an integer $p \geq 1$), which admits weaker approximation guarantees, but captures additional kinds of constraints such as (hyper) matching, asymmetric traveling salesperson and packing problems with low column sparsity.

As mentioned above, submodular maximization subject to classes of independence systems has been studied extensively. At first, these studies concentrated on the standard offline computational model (see the works referenced above). However, motivated by various applications, later studies considered submodular maximization in a variety of other computational models, including online and secretary settings~\cite{bateni2013submodular,buchbinder2020online,buchbinder2019online,feldman2018submodular}, parallel computation~\cite{balkanski2022adaptive2,balkanski2018adaptive1,chekuri2019parallel1,chekuri2019parallel2} and distributed computing~\cite{mirzasoleiman2016distributed,pontebarbosa2016new}.
Our focus in this work is the \emph{data stream} model, which is motivated by modern applications in data science and machine learning involving massive ground sets. In this model, the elements of the ground set $\cN$ arrive one after the other in a stream. The goal is to design an algorithm that scans the elements of $\cN$ in one pass over the stream (or a small number of passes), and then outputs a good feasible solution. Since the motivation for this model is in massive ground sets, an algorithm for this model is considered interesting only if its memory is small compared to the size $n$ of $\cN$. Note, however, that even storing the output set might require $\Theta(r)$ memory, where $r$ is the maximum size of a feasible set (in the context of independence system constraints, $r$ is known as the \emph{rank} of the independence system). Thus, it is natural to aim for algorithms whose space complexity is $O(r\cdot \polylog(n))$. Such algorithms are called \emph{semi-streaming} algorithms.\footnote{When the class of constraints is parametrized by $p$, we say that a data stream algorithm is a semi-streaming algorithm if its space complexity simplifies to $O(r\cdot \polylog(n))$ when $p$ is treated as a constant.}

Semi-streaming algorithms have been developed for submodular maximization problems involving many kinds of constraints, including cardinality constraints \cite{badanidiyuru2014streaming,kazemi2019submodular,alaluf2022optimal,FNSZ23}, matroid and matroid intersection constraints~\cite{chakrabarti2015submodular,chekuri2015streaming,S20,harshaw2022power,rubinstein2022maximizing,FNSZ23,garg2023semistreaming,NVZ25,feldman2026streaming}, matching~\cite{levin2021streaming}, $p$-systems~\cite{haba2020streaming} and knapsack~\cite{huang2020streaming,huang2021improved}.
Running time is typically a secondary concern in this setting. Accordingly, the literature includes both polynomial-time algorithms with running time $\poly(n)$ and \emph{fixed-parameter tractable} (FPT) algorithms with running time $g(r)\cdot \poly(n)$, where the dependence $g$ on $r$ may be arbitrary, while the dependence on $n$ remains polynomial. The standard data stream model assumes an adversarial arrival order of the elements. However, in many applications it is more realistic to assume that the arrival order is not adversarial, but rather uniformly random. This \emph{random-order} assumption, originating in the classical secretary problem and its many extensions (see, e.g.,~\cite{kleinberg2005multiple,lachish2014competitive,babaioff2018matroid,feldman2018simple,feldman2018submodular,soto2021strong}), is well-motivated and has become central in practice (see, e.g., \cite[Chapter 11]{R2020}).

\subsection{Our Results}

We design improved semi-streaming approximation algorithms under a uniformly random arrival order for maximizing linear, submodular and monotone submodular objectives.\footnote{A set function $f\colon 2^\cN \to \bR$ is \emph{monotone} if $f(S)\leq f(T)$ for every two sets $S\subseteq T \subseteq \cN$.} Our algorithms apply to a variety of classes of independence systems from the literature. Readers less familiar with these constraint classes are referred to Section~\ref{sec:preliminaries} for formal definitions and the relationships between them.
Our first result concerns single-pass semi-streaming algorithms.

\begin{theorem}[Single-pass algorithms]\label{thm:main2}
Let $\cM=(\cN,\cI)$ be an independence system of rank $r$, and let $f\colon 2^\cN\to \nnR$ be a non-negative function. For any sufficiently small $\eps>0$, there are single-pass semi-streaming algorithms that under a uniformly random arrival order store $O\!\left(\nicefrac{r}{\eps^2}\log^2(r/\eps)\right)$ elements and achieve the approximation ratios appearing in bold on the two rightmost columns of Table~\ref{table2}---up to an error term of $O(\eps)$.
\end{theorem}

\begin{table}[p]
\begin{center}\small
\begin{threeparttable}
\begin{NiceTabular}{c|c|c|c|c|c}
\multirow{3}{*}{Constraint} & \multirow{3}{*}{\hspace{-2mm}Obj.\hspace{-2mm}} & \multirow{3}{*}{Offline} & \multicolumn{1}{m{32mm}}{\multirow{3}{32mm}{\centering Semi-Streaming Adversarial Order}} & \multicolumn{2}{c}{Semi-Streaming}\\
&&&& \multicolumn{2}{c}{Random Order}\\
\cline{5-6}
&&&& FPT & Poly\\
\hline 
\multirow{5}{*}{Matroid} & L & $1$~\cite{rado1957note} & $1$\tnote{a} & $-$ & $-$ \\[-1pt] \cdashedline{2-6} &&&&&\\[-10pt]
& \multirow{2}{*}{\hspace{-1mm}MS\hspace{-1mm}} & \multirow{2}{*}{$\frac{e}{e - 1} \sim 1.582$~\cite{calinescu2011maximizing}} & $3.147$~\cite{feldman2026streaming} & $\bm{1.801}$\tnote{b} & $\bm{1.972}$\tnote{b}  \\ 
& & & H: $2$~\cite{FNSZ23}\tnote{c} & \hspace{-1mm}($2$~\cite{NVZ25})\hspace{-1mm} & ($2.314$~\cite{S20}) \\ \cdashedline{2-6} &&&&&\\[-10pt]
& \multirow{2}{*}{\hspace{-1mm}NS\hspace{-1mm}} & $2.494$~\cite{buchbinder2024constrained} & \multirow{2}{*}{5.206~\cite{feldman2026streaming}} & 2~\cite{NVZ25} & \multirow{2}{*}{$\bm{3.494}$} \\ 
& & \hspace{-1mm}FPT: $\frac{e}{e - 1} \sim 1.582$ \cite{NVZ25}\hspace{-1mm} & & \hspace{-1mm}H: $2$~\cite{rubinstein2022maximizing}\tnote{d} &\\
\hline
\multicolumn{1}{m{17mm}}{\hspace*{-2.5mm}\multirow{3}{22mm}{\centering Matching / Matroid $2$\nolinebreak-Intersection}\hspace*{-3mm}} & L & $1$~\cite{cunningam1986improved,edmonds1965paths} & $2$~\cite{garg2023semistreaming,paz2017approximation} & \hspace*{-2mm}$1.5$~\cite{hashemi2024weighted}\hspace*{-2mm} & $1.5$~\cite{hashemi2024weighted} \\ 
& \hspace{-1mm}MS\hspace{-1mm} & $2$~\cite{feldman2011improved,lee10submodular} & \hspace{-1mm}$3 + 2\sqrt{2}$~\cite{garg2023semistreaming,levin2021streaming}\hspace*{-1mm} & $\bm{3}$ & $3.158$~\cite{S20} \\
& \hspace{-1mm}NS\hspace{-1mm} & $4$~\cite{feldman2011improved,lee10submodular} & \hspace{-1mm}$4 + 3\sqrt{2}$~\cite{garg2023semistreaming,levin2021streaming}\hspace*{-1mm} & $\bm{3}$ & $\bm{6}$ \\
\hline
\multicolumn{1}{m{17mm}}{\hspace*{-2.5mm}\multirow{6}{22mm}{\centering Matroid $p$-Parity / $p$-Matchoid / Matroid $p$-Intersection}\hspace*{-3mm}} & \multirow{2}{*}{L} & \multirow{2}{*}{$(p + 1) \ln 2$~\cite{feldman2026submodular}} & $4p - 1$~\cite{chakrabarti2015submodular}\tnote{e} & \multirow{2}{*}{$\bm{p}$} & $\bm{(1 + \ln 2)p}$\tnote{g}\\
&&& H: $p$~\cite{FNSZ22}\tnote{f} && H: $\frac{p}{12}$~\cite{lee2025asymptotically}\tnote{h}\hspace*{-2mm}\\  \cdashedline{2-6} &&&&&\\[-11pt]
& \multirow{2}{*}{\hspace{-1mm}MS\hspace{-1mm}} & \multirow{2}{*}{\hspace{-1mm}$\frac{2p\ln 2}{1 + \ln 2} + O(\sqrt{p})$~\cite{feldman2026submodular}\hspace*{-1mm}}& \multirow{2}{*}{$4p$~\cite{chakrabarti2015submodular,chekuri2015streaming}\tnote{i}} & \multirow{2}{*}{$\bm{p + 1}$} & \hspace{-2mm}$\bm{\frac{p + 3p\ln 2 + O(\sqrt{p})}{1 + \ln 2}}$\hspace*{-2mm}\\ 
&&&&& \hspace{-1mm}$\frac{p + 1}{1 - 1/e^{p + 1}}$~\cite{S20}\tnote{i}\hspace*{-1mm}\\\cdashedline{2-6} &&&&& \\[-11pt]
& \hspace{-1mm}NS\hspace{-1mm} & \hspace{-1mm}$\frac{2p\ln 2}{1 + \ln 2} + O(p^{\nicefrac{2}{3}})$~\cite{feldman2026submodular}\hspace*{-1mm}& $4p+2$~\cite{harshaw2022power}\tnote{i,j}& $\bm{p + 1}$ & \hspace{-1mm}$\bm{\frac{p + 3p\ln 2 + O(p^{\nicefrac{2}{3}})}{1 + \ln 2}}$\hspace*{-2mm}\\
\hline
\multirow{4}{*}{$p$-Exchange} & \multirow{2}{*}{L} & \multirow{2}{*}{$\frac{p + 1}{2}$~\cite{feldman2011improved}} & \multirow{4}{32mm}{\centering \hspace{-1mm}$O(p \log p)$~\cite{haba2020streaming}\hspace*{-1mm}} & \multirow{2}{*}{$\bm{p}$} & $\bm{\frac{3p-1}{2}}$ \\
&&&&& H: $\frac{p}{12}$~\cite{lee2025asymptotically}\tnote{h}\hspace*{-2mm}\\[1pt] \cdashedline{2-3} \cdashedline{5-6}&&&&&\\[-11pt]
& \hspace{-1mm}MS\hspace{-1mm} & $\frac{p + 3}{2}$~\cite{ward2012approximation} & & $\bm{p + 1}$ & $\bm{\frac{3p + 3}{2}}$ \\[1pt] \cdashedline{2-3} \cdashedline{5-6} &&&&&\\[-10pt]
& \hspace{-1mm}NS\hspace{-1mm} & $p + 1 + \frac{1}{p - 1}$~\cite{feldman2011improved} && $\bm{p + 1}$ & $\bm{2p + 1 + \frac{1}{p - 1}}$ \\
\hline
\multicolumn{1}{m{17mm}}{\hspace*{-2.5mm}\multirow{5}{22mm}{\centering $p$-Extendible / $p$-System}\hspace*{-3mm}} & \multirow{2}{*}{L} & \multirow{2}{*}{$p$~\cite{jenkyns1976efficacy}} & & \multirow{2}{*}{$\bm{p + 1}$} & $\bm{2p}$  \\
&&& \hspace{-1mm}$O(p \log p)$~\cite{haba2020streaming}\tnote{l}\hspace*{-1mm} && H: $p$~\cite{feldman2023how}\tnote{k}\\ \cdashedline{2-3} \cdashedline{5-6} &&&&&\\[-10pt]
& \hspace{-1mm}MS\hspace{-1mm} & $p + 1$~\cite{fisher1978analysis} & \hspace{-1mm}$O(p^2 \log p)$~\cite{haba2020streaming}\hspace{-1mm} & $\bm{p + 2}$ & $\bm{2p + 2}$ \\ \cdashedline{2-3} \cdashedline{5-6} &&&&&\\[-11pt]
& \multirow{2}{*}{\hspace{-1mm}NS\hspace{-1mm}} & $p + 2 + 1/p$~\cite{feldman2023how}\tnote{l} & H:  $\bm{\Omega(p^2)}$ & \multirow{2}{*}{$\bm{p + 2}$} & $\bm{2p + 3 + 1/p}$\tnote{l}\\
&& $p + O(\sqrt{p})$~\cite{feldman2023how} &&& $\bm{2p + O(\sqrt{p})}$\\
\end{NiceTabular}

\begin{tablenotes}\setlength{\itemsep}{1pt}
{\footnotesize
\item[a] Folklore. Maintain a local optimum solution via swaps.
\item[b] Hardness of $\frac{e}{e - 1}$ holds even for a cardinality constraint. For poly-time algorithms, this hardness applies even to offline algorithms~\cite{nemhauser1978best}. McGregor and Vu~\cite{MV19} showed the same hardness for data stream algorithms using $o(n/r^2)$ memory, even if they are allowed arbitrary amount of computation. Technically, the hardness of~\cite{MV19} was for the related $k$-Coverage problem, but it was adapted by \cite{FNSZ22} to the current problem.
\item[c] Applies even for a cardinality constraint and algorithms that use $o(n/r^3)$ memory. Tight in the sense that an approximation ratio of $2$ can be obtained using space that is exponential in $r$ but independent of $n$~\cite{feldman2026streaming}.
 \item[d] Applies even for a cardinality constraint and algorithms that use $o(n/r^2)$ memory.
\item[e] Applies only to matroid $p$-intersection. The approximation ratio proved by~\cite{chakrabarti2015submodular} was $2p + 2\sqrt{p(p - 1)} - 1$.
 \item[f] Applies even for a cardinality objective function and algorithms using $o(\frac{\eps n}{p^5 \log p})$ memory, but assumes that one is given access only to one independence oracle for the entire constraint (rather than one oracle per matroid). If an independence oracle is available for each matroid, there is a weaker inapproximability excluding $o(\frac{p}{\log p})$-approximation by a semi-streaming algorithm, but $1$-approximation can be obtained in this case for $p$-matchoid constraints using space that is independent of $n$ but exponential in $r$.~\cite{huang2023fpt}. A similar result was proved by~\cite{FNSZ22} for (not necessarily monotone) submodular functions, but it guarantees a slightly weaker approximation ratio of $2$ and applies only to matroid $p$-intersection constraints. 
 \item[g] More accurately, the ratio is $(1 + \ln 2)p - (1 - \ln 2)$.
 \item[h] Applies even for offline algorithms and the special case of $p$-dimensional matching.
 \item[i] Does not apply to matroid $p$-parity. 
\item[j] The approximation ratio proved by~\cite{harshaw2022power} for non-monotone functions was $2p + 2\sqrt{p(p + 1)} + 1$.
\item[k] Applies even to offline algorithms. For submodular functions, the hardness slightly improves to $(1 - e^{-1/p})^{-1}$.
\item[l] Applies only to $p$-extendible systems.
}
\end{tablenotes}
\end{threeparttable}
\end{center}
\vspace{-1.2mm}
\caption{Results for a single-pass semi-streaming algorithms (ignoring dependencies on $\eps$). Our new results appear in bold. Results we improve over appear in parenthesis. Shorthands: 'L' (Linear), 'MS' (Monotone Submodular), 'NS' (Non-monotone Submodular), 'H' (Hardness), 'FPT' (fixed parameter tractable).}\label{table2}
\end{table}

Let us highlight the consequences of several of the results we report in Table~\ref{table2}.
\paragraph{Matroid Constraints and Monotone Submodular Objective.}
We obtain $1.972$-approx\-imation in polynomial time and $1.801$-approximation in FPT time, improving over the $2.315$-approximation of \cite{S20} and the $2$-approximation of \cite{NVZ25}, respectively. Our results are the first to improve over $2$-approximation, which shows a separation between random-order and adversarial-order semi-streaming algorithms because \cite{FNSZ23} proved a lower bound of $2$ on the approximation ratio of any adversarial-order semi-streaming algorithm, even when one allows FPT time and even for the simpler class of cardinality constraints. We note that $2$-approximation is also a natural barrier for many natural algorithmic techniques. In fact, the state-of-the-art offline approximation ratio was $2$ for $30$ years since the work of~\cite{fisher1978analysis}, until it was finally improved over by~\cite{calinescu2011maximizing}. 

\vspace{-4mm}
\paragraph{Matching / Intersection of Two Matroids Constraint.} 

We obtain $3$-approximation in FPT time for both monotone and non-monotone submodular objectives. For monotone submodular functions, this matches the guarantee of the offline (polynomial time) greedy algorithm and slightly improves over the polynomial-time semi-streaming $3.158$-approximation of \cite{S20}. For non-monotone submodular functions, it outperforms the current state-of-the-art polynomial-time offline algorithm, which achieves 4-approximation~\cite{feldman2011improved,lee10submodular}. We also get for non-monotone submodular objectives $6$-approximation in polynomial time, substantially improving over the best known adversarial order bound of $4+3\sqrt{2}\approx 8.24$~\cite{garg2023semistreaming,levin2021streaming}. 


\vspace{-4mm}
\paragraph{$p$-Systems Constraints.}
We obtain $2p$-approximation (up to lower-order terms) in polynomial time and $p$-approximation (up to an additive term of $1$ or $2$) in FPT time. These guarantees nearly match the lower bounds known for polynomial time offline algorithms~\cite{badanidiyuru2014fast,feldman2023how} and improve quadratically over the $O(p^2\log p)$-approximation known for adversarial-order streaming. In addition, we prove below (Theorem~\ref{thm:hardness}) a new lower bound of $\Omega(p^2)$ for adversarial-order semi-streaming algorithms, providing another separation between random-order and adversarial-order.

Next, we study multi-pass semi-streaming algorithms. A natural question here is how many passes over the stream are needed to match the best offline approximation. Our result in this context is for maximizing a monotone submodular function under a matroid constraint, which is one of the most central problems in submodular maximization. The optimal offline approximation factor for this problem is $\frac{e}{e-1}\approx 1.582$ (optimal already for a cardinality constraint) \cite{nemhauser1978best,nemhauser1978analysis,calinescu2011maximizing}. In the data stream model, \cite{feldman2026streaming} showed that this approximation ratio can be achieved by a semi-streaming algorithm using $O(\eps^{-3})$ passes under adversarial order, and $O(\eps^{-2}\log(1/\eps))$ passes under random order. We improve exponentially over the latter, showing that $O(\log(1/\eps))$ passes suffice in the random-order model.

\begin{restatable}[Multi-pass Algorithm]{theorem}{thmMain}\label{thm:main}
    Let $\cM = (\cN, \cI)$ be a matroid of rank $r$ and $f\colon 2^\cN \rightarrow \nnR$ be a non-negative monotone submodular function. There is a semi-streaming algorithm in the random-order data stream model that given a small enough value $\eps > 0$, makes $O(\log(1/\eps))$ passes, stores $O(\nicefrac{r}{\eps})$ elements, and achieves $(\frac{e}{e-1}+\eps)$-approximation.  
\end{restatable}

\subsection*{Techniques:}

Our main results rely on two new algorithmic components. The first component is a single-pass semi-streaming algorithm with the guarantee stated by the next proposition. Intuitively, this algorithm generates two sets $S_\delta$ and $H$ such that either $S_\delta$ is a high value feasible set, or there exists a high value feasible subset in $S_\delta \cup H$.

\begin{restatable}{proposition}{PropMainH}\label{prop:main2}
Let $f\colon 2^\cN \rightarrow \nnR$ be a non-negative submodular function and let $\cM = (\cN, \cI)$ be a $p$-system of rank $r$. There exists a single-pass semi-streaming algorithm in the random-order model that, given $\delta \in (0, 1)$,
\begin{itemize}
\item after seeing only the first $\lfloor \delta \cdot n \rfloor$ elements of the stream, outputs a set $S_{\delta}\in \cI$, and
\item after scanning the entire stream, outputs a second set $H$. The size of $H$ is $O(\nicefrac{r}{\delta^2}\log^2(r/\delta))$, and this expression also bounds the space complexity of the algorithm.
\end{itemize}
Furthermore,
\[
(1+p)\cdot \bE[f(S_{\delta})]+\bE[f(A^*)] \geq (1 - O(\delta p)) \cdot f(OPT) \enspace,
\]
where $A^* \triangleq \arg \max\{f(A)\,:\, A\subseteq(S_{\delta}\cup H), A \in\cI\}$ and $OPT$ is an optimal independent subset of $\cN$. The inequality above can be strengthened in special cases. In particular, the coefficient of $\bE[f(S_{\delta})]$ can be reduced to $p$ if one of the following conditions holds, and to $p-1$ if both hold.
\begin{itemize}\setlength\itemsep{0pt}
	\item Condition~1: $(\cN, \cI)$ is a $p$-exchange system or a matroid $p$-parity.
	\item Condition~2: $f$ is linear.
\end{itemize}
\end{restatable}

Since $|S_{\delta}\cup H| = O(\nicefrac{r}{\delta^2}\log^2(r/\delta))$ in Proposition~\ref{prop:main2}, it is possible to compute the set $A^*$ in FPT time via exhaustive search. Then, by returning the better set among $S_\delta$ and $A^*$, one can get $p + O(1)$ approximation. In the polynomial time case, one cannot in general obtain $A^*$. However, it is still possible to use an offline algorithm to compute a feasible subset $A$ of $S_{\delta}\cup H$ that approximates $A^*$. Outputting the better set among $S_\delta$ and $A$ then yields some approximation guarantee that depends on the approximation ratio $\alpha$ of the offline algorithm.
Corollary~\ref{cor:offlinetoStream} formally states the approximation ratio that can be obtained in this way. Most of our new algorithmic results stated in Table~\ref{table2} follow by plugging offline algorithms referenced in the table into this corollary.\footnote{One can verify that all these offline algorithms can be implemented to use only nearly-linear space, and thus, keep the space complexity of the data stream algorithm they are part of within the allowed limit for a semi-streaming algorithm.} The only exception are the results about maximizing a monotone submodular function subject to a matroid constraint, which require an additional tool, and are further discussed below.

\begin{corollary}\label{cor:offlinetoStream}
\begin{samepage}
Assume that
\begin{itemize}
	\item in Proposition~\ref{prop:main2} we have $c(p)\cdot \bE[f(S_{\delta})]+\bE[f(A^*)] \geq (1 - \beta \delta p) \cdot f(OPT)$---where $\beta > 0$ is the constant hidden by the big $O$ notation; and
	\item given the available time and space, we can compute a set $A\subseteq S_{\delta}\cup H$ such that $f(A)\geq \nicefrac{1}{\alpha}\cdot f(A^*)$ for some value $\alpha \geq 1$.
\end{itemize}
\end{samepage}
Then, by setting $\delta = \frac{\eps}{2\beta p(c(p) + \alpha)}$ and outputting the better set among $S_\delta$ and $A$, one can get an approximation ratio of
$
	\alpha + c(p) + \eps 
$.
\end{corollary}
\begin{proof}
Note that
\begin{align*}
  \max\{f(S_{\delta}), f(A)\}
  &\geq \frac{c(p)}{c(p)+\alpha}\cdot f(S_{\delta})+\frac{\alpha}{c(p)+\alpha}\cdot f(A)\\
  &\geq \frac{c(p)\cdot f(S_{\delta})+ f(A^*)}{c(p)+\alpha}
  \geq \frac{1-\beta\delta p}{c(p)+\alpha} \cdot f(OPT)
	\geq
	\frac{f(OPT)}{c(p) + \alpha + \eps}
	\enspace.
	\qedhere
\end{align*}
\end{proof}

The proof of Proposition~\ref{prop:main2} appears in Section~\ref{sec:single-pass}. It builds on algorithmic ideas introduced recently by \cite{NVZ25} for non-monotone submodular maximization under a matroid constraint, but we substantially simplify and extend them.
At a high level, the first phase of the algorithm scans the initial $\delta n$ elements and greedily constructs a feasible set $S_\delta\in\cI$ by selecting the most valuable element from each ``window'' of $\nicefrac{\delta n}{r}$ elements. 
In its second phase, the algorithm scans the remaining $(1-\delta)n$ elements and adds each element $u$ to $H$ whenever $u$ can ``replace'' some $s\in S_\delta$ and has a marginal contribution that is strictly larger than the one of $s$ at the time $s$ entered $S_\delta$. With a small probability this process might collect too many elements, which forces the algorithm to ``abort'' to preserve its space complexity guarantee. However, whenever this algorithm does not abort, it collects enough valuable elements into $H$ to guarantee that $S_\delta \cup H$ contains a good solution. Note that since the decision whether to store in $H$ each element of the second phase is done independently for each element, this algorithm is not affected by the order of the $(1-\delta)n$ elements in this second phase. Thus, the algorithm works correctly even when they arrive in an adversarial, rather than uniformly random, order.

To get our improved guarantees for monotone submodular maximization under a matroid constraint, we need a second component, which is motivated by the following general idea. The set $S_{\delta}$ is computed after observing only $\delta n$ elements from the stream, so in general we cannot expect $f(S_\delta)$ to be large. Fortunately, however, when $f(S_{\delta})$ is indeed small, Proposition~\ref{prop:main2} already implies that the best solution inside $S_{\delta}\cup H$ (namely $A^*$) is close to optimal. Therefore, the more problematic case from the point of view of Proposition~\ref{prop:main2} is when $f(S_{\delta})$ happens to be substantial. To handle this case, we take advantage of the fact that $f(S_\delta)$ is substantial, and treat it as a starting point for a process that \emph{boosts} its value further while processing the remaining $(1-\delta)n$ elements. Note that typically most of the elements of the optimal solution are among these $(1-\delta)n$ elements, and thus, the boosting step indeed encounters valuable elements that can be used to increase the value of the solution. The next proposition formalizes this boosting step (see Section~\ref{sec:preliminaries} for the definition of the multilinear extension).

\begin{restatable}[Single-pass random-order semi-streaming with boosting]{proposition}{PropMainBoost}
\label{prop:main}
Let $f\colon 2^\cN \rightarrow \nnR$ be a non-negative monotone submodular function, and let $\cM = (\cN \cup A, \cI)$ be a matroid such that $A \in \cI$ (possibly with $A \cap \cN \neq \varnothing$). There exists a single-pass semi-streaming algorithm that gets the elements of $A$ upfront, and the elements of $\cN$ in the form of a random order data stream. Given $\delta>0$, this algorithm stores $O(\nicefrac{r}{\delta})$ elements and outputs a set $S\in \cI$ such that
\[
    \bE[f(S)] \geq
    (1 - e^{-1+1/e} - \delta) \cdot f(OPT) + e^{-2+1/e} \cdot f(A)
    \geq 0.468 \cdot f(OPT) + 0.195 \cdot f(A)\enspace,
\]
where $OPT$ is an optimal independent subset of $\cN$. More generally, given $h\in(0,1/e]$ (the above guarantee corresponds to $h=\nicefrac{1}{e}$), the algorithm outputs a set $S\in \cI$ such that
\[
\bE[F(e\cdot h \cdot \characteristic_S )]\ge (1-\lambda)(1 - \delta)\cdot f(OPT)+\lambda\cdot F(h\cdot \characteristic_{A}) \enspace,
\]
where $\lambda=e^{-h(e-1)} \in[0.53,1]$ and $F$ is the multilinear extension of $f$.
\end{restatable}

The algorithm of Proposition~\ref{prop:main} is inspired by an offline algorithm due to Ganz et al.~\cite{ganz2026poisson}. Both algorithms maintain an integral solution $S$, but consider its value to be $F(h \cdot \characteristic_S)$ for some $h \in (0, 1]$ that grows over time. Both algorithms also occasionally improve the solution $S$ via local search like swaps. However, they differ regarding the way the times of these swaps are chosen. The algorithm of~\cite{ganz2026poisson} uses a Poisson process for this purpose, while the algorithm of Proposition~\ref{prop:main} partitions the stream into blocks, or windows, and makes a swap after inspecting each block. One complication in this algorithm is that the locations of elements in the stream are correlated with the solution that the algorithm has at every time point, which requires us to adopt tools from the literature about short list secretary algorithms~\cite{agrawal2019submodular,S20,feldman2026streaming}. We prove Proposition~\ref{prop:main} in Section~\ref{sec:mainAlg}. Along the way, we present in this section also a simple version of the offline algorithm of~\cite{ganz2026poisson} that avoids the use of a Poisson process, and might be of independent interest.

In Section~\ref{sec:matroid}, we show that the results of Theorem~\ref{thm:main2} for monotone submodular maximization subject to matroid constraint follow by combining Propositions~\ref{prop:main2} and~\ref{prop:main}. In the same section, we also show that Theorems~\ref{thm:main} follows by repeated applications of Propositions~\ref{prop:main}.

\paragraph{Additional result:}

Our final result is an inapproximability result. In the offline setting, one can get $p$-approximation for maximizing a linear function subject to a $p$-system constraint via a standard greedy algorithm~\cite{jenkyns1976efficacy}, and slightly worse guarantees can be obtained also for submodular functions~\cite{calinescu2011maximizing,feldman2023how,fisher1978analysis}. However, the best known semi-streaming algorithm only obtains $O(p^2 \log p)$-approximation for this problem under an adversarial arrival order~\cite{haba2020streaming}. Theorem~\ref{thm:main2} above shows that one can improve over that by assuming a uniformly random arrival order, and get an approximation ratio of $O(p)$, like in the offline case. Theorem~\ref{thm:hardness} below shows that a similar improvement cannot be obtained for adversarial arrival order, and thus, proves a separation between adversarial and uniformly random arrival orders.

To be more specific, Theorem~\ref{thm:hardness} shows that any data stream algorithm obtaining an approximation ratio of $o(p^2)$ for maximizing a linear function subject to a $p$-system must use a memory that is linear in the size $n$ of the ground set. Furthermore, this is true even when the rank the $p$-system is arbitrarily small compared to $n$, which excludes the possibility of a semi-streaming algorithm for the above problem with an approximation ratio of $o(p^2)$. Similarly to the hardness result of~\cite{alaluf2022optimal}, Theorem~\ref{thm:hardness} applies to what we call \emph{non-encoding algorithms}, i.e., algorithms that store elements in their memory in an explicit form rather than via some encoding or sketching. This natural class of algorithms captures every algorithm that we are aware of for submodular maximization. See Section~\ref{sec:hardness} for the formal definition of non-encoding algorithms and for the proof of Theorem~\ref{thm:hardness}.

\begin{restatable}{theorem}{thmHardness} \label{thm:hardness}
Under an adversarial arrival order, every $o(p^2)$-approximation non-encoding data stream algorithm for maximizing a linear function subject to a $p$-system constraint must store at least $\Omega(n / p^{2p + 5})$ elements. Moreover, this is true even when we are guaranteed that (i) the rank of the $p$-system is $O(p^{2p})$ and (ii) $n$ is larger than any given function of $p$.
\end{restatable}

%% file: preliminaries.tex
\section{Preliminaries} \label{sec:preliminaries}

In this section, we present some notation that is used in this paper, as well as known definitions and results. 

\subsection*{Submodular Functions}

A set function $f$ over a ground set $\cN$ is a function of the form $f\colon 2^\cN \to \bR$. Recall that given a set $S \subseteq \cN$ and an element $u \in \cN$, we denote by $f(u \mid S) \triangleq f(S \cup \{u\}) - f(S)$ the marginal contribution of $u$ to $S$. Similarly, given two sets $S, T \subseteq \cN$, we use $f(T \mid S) \triangleq f(S \cup T) - f(S)$ to denote the marginal contribution of $T$ to $S$. For ease of the reading, we often use $S + u$ and $S - u$ as shorthands for $S \cup \{u\}$ and $S \setminus \{u\}$, respectively. 

Above we have given one definition of submodularity. Specifically, we stated that a set function $f\colon 2^\cN \to \bR$ is submodular if $f(u \mid S) \geq f(u \mid T)$ for every two sets $S \subseteq T \subseteq \cN$ and element $u \in \cN \setminus T$. Another well-known equivalent definition of submodularity is that $f$ is submodular if $f(S) + f(T) \geq f(S \cap T) + f(S \cup T)$ for every two sets $S, T \subseteq \cN$. Following is one property of submodular functions that we employ in this paper.
\begin{lemma}[Lemma~2.2 of~\cite{FMV11}] \label{lem:sampling}
Let $f\colon 2^\cN \to \bR$ be a submodular function, and let $A$ be a subset of $\cN$. If $A(p)$ is a random set that contains every element of $A$ with probability $p$, not necessarily independently, then $\bE[f(A(p))] \geq (1 - p) \cdot f(\varnothing) + p \cdot f(A)$.
\end{lemma}

\subsection*{Extensions and their properties}

The sets of $2^{\cN}$ can be naturally identified with vertices of the $[0, 1]^\cN$ cube by identifying  each set $A \subseteq \cN$ with its characteristic vector (which we denote here by $\characteristic_A$). For linear functions, there is a natural way to extend any function from $2^\cN$ to the cube $[0, 1]^\cN$. For submodular functions, the situation is more involved, and several useful extensions have been proposed. We use two such extensions. The first is the Lov\'{a}sz extension, defined as follows.
\begin{definition}\label{def:lov}{(The Lov\'{a}sz extension)}
Given a vector $\vx \in [0, 1]^\cN$ and a scalar $\lambda \in [0, 1]$, let $T_\lambda(\vx)=\left\{u \in \cN \mid x_u\geq \lambda \right\}$ be the set of elements in $\cN$ whose coordinate in $\vx$ is at least $\lambda$. Then,
\[
    \hat{f}(\vx) = \int_{\lambda=0}^1 f(T_\lambda(\vx)) d\lambda
    \enspace.
\]
\end{definition}

It is known that the Lov\'{a}sz extension of a submodular function $f$ is equivalent to the convex closure of $f$, which yields the following lemma.

\begin{lemma}[Due to \cite{L83}] \label{lem:lovasz_lower_bound}
Let $f\colon 2^\cN \to \bR$ be a submodular function, and let $\hat{f}$ be its Lov\'{a}sz extension. For every vector $\vx \in [0, 1]^\cN$, let $D_\vx$ denote an arbitrary distribution over $\{0,1\}^\cN$ such that $\Pr[u \in D_\vx] = x_u$ for every $u \in \cN$ (\ie, the marginals of $D_\vx$ agree with $\vx$). Then, $\hat{f}(\vx) \leq \bE_{A \sim D_\vx}[f(A)]$.
\end{lemma}

The second extension is the \emph{multilinear extension}, defined by \cite{calinescu2011maximizing}. 
\begin{definition}\label{def:multilinear}
    The \emph{multilinear extension} $F\colon[0,1]^\cN\to \mathbb{R}$ of a set function $f\colon 2^\cN\to \mathbb{R}$ is given by \[F(\vx)=\sum_{S\subseteq \cN} \bigg(\prod_{u\in S} x_u \prod_{u\in\cN \setminus S} \mspace{-9mu} (1-x_u) \cdot f(S)\bigg) \enspace.\]
\end{definition}
Notice that the multilinear extension is, as is implied by its name, a multlinear function of the coordinates of the input vector $\vx$. We explore additional properties of the multilinear extension that we need for our analysis in Section~\ref{sec:mainAlg}.

\subsection*{Independence Systems}
Recall that an independence system is a pair $(\cN, \cI)$, where $\cN$ is a ground set and $\cI \subseteq 2^\cN$ is a non-empty collection of subsets of $\cN$ that is down-closed in the sense that $T \in \cI$ implies that every set $S \subseteq T$ also belongs to $\cI$. Following the terminology used for linear spaces, it is customary to refer to the sets of $\cI$ as \emph{independent} sets. Given a set $S\subseteq \cN$, every inclusion-wise maximal subset of $T$ is called a \emph{base} of $T$, and the rank of $S$, denoted by $r(S)$, is defined as the size of the largest base of $S$. It is customary to refer to the rank and bases of $\cN$ as the rank and bases of the independence system $(\cN, \cI)$ itself. In this paper, we denote this rank by $r$.

Many classes of independence system have been studied. Following is a partial lists of these classes, and their definitions. Notice that many of the definitions in the list are parametrized by a positive integer value $p$, and the classes they define become larger when the value of $p$ grows. In general, each one of the classes in the list strictly contains the ones that appear before it (for the same value of $p$). The sole exception is the class of $p$-exchange, which is contained in the class of $p$-extendible, but has a non-trivial intersection with each one of the classes preceding it in the list. See Figure~\ref{fig:classes} for a graphical representation of containment relationship between the classes.

\begin{itemize}
    \item {\bf Matroid:} For every set $S \subseteq \cN$, all the bases of $S$ are of the same size. Equivalently, for every two sets $S, T \in \cI$ such that $|S| < |T|$, there must exist an element $u \in T \setminus S$ such that $S + u \in \cI$.
		\item {\bf Matroid $p$-Intersection:} There exist $p$ matroids $(\cN, \cI_1), (\cN, \cI_2), \dotsc, (\cN, \cI_p)$ such that $\cI = \cap_{i = 1}^p \cI_i$. It is customary to refer to the constraint $(\cN, \cI)$ as the intersection of these $p$ matroids.
		\item {\bf $p$-Matchoid:} There exist matroids $(\cN_1, \cI_1), (\cN_2, \cI_2), \dotsc, (\cN_h, \cI_h)$ such that every element $u \in \cN$ appears in the ground sets of at most $p$ of these matroids, and a set $S \subseteq \cN$ is independent if $S \cap \cN_i \in \cI_i$ for every $i \in [h]$.
    \item {\bf Matroid $p$-Parity:} There exists a matroid $(E, \cJ)$ and a function $v \colon \cN \to 2^E$ mapping every element of $\cN$ to a disjoint subset of $E$ of size at most $p$ such that $S \subseteq \cN$ is independent if $\cup_{u \in S} v(u) \in \cJ$.
    \item {\bf $p$-Exchange:} For every two independent sets $S, T \subseteq \cI$, there is a collection of sets $\{Y_u \subseteq S \setminus T \mid u \in T \setminus S\}$ such that (i) every set in the collection is of size at most $p$, (ii) every element of $S \setminus T$ appears in at most $p$ collection sets, and (iii) for every subset $T' \subseteq T \setminus S$, $(S \setminus \cup_{u \in T'} Y_u) \cup T' \in \cI$.
    \item {\bf $p$-Extendible:} For every two independent sets $S, T \subseteq \cI$ and element $u \in T \setminus S$, there exists a set $Y \subseteq S \setminus T$ of size at most $p$ such that $S \setminus Y + u \in \cI$.
    \item {\bf $p$-System:} For every set $S \subseteq \cN$, the ratio between the sizes of the largest and smallest bases of $S$ is at most $p$.
\end{itemize}

\begin{figure}[t!]
\begin{center}
\input{class-diagram.tikz}
\end{center}
\caption{Containment relationship between classes of independence systems.} \label{fig:classes}
\end{figure}
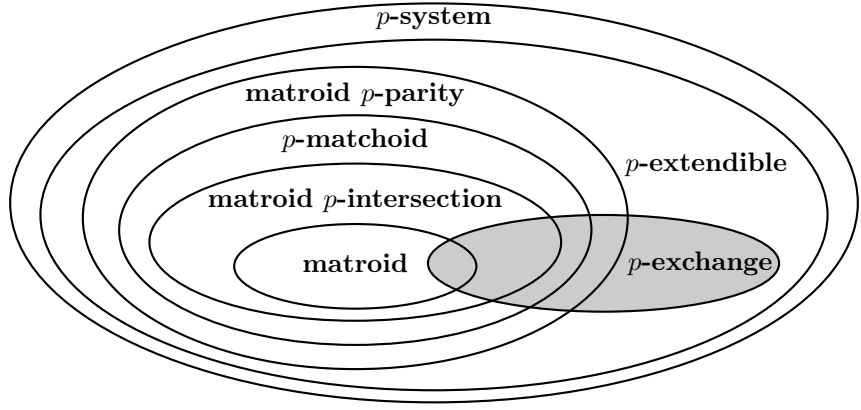

\subsection*{Oracles}

Explicit descriptions of set functions and independence system constraints over a ground set $\cN$ might require space that is exponential in $n = |\cN|$. However, algorithms processing such functions and constraints are usually only useful for applications when they run in time $\poly(n)$. Accordingly, it is standard in such algorithms to assume that the description of the function or constraint is not part of the input. Instead, the algorithms are given access to these objects through an appropriate oracle.
\begin{itemize}
	\item If the algorithm needs access to a set function $f$, then it gets a \emph{value oracle} that given a set $S$ returns $f(S)$.
	\item If the algorithm needs access to an independence system $(\cN, \cI)$, then it gets an \emph{independence oracle} that given a set $S$ answers whether $S \in \cI$.
\end{itemize}

It is sometimes convenient to assume that the algorithm also has value oracle access to the multilinear extension $F$ of $f$, i.e., an oracle that given a vector $\vx \in [0, 1]^\cN$, returns $F(\vx)$. As was shown by C{u{a}}linescu et al.~\cite{calinescu2011maximizing}, an approximate version of this oracle can be efficiently implemented by viewing $F(\vx)$ as the expected value of $f$ over sets drawn from an appropriate distribution, and then averaging the value of $f$ over multiple samples from this distribution. For simplicity, we ignore in this paper the need to approximate the oracle via sampling in this way, and just assume when necessary that we have value oracle access to $F$.

%% file: class-diagram.tikz
\begin{tikzpicture}[font=\small,scale=0.8]
    
    \draw[thick] (0, 0) ellipse (7cm and 3.3cm);
    \node[anchor=north] at (0, 3.4) {\textbf{$p$-system}};
    
    \draw[thick] (0, -0.2) ellipse (6.5cm and 2.9cm);
    \node[anchor=north] at (4.5, 1) {\textbf{$p$-extendible}};

    \draw[thick] (-1.3, -0.25) ellipse (4.5cm and 2.5cm);
    \node[anchor=north] at (-1.3, 2.15) {\textbf{matroid $p$-parity}};

    \draw[thick] (-1.3, -0.45) ellipse (3.9cm and 1.9cm);
    \node[anchor=north] at (-1.3, 1.4) {\textbf{$p$-matchoid}};

    \draw[thick] (-1.3, -0.65) ellipse (3.4cm and 1.3cm);
    \node[anchor=north] at (-1.3, 0.4) {\textbf{matroid $p$-intersection}};

    \draw[thick] (-1.3, -1.05) ellipse (2cm and 0.7cm);
    \node at (-1.3, -1) {\textbf{matroid}};

    \draw[fill=black, fill opacity=0.2, thick] (2.8, -1) ellipse (2.9 and 0.8);
    \node at (4.4, -1) {\textbf{$p$-exchange}};

\end{tikzpicture}

%% file: algorithms-matroid.tex
\section{Algorithms for a Matroid Constraint}
\label{sec:matroid}

In this section, we design single-pass and multi-pass semi-streaming algorithms for maximizing a monotone submodular function subject to a matroid constraint, proving the results stated by Theorems~\ref{thm:main2} and~\ref{thm:main} for this combination of objective function and constraint.
The multi-pass algorithm appears in Section~\ref{sssc:multi-pass_proof}, and the single-pass algorithm appears in Section~\ref{sssc:single-pass_proof}.

\subsection{\texorpdfstring{$(1-1/e-\eps)$}{1-1/e-ε}-approximation in \texorpdfstring{$O(\log(1/\eps))$}{O(log(1/ε))} Passes} \label{sssc:multi-pass_proof}

In this section, we prove Theorem~\ref{thm:main}, which we repeat here for convenience.

\thmMain*

\begin{algorithm}[t!]
\caption{\textsc{Multi-Pass Algorithm (Monotone Submodular/Matroid)}} \label{alg:multi-pass}
Let $\ell \gets \lceil \ln (3\eps^{-1}) \rceil$, and let $A_0 \gets \varnothing$.\\
\For{$i = 1$ \KwTo $\ell$}
{
    Execute Algorithm~\ref{alg:streaming2} of Proposition~\ref{prop:main} with $A = A_{i - 1}$, $h = e^{i - \ell - 1}$ and $\delta = \eps / 6$.\\
    Let $A_i$ be the output set of the algorithm of Proposition~\ref{prop:main}.
}
\Return $A_\ell$.
\end{algorithm}

The algorithm that we use to prove Theorem~\ref{thm:main} is formally stated as Algorithm~\ref{alg:multi-pass}. In essence, this algorithm executes the algorithm of Proposition~\ref{prop:main} roughly $\ln (\nicefrac{3}{\eps})$ times, and feeds the output of each execution as the input set $A$ for the next execution. Since each execution of the algorithm of Proposition~\ref{prop:main} requires a single pass and stores $O(\nicefrac{r}{\delta}) = O(\nicefrac{r}{\eps})$ elements, it is clear that Algorithm~\ref{alg:multi-pass} requires $\ell = O(\log (\nicefrac{1}{\eps}))$ passes and stores $O(\nicefrac{r}{\eps})$ elements. Thus, to complete the proof of Theorem~\ref{thm:main}, it only remains to show that the approximation ratio of Algorithm~\ref{alg:multi-pass} is at most $\frac{e}{e - 1} + \eps$.

Let us denote $h_i \triangleq e^{i - \ell}$. By plugging in $h=h_{i-1}$ and $A=A_{i-1}$ into the guarantee of Proposition~\ref{prop:main} and rearranging, one obtains
\[
 e^{h_i} \cdot \bE[F(h_i \cdot \characteristic_{A_{i}})] - e^{h_{i-1}} \cdot \bE[F(h_{i-1}\cdot \characteristic_{A_{i-1}})] \geq (e^{h_i}-e^{h_{i-1}})(1 - \delta) \cdot F(OPT) 
 \quad
 \forall\; i \in [\ell]
 \enspace.
\]
Combining this inequality for all $i \in [\ell]$ yields
\begin{align*} 
	  e \cdot \bE[f(A_\ell)]
		\geq{}&
    e^{h_\ell} \cdot \bE[f(A_\ell)] - e^{h_0} \cdot \bE[F(h_0\cdot \characteristic_{A_0})]\\
    ={}&\sum_{i=1}^\ell \{e^{h_i}\cdot \bE[F(h_i \cdot \characteristic_{A_{i}})]- e^{h_{i-1}} \cdot \bE[F(h_{i-1}\cdot \characteristic_{A_{i-1}})]\}\\
    \ge{}& \sum_{i=1}^\ell (e^{h_i}-e^{h_{i-1}})(1 - \delta) \cdot F(OPT) \\
    ={}&(e^{h_\ell}-e^{h_0})(1-\delta)\cdot F(OPT)\\
    ={}&(e-e^{e^{-\ell}})(1-\delta)\cdot F(OPT)\\
    \geq{}&(e-e^{\eps/3})(1-\eps/6)\cdot F(OPT)
		\enspace,
\end{align*}
where the first inequality holds by the non-negativity of $f$ and the observation that $h_\ell = 1$. Rearranging the last inequality now gives, for a small enough $\eps$,
\[
	\frac{f(OPT)}{\bE[f(A_\ell)]}
	\leq
	\frac{1}{(1 - e^{-\eps/3 - 1})(1 - \eps/6)}
	\leq
	\frac{1}{(1 - e^{-1} - \eps/6)(1 - \eps/6)}
	\leq
	\frac{1}{1 - e^{-1} - \eps/3}
	\leq
	\frac{e}{e - 1} + \eps
	\enspace,
\]
which proves the promised approximation ratio for Algorithm~\ref{alg:multi-pass} because  $A_\ell$ is the output set of this algorithm.

\subsection{Single-pass Algorithm}\label{sssc:single-pass_proof}

In this section, we prove the guarantees of Theorem~\ref{thm:main2} for maximizing a monotone submodular function subject to a matroid constraint. For coniviniece, the following theorem restates these guarantees.
\begin{theorem}\label{thm:main2_matroid_only}
Let $\cM=(\cN,\cI)$ be a matroid of rank $r$, and let $f\colon 2^\cN\to \nnR$ be a non-negative monotone submodular function. For any sufficiently small $\eps>0$, there are single-pass semi-streaming algorithms that under a uniformly random arrival order store $O\!\left(\nicefrac{r}{\eps^2}\log^2(r/\eps)\right)$ elements and achieve approximation ratios of 
\[
	\frac{1+\left(\nicefrac{e}{e-1}\right)e^{-2+1/e}}{e^{-2+1/e}+1-e^{-1+1/e}} + O(\eps) \leq 1.972
	\qquad
	\text{and}
	\qquad
	\frac{1+e^{-2+1/e}}{e^{-2+1/e}+1-e^{-1+1/e}} + O(\eps) \leq 1.801
\]
in polynomial and FPT time, respectively.
\end{theorem}

Recall that Nematollahi et al.~\cite{NVZ25} proved that one can obtain $2$-approximation (up to an error term) for this problem in FPT time, and this is captured also by Corollary~\ref{cor:offlinetoStream}. Furthermore, combining Corollary~\ref{cor:offlinetoStream} with the optimal offline $\frac{e}{e-1}$-approximation algorithm for maximizing a monotone submodular function under a matroid constraint (due to \cite{calinescu2011maximizing}) yields $(1+\frac{e}{e-1})\approx 2.582$-approximation in polynomial time. A better polynomial time result can be obtained by executing the algorithm of Proposition~\ref{prop:main} with $A=\varnothing$, which produces a set $S$ satisfying
\[
  \bE[f(S)] \geq
  (1 - e^{-1+1/e} - \delta)\cdot f(OPT) + (e^{-2+1/e} + \delta)\cdot f(A)
  \geq 0.468\cdot f(OPT)
	\enspace,
\]
and hence, gives $\nicefrac{1}{0.468}\approx 2.137$-approximation, which already improves over the state-of-the-art polynomial time $2.315$-approximation due to Shadravan~\cite{S20}. We manage to further improve the approximation ratios for both polynomial and FPT times by combining the two above techniques (recall that Proposition~\ref{prop:main2} is the technical result underlying Corollary~\ref{cor:offlinetoStream}). The algorithm obtained by this combination is formally stated as Algorithm~\ref{alg:single-pass-matroid}.

\begin{algorithm}[t!]
\caption{\textsc{Single-Pass Algorithm (Monotone submodular/Matroid)}}\label{alg:single-pass-matroid}
\DontPrintSemicolon
\KwIn{Matroid $(\cN,\cI)$, non-negative monotone submodular function $f$, parameters $\eps\in(0,1)$, rank $r\in Z_+$.}
\smallskip\hrule\smallskip
Run the algorithm of Proposition~\ref{prop:main2} with $\delta=\eps$, and let $S_{\eps}$ be its output after scanning the first $\lfloor \eps n \rfloor$ elements of the stream.\\
On the remaining $n - \lfloor \eps n \rfloor$ elements, continue running algorithm of Proposition~\ref{prop:main2}, and also in parallel execute the boosting algorithm of Proposition~\ref{prop:main} with $A=S_{\eps}$ and $\delta = \eps$.\\
Let $H$ be the second set returned by the algorithm of Proposition~\ref{prop:main2}, and let $T$ be the set returned by the algorithm of Proposition~\ref{prop:main}.\\
Define $A^*$ to be $\arg\max\{f(S): S\subseteq S_{\eps}\cup H,\ S\in\cI\}$.\\
{\bf FPT time:} \Return $\arg\max\{f(T),f(A^*)\}$.\\
{\bf Polynomial time:} Compute $A\in\cI$ such that $f(A)\geq (1-1/e - O(\eps))\cdot f(A^*)$---for example, via the algorithm of~\cite{calinescu2011maximizing}, and \Return $\arg\max\{f(T),f(A)\}$.
\end{algorithm}

Algorithm~\ref{alg:single-pass-matroid} clearly makes a single pass and stores
$O(\nicefrac{r}{\eps^2}\log^2(r/\eps)) + O(\nicefrac{r}{\eps})
= O(\nicefrac{r}{\eps^2}\log^2(r/\eps))$ elements.
Therefore, we concentrate on analyzing its approximation guarantees.
Notice that, when invoked by Algorithm~\ref{alg:single-pass-matroid}, the algorithm of Proposition~\ref{prop:main} processes only the elements of the stream arriving after the first $\lfloor \eps n \rfloor$ positions. Thus, if we denote by $OPT'$ the intersection between the optimal solution and these elements, then by Proposition~\ref{prop:main},
\begin{align}\label{ineq11}
    \bE[f(T)] \geq{} & (1-e^{-1+1/e}-O(\eps))\cdot \bE[f(OPT')] + e^{-2+1/e}\cdot \bE[f(S_{\eps})] \\\nonumber
    \geq{} & (1-e^{-1+1/e}-O(\eps))\cdot f(OPT) + e^{-2+1/e}\cdot \bE[f(S_{\eps})]
    \enspace,
\end{align}
where the second inequality holds because $OPT'$ contains each element of $OPT$ with probability $p = \frac{n - \lfloor \eps n \rfloor}{n} \geq 1 - \eps$, which by Lemma~\ref{lem:sampling}, guarantees that $\bE[f(OPT')]\geq p \cdot f(OPT) \geq (1 - \eps) \cdot f(OPT)$.

Next, suppose we compute $A\in\cI$ with $f(A)\geq f(A^*)/\alpha$. Then
{\allowdisplaybreaks\begin{align} \label{eq:approximation_alpha}
\bE[\max\{f(A),f(T)\}]
&\geq \frac{\alpha e^{-2+1/e}}{1+\alpha e^{-2+1/e}}\cdot \bE[f(A)]
     +\frac{1}{1+\alpha e^{-2+1/e}}\cdot \bE[f(T)]\\\nonumber
&\geq \frac{e^{-2+1/e}}{1+\alpha e^{-2+1/e}}\cdot \bE[f(A^*)]
     +\frac{1}{1+\alpha e^{-2+1/e}}\cdot \bE[f(T)]\\\nonumber
&\geq \frac{e^{-2+1/e}}{1+\alpha e^{-2+1/e}}
     \cdot\Big((1-O(\eps))\cdot f(OPT)-\bE[f(S_{\eps})]\Big)\\\nonumber
&\qquad
+\frac{1}{1+\alpha e^{-2+1/e}}
     \cdot\Big((1-e^{-1+1/e}-O(\eps))\cdot f(OPT)+e^{-2+1/e}\cdot \bE[f(S_{\eps})]\Big)\\\nonumber
&= \bigg(\frac{e^{-2+1/e}+1-e^{-1+1/e}}{1+\alpha e^{-2+1/e}}-O(\eps)\bigg)\cdot f(OPT)\enspace,
\end{align}}%
where the third inequality follows from Inequality~\eqref{ineq11} and the fact that Proposition~\ref{prop:main2} guarantees that
$\bE[f(A^*)]+\bE[f(S_{\eps})]\geq (1-O(\eps))\cdot f(OPT)$. The FPT version of Algorithm~\ref{alg:single-pass-matroid} outputs the better set among $T$ and $A^*$, which corresponds to $\alpha = 1$, while the polynomial-time version of Algorithm~\ref{alg:single-pass-matroid} has $\alpha=\frac{e}{e-1} + O(\eps)$. Plugging these values of $\alpha$ into Inequality~\eqref{eq:approximation_alpha} yields the approximation ratios stated in Theorem~\ref{thm:main2_matroid_only}.

%% file: single-pass-alg.tex
\section{Offline Algorithm to a Single-pass Semi-streaming Algorithm}\label{sec:single-pass}

Our goal in this section is to prove Proposition~\ref{prop:main2}, which we repeat here for convenience. 

\PropMainH*

The algorithm that we use to prove Proposition~\ref{prop:main} is given as Algorithm~\ref{alg:filter}. This algorithm gets as parameters the rank $r$ of the $p$-system and the error control parameter $\delta$ from the statement of the proposition. We assume, without loss of generality, that $\delta \leq 1/2$ and $\delta n/r$ is an integer. If the first assumption is violated, we simply replace $\delta$ with $1/2$. If the second assumption is violated, this can always be fixed in one of two ways. If $\delta n/r > 1$, then $\delta$ can be replaced with some value $\delta'\in[\delta/2, \delta]$ such that $\delta' n/r$ is an integer. Otherwise, if $\delta n/r < 1$, then $n\leq r/\delta$, and thus, one can obtain the guarantee of the proposition by outputting the empty set as $S_\delta$ and the entire stream as $H$. 

Algorithm~\ref{alg:filter} treats differently the first $\delta n$ elements of the random order stream. These elements are logically further partitioned into $r$ consecutive blocks of size $\delta n/r$ each. The pseudocode of Algorithm~\ref{alg:filter} refers to these blocks as $\cN_1, \cN_2, \dotsc, \cN_r$. In other words, $\cN_1$ is the set of the first $\delta n / r$ elements in the stream, $\cN_2$ is the set of the next $\delta n / r$ elements of the stream, and so on, until $\cN_r$. We denote the remaining $(1 - \delta)n$ elements of the stream by $\tilde{\cN}$, i.e., $\tilde{\cN} = \cN \setminus \cup_{j = 1}^r \cN_i$. When processing the first $\delta n$ elements, Algorithm~\ref{alg:filter} greedily constructs a set $S_r$. To do so, it starts by setting $S_0$ to be the empty set, and then when processing each block $\cN_j$, it sets $S_j$ to $S_{j - 1}$ plus at most one element from this block. More specifically, the algorithm finds the element $s_j$ that is most valuable with respect to $S_{j - 1}$ among the elements of $\cN_j$ that (i) can be added to $S_{j - 1}$ without violating feasibility and (ii) have a non-negative value with respect to $S_{j - 1}$. If any element of $\cN_j$ obeys these two conditions, then $S_j$ is set to $S_{j - 1} + s_j$. Otherwise, the element $s_j$ is not well-defined. In this case, we set $S_j$ to be $S_{j - 1}$, and we treat $f(s_j \mid S_{j-1})$ as zero in the rest of the algorithm and in its analysis. After inspecting the first $\delta n$ elements, Algorithm~\ref{alg:filter} outputs $S_r$ as $S_{\delta}$. 

Before inspecting any additional elements, Algorithm~\ref{alg:filter} defines a set of $k+1 \triangleq 1+\lceil2\delta^{-1}\ln (r/\delta)\rceil$ values $v_0, v_1, \dotsc, v_k$ and chooses a threshold value $w_j$ for each $j \in [r]$. The value $v_0$ is simply the largest marginal contribution of any element $s_j$, and the other values complement it to an exponentially decreasing series with a ratio of $1 + \delta$. Each threshold value $w_j$ is set to be the lowest value $v_i$ that is at least the marginal contribution of $s_j$. Once these values and thresholds are determined, Algorithm~\ref{alg:filter} starts scaning the remaining $(1-\delta)n$ elements of the stream, i.e., the elements of $\tilde{\cN}$. Each such element $u$ is added to the set $H$ if there exists $j \in [r]$ such that $u$ can be added to $S_{j-1}$ and its marginal contribution to $S_{j-1}$ is strictly above the threshold value $w_j$. If, at any point of time, the number of elements added to $H$ exceeds $4r\delta^{-2}\cdot \log^2(\nicefrac{r}{\delta})$, the algorithm declares failure and returns the empty set; otherwise, it outputs the set $H$.

\begin{algorithm}
\caption{\textsc{Filtering Streaming}} \label{alg:filter}
\DontPrintSemicolon
\KwIn{Independence system $(\cN, \cI)$, non-negative submodular function $f$, parameters $\delta \in(0,1), r\in Z_{\geq 1}$. 
}
\smallskip\hrule\smallskip
\nonl\textbf{Inspecting the first $\delta n$ elements of the stream:} \\
Let $S_0 \gets \varnothing$.\\
\For{$j = 1$ \KwTo $r$}
{
Let $s_j \gets \arg \max_{u\in \cN_j, S_{j-1}+u\in \cI, f(u \mid S_{j-1})\geq 0}f(u \mid S_{j-1})$.\label{line:choose} \tcp*{There may be no such element.}
\lIf{there is an element $s_j$}{Let $S_j \gets S_{j-1} + s_j$,}
\lElse{Let $S_j \gets S_{j - 1}$ and treat $f(s_j \mid S_{j - 1})$ as $0$.}
}
\Return $S_r$.  \tcp*{Returned after inspecting $\delta n$ elements.}
\smallskip\hrule\smallskip
\nonl\textbf{Inspecting the remaining $(1-\delta)n$ elements of the stream:}\\
Let $k \gets \lceil 2\delta^{-1} \ln (r/\delta)\rceil$, and set
\begin{align*}
 v_i& \gets\frac{\max_{j=1}^{r}f(s_j \mid S_{j-1})}{(1+\delta)^i} &&  \forall\; i \in \{0, 1, \ldots, k\} \enspace, \quad\text{and} \\ 
 w_j & \gets \min\{v_i \mid v_i\geq f(s_j \mid S_{j-1})\} && \forall\; j \in \{1, 2, \ldots, r\} \enspace.
\end{align*}\\
Let $H \gets \cup_{j=1}^{r}\{u\in \tilde{\cN}, S_{j-1}+u\in \cI, f(u \mid S_{j-1})> w_j\}$. \tcp*{Note the strict inequality.}
\lIf{$|H|\leq 4r\delta^{-2} \ln^2 (\nicefrac{r}{\delta})$}{\Return $H$, }\lElse{\Return $\varnothing$ and indicate failure.}
\end{algorithm}

We begin the analysis of Algorithm~\ref{alg:filter} with the following observation. This observation holds since, as explained above, when constructing the set $H$, if this set becomes too large, then the algorithm can immediately abort this construction, declare failure and return the empty set.

\begin{observation}
Algorithm~\ref{alg:filter} can be implemented in the data stream model using $O(r\delta^{-2}\cdot \log^2(\nicefrac{r}{\delta}))$ memory.
\end{observation}

Our next goal is to show that Algorithm~\ref{alg:filter} is unlikely to declare failure, which is shown by Lemma~\ref{lem:failure_prob} below. However, before proving this lemma, we need the next observation, which follows by the algorithm's choice of the values $v_i$ and thresholds $w_j$. We note that the proof of this observation works even when $f$ is a non-negative function that is not necessarily submodular.

\begin{observation} \label{obs:w_bounds}
For any $j \in [r]$, 
\[
  f(s_j \mid S_{j-1}) \leq w_j  \leq \max\bigg\{(1+\delta)f(s_j \mid S_{j-1}), \frac{\delta}{r} \cdot f(S_r)\bigg\}
	\enspace.
\]
Furthermore, $v_k \leq \frac{\delta}{r}\cdot f(S_r)$.
\end{observation}

\begin{proof}
Since $v_0=\max_{j=1}^{r}f(s_j \mid S_{j-1})$, for every $j \in [r]$ there exists some $v_i\geq f(s_j \mid S_{j-1})$. Hence, by definition, $w_j =\min\{v_i \mid v_i\geq f(s_j \mid S_{j-1})\} \geq f(s_j \mid S_{j-1})$.
Next, observe that 
\[
    v_k = \frac{\max_{j=1}^{r}f(s_j \mid S_{j-1})}{(1+\delta)^{\lceil2\delta^{-1}\ln (r/\delta)\rceil}} \leq \frac{\delta}{r}\cdot \max\nolimits_{j=1}^{r}f(s_j \mid S_{j-1}) \leq \frac{\delta}{r}\cdot f(S_r) \enspace,
\]
where the second inequality follows from the non-negativity of $f$ because $f(s_j \mid S_{j-1}) \geq 0$ for all $j \in [r]$.
 If $w_j=v_k$, then 
$w_j=v_k \leq \frac{\delta}{r}\cdot f(S_r)$.
Otherwise, there must be some integer $0 \leq i < k$ such that $w_j=v_i$ and $v_{i+1}< f(s_j \mid S_{j-1})$, which implies
$w_j = v_i= (1+\delta) v_{i+1} < (1+\delta) f(s_j \mid S_{j-1})$.
\end{proof}

\begin{lemma} \label{lem:failure_prob}
With probability at least $1-\delta$, for every index $j \in [r]$,
\begin{equation} \label{const-main}
    \big| \{u\in \tilde{\cN}, S_{j-1}+u\in \cI, f(u \mid S_{j-1}) > w_j\}\big|\leq \frac{r}{\delta}\cdot \ln(\nicefrac{r}{\delta})
		\enspace.
\end{equation}
Furthermore, if this event happens, then $|H| \leq 4r\delta^{-2} \ln^2 (\nicefrac{r}{\delta})$, which guarantees that Algorithm~\ref{alg:filter} does not declare failure. When $f$ is non-negative but not submodular, the above event still implies the weaker inequality $|H| \leq r^2\delta^{-1} \cdot \ln(\nicefrac{r}{ \delta})$. 
\end{lemma}

\begin{proof}
       Fix the randomness of the sets $\cN_1, \ldots, \cN_{j-1}$ (that is inherited from the uniformly random order of the stream), and let $T_j \triangleq \{u\in \tilde{\cN} \cup_{i= j}^{r} \cN_i, S_{j-1}+u\in \cI\}$ be the set of elements that can be added to $S_{j-1}$ and do not belong to the first $j - 1$ blocks. Let $\ell\triangleq \lceil r\delta^{-1}\cdot \ln(\nicefrac{r}{\delta})\rceil$. If $|T_j| < \ell$, then Inequality~\eqref{const-main} is always satisfied for $j$. Otherwise, define $m\triangleq|T_j|\geq \ell$, and let us denote by $v_1, v_2, \dotsc, v_m$ the marginal contributions of the elements of $T_j$ with respect to $S_{j - 1}$ in a non-increasing contribution order.\footnote{Note that there might be multiple non-increasing orders if there are ties between the marginal contributions of different elements of $T_j$. When this happens, we arbitrarily fix one such order.} We are interested in the set $T'_j= \{u\in T_j, f(u \mid S_{j-1})\geq v_{\ell}\}$. Notice that the size of this set is always at least $\ell$, but it can be strictly larger when there are multiple elements in $T_j$ whose marginal contribution to $S_{j - 1}$ is $v_\ell$.
			
       Recall that we fixed the randomness of $\cN_1, \ldots, \cN_{j-1}$, and hence, the set $\cN_j$ is a uniformly random subset of size $\delta n/r$ elements out of the set $\tilde{\cN} \cup_{i= j}^{r} \cN_i$. Since the last set contains $(1-\delta) n+ (r-j+1) \cdot \delta n/r$ elements,
\begin{align*}
	\Pr[T'_j \cap \cN_{j}= \varnothing]
	={}	&
	\prod_{i=1}^{|T'_j|}\left(1- \frac{(\delta n)/r}{\left[(1-\delta)n + (r-j+1) \cdot \delta n/r \right]- i+1}\right)\\
	\leq{} &
	\bigg(1-\frac{\delta}{r}\bigg)^{\!|T'_j|}
	\leq
	\bigg(1-\frac{\delta}{r}\bigg)^{\!\ell}
	\leq
	\bigg(1-\frac{\delta}{r}\bigg)^{\!r\delta^{-1}\cdot \ln(\nicefrac{r}{\delta})}
	\leq
	\frac{\delta}{r}
	\enspace.
\end{align*}
Notice also that if $T'_j \cap \cN_{j} \neq \varnothing$, then Algorithm~\ref{alg:filter} picks an element $s_j$ such that $f(s_j \mid S_{j-1})\geq v_{\ell}$, and thus, 
by Observation~\ref{obs:w_bounds}, $w_j \geq f(s_j \mid S_{j-1})\geq v_{\ell}$, which implies
\begin{align*}
| \{u\in \tilde{\cN}, S_{j-1}+u\in \cI, f(u \mid S_{j-1}) > w_j\}|
\leq{} &
|\{u\in \tilde{\cN}, S_{j-1}+u\in \cI, f(u \mid S_{j-1}) > v_{\ell}\}|\\
\leq{} &
|\{v_1, v_2, \dotsc, v_{\ell - 1}\}|
=
\ell - 1
<
\frac{r}{\delta}\cdot \log(\nicefrac{r}{\delta})
\enspace.
\end{align*}
Therefore, we have proved that, after fixing $\cN_1, \ldots, \cN_{j-1}$, the probability that Inequality~\eqref{const-main} holds for $j$ is at least $\Pr[T'_j \cap \cN_{j} \neq \varnothing] \geq 1 - \delta/r$. Since this is true regardless of the values fixed for $\cN_1, \ldots, \cN_{j-1}$, by the law of total probability, this holds also when these sets are not fixed. Taking a union bound over all $j \in [r]$, we get that Inequality~\eqref{const-main} holds for all these $j$ values (at the same time) with probability at least $1 - \delta$.

Next, let us prove the second part of the lemma. In other words, we assume that Inequality~\eqref{const-main} holds for all $j \in [r]$, and our goal is to upper bound the size of $H$. The bound $|H|\leq r^2\delta^{-1}\cdot \log(\nicefrac{r}{\delta})$ follows immediately from Inequality~\eqref{const-main} since every element of $H$ is contained in the set on the left hand side of this inequality for some $j \in [r]$. When $f$ is submodular, we can get a stronger bound on the size of $H$ by observing that for every two integers $1 \leq j \leq j' \leq r$, if $w_j = w_{j'}$ then
\[
	\{u\in \tilde{\cN}, S_{j-1}+u\in \cI, f(u \mid S_{j-1}) > w_j\}
	\supseteq
	\{u\in \tilde{\cN}, S_{j'-1}+u\in \cI, f(u \mid S_{j'-1}) > w_{j'}\}
\]
since $S_{j - 1} \subseteq S_{j' - 1}$. This implies that for every integer $0 \leq i \leq k$,
\begin{multline*}
	|\cup_{j : w_j=v_i}\{u\in \tilde{\cN}, S_{j-1}+u\in \cI, f(u \mid S_{j-1}) > w_j\}|\\
	=
	|\{u\in \tilde{\cN}, S_{j^*-1}+u\in \cI, f(u \mid S_{j^*-1}) > w_{j^*}\}|
	\leq
	\frac{r}{\delta}\cdot \ln(\nicefrac{r}{\delta})
	\enspace,
\end{multline*}
where $j^* \triangleq \min\{j \mid w_j=v_i\}$. Every element of $H$ is contained in the set on the leftmost side of the last inequality for some $i$, and since there are only $k + 1$ different values for $i$,
\[
	|H|
	\leq
	(k+1) \cdot \frac{r}{\delta}\cdot \ln(\nicefrac{r}{\delta})= (\lceil2\delta^{-1}\ln (r/\delta)\rceil + 1) \cdot \frac{r}{\delta}\cdot \ln(\nicefrac{r}{\delta}) \leq \frac{4r}{\delta^2}\cdot \ln^2(\nicefrac{r}{\delta})
	\enspace.
	\qedhere
\]
%
\end{proof}

The following lemma considers any arrival order of the elements that does not cause Algorithm~\ref{alg:filter} to declare failure. The lemma proves that under such an order the algorithm achieves a good value compared with the elements of $OPT$ that appear as part of the final $(1-\delta)n$ elements of the stream. Since the proof of this lemma is quite complex, it is deferred to Section~\ref{ssc:deterministic_guarantee}.

\begin{restatable}{lemma}{LemDeterministicLemma} \label{lem:deterministic_lemma}
Recall that $A^* \triangleq \arg \max_{A\subseteq(S_{r}\cup H), A \in\cI}f(A)$. Consider any order of the elements in which  Algorithm~\ref{alg:filter} does not fail, and let $O'=OPT \cap \tilde{\cN}$.
If $(\cN, \cI)$ is a $p$-system, then
\[
    (p + 1 + O(\delta p)) \cdot f(S_r) + f(A^*) \geq f(O' \cup S_r) \enspace.
\]
    If $(\cN, \cI)$ is a $p$-exchange system or a matroid $p$-parity, then the above inequality improves to
\[
    (p + O(\delta p)) \cdot f(S_r) + f(A^*) \geq f(O' \cup S_r) \enspace.
\]
When $f$ is linear, the above inequalities still hold if we decrease the coefficient of $f(S_r)$ on the left hand side by $1$ and replace the right hand side with $f(O')$.
\end{restatable}

We are now ready to complete the proof of Proposition \ref{prop:main2}.

\begin{proof}[Proof of Proposition~\ref{prop:main2}]
Let $E$ be the event that the algorithm fails, and $\bar{E}$ be the complementing event. Then, by the non-negativity of $f$ and the law of total expectation,
\begin{align} \label{eq:total_probability}
    (1+p + O(p\delta)) \cdot \bE[f(S_r)] + \bE[f(A^*)] & \geq \bE[(1+p + O(p\delta)) \cdot f(S_r) + f(A^*) \mid \bar{E}]\cdot \Pr[\bar{E}]\\\nonumber  
    & \geq \bE[f(O' \cup S_r) \mid \bar{E}]\cdot \Pr[\bar{E}]\\\nonumber
    & = \bE[f(O' \cup S_r) ] - \bE[f(O' \cup S_r)]\cdot \Pr[E]
		\enspace,
\end{align}
where the second inequality follows by Lemma~\ref{lem:deterministic_lemma}. Let us now bound the two terms on the rightmost side of the above inequality. First, as both $O'$ and $S_r$ are feasible subsets, we get by the non-negativity and submodularity of $f$ that it deterministically holds that $f(O'\cup S_r) \leq f(O')+f(S_r) \leq 2 f(OPT)$. Second, if we denote by $\vy\in [0,1]^\cN$ a vector whose value $y_u$ for every element $u \in \cN$ is the probability that this element belongs to $S_r \cup O'$, then by Lemma~\ref{lem:lovasz_lower_bound},
\begin{align*}
	\bE[f(O' \cup S_r)] \geq{} & \hat{f}(\vy) = \int_{t = 0}^{1} f(T_\lambda(\vy)) \cdot dt\\
	\geq{}& \int_{t = \delta}^{1-\delta} f(T_\lambda(\vy)) \cdot dt = \int_{t=\delta}^{1-\delta}f(OPT) \cdot dt = (1-2\delta)\cdot f(OPT) \enspace,
\end{align*} 
where the second inequality uses the non-negativity of $f$, and the penultimate equality holds since every element $o \in OPT$ belongs to $O'$ whenever it is not one of the first $\delta n$ elements, which happens with probability $1 - \delta$, and every element $u \in \cN \setminus OPT$ can belong to $S_r$ only when it is one of the first $\delta n$ elements, which happens with probability $\delta$. 

Plugging the above bounds into Inequality~\eqref{eq:total_probability} and rearranging yields
\begin{align*}
	(1+p)\cdot \bE[f(S_r)] + \bE[f(A^*)]
	\geq{} &
	(1-2\delta) \cdot f(OPT) - 2\delta \cdot f(OPT) - O(p\delta) \cdot \bE[f(S_r)]\\
	={} &
	(1-4\delta) \cdot f(OPT) - O(p\delta) \cdot \bE[f(S_r)]
	\geq
	(1-O(\delta p))\cdot f(OPT)
	\enspace,
\end{align*}
where the second inequality holds since the inequality $f(S_r) \leq f(OPT)$ holds deterministically because $S_r$ is a feasible solution.

This complete the proof of the general case of the proposition. To get the improved guarantees of the proposition for special cases, one needs to use the guarantee of Lemma~\ref{lem:deterministic_lemma} for these special cases instead of the guarantee of Lemma~\ref{lem:deterministic_lemma} for the general case, and then adapt the rest of the proof in the natural way.
\end{proof}

Before concluding this section, let us state the following observation, which shows that our analysis of Algorithm~\ref{alg:filter} is close to tight. A simplified partition matroid is a matroid $(\cN, \cI)$ defined by a partition of the ground set $\cN$ into $r$ sets $A_1, A_2, \dotsc, A_r$. A set $S$ is independent in this matroid if $|S \cap A_i| \leq 1$ for every $i \in [r]$.

\begin{observation}
    There is an an instance in which $f$ is a non-negative monotone submodular function and $\cM$ is the intersection of $p$ simplified partition matroids such that it deterministicly holds that the output sets $S_{\delta}=S_r$ and $H$ of Algorithm~\ref{alg:filter} satisfy 
    \[p  \cdot f(S_\delta) + f(A^*) \leq f(OPT)
		\enspace,
		\]
		where $A^* \triangleq \arg \max_{A\subseteq(S_{\delta}\cup H), A \in\cI}f(A)$, and $f(A^*)=f(S_{\delta}) = \frac{1}{p+1} \cdot f(OPT)$.
\end{observation}

\begin{proof}
The ground set is $\cN = O\cup T$, where $O=o_1, o_2, \dotsc, o_{p + 1}$ are the elements in the optimal solution and $T$ are additional elements.
The objective function is given by
\[
	f(S)
	=
	\min\{1, |S \cap (T+o_{p+1})|\} + |S \cap\{o_1, \ldots, o_p\}|
	\qquad
	\forall\; S \subseteq \cN
	\enspace.
\]
One can verify that this is a non-negative monotone submodular function (and even a coverage function).
Next, we define $p$ simplified partition matroids $(\cN, \cI_1), (\cN, \cI_2), \dotsc, (\cN, \cI_p)$. The partition defining each matroid $(\cN, \cI_j)$ has $p+1$ parts 
$A^{j}_1, \ldots, A^{j}_{p+1}$ given by
\[A^j_{i} = \left\{\begin{array}{ll} \{o_i\} & \text{if $i\neq j$} \enspace,\\
\{o_{j}\} \cup T & \text{if $i = j$} \enspace.
\end{array}\right.\]
The independence system $\cM$ is the intersection of these $p$ simplified partition matroids, i.e., $\cM = (\cN,\cap_{j=1}^{p}\cI_j)$.
Clearly, $O\in \cap_{j=1}^{p}\cI_j$ and $f(O)=p+1$.
In addition, for any $u\in T$, we have $\{u, o_{p+1}\}\in \cap_{j=1}^{p}\cI_j$, but for any $j \in [p]$, $\{u,o_j\}\not \in I_j$, and hence, $\{u,o_j\}\not \in\cap_{j=1}^{p}\cI_j$.

We note that these function $f$ and independence system $\cM$ form a canonical example for showing that the greedy algorithm and the standard local search algorithm both have approximation ratios no better than $p+1$. To see that this example indeed proves this, notice that if some element $u\in T$ is chosen first to the solution (which can be guaranteed by breaking ties appropriately), then only $o_{p+1}$ can be added after this point to the solution, which leads to a greedy solution of value $f(\{o_{p+1}, u\})=f(\{u\})=1$. Furthermore, no swap of a single element can improve the solution $\{u\}$, and thus, it becomes the output of the standard local search algorithm.

When the set $T$ is large enough, the function $f$ and the independence system $\cM$ form a bad example also for Algorithm~\ref{alg:filter} for similar reasons. Specifically, if we choose the set $T$ to be large enough so that $|T| > (p + 1)^2/\delta$, then
\[
	\frac{\delta n}{p}
	\geq
	\frac{\delta |T|}{p}
	>
	p + 1
	=
	|O|
	\enspace,
\]
which implies that there must always be at least one element of $T$ among the first $\frac{\delta n}{p}$ elements of the stream. Given appropriate tie-breaking, this guarantees that Algorithm~\ref{alg:filter} selects some element $u \in T$ as the first element of $S_{\delta}$, and the only element that may be added to $S_{\delta}$ after $u$ is $o_{p+1}$. Additionally, the set $H$ cannot contain any element $v \in O\cup T\setminus \{o_{p+1}, u\}$ because for every such element it holds that $f(v \mid \varnothing) = f(u \mid \varnothing)$ and also $\{u,v\} \not\in \cI$. Combining these observations, we get that both $S_\delta$ and $A^*$ must be either $\{u\}$ or $\{u, o_{p + 1}\}$, and thus, have a value of $1$.
\end{proof}

\subsection{Proof of Lemma~\texorpdfstring{\ref{lem:deterministic_lemma}}{\ref*{lem:deterministic_lemma}}} \label{ssc:deterministic_guarantee}

In this section, we prove Lemma~\ref{lem:deterministic_lemma}, which we repeat here for convinience.

\LemDeterministicLemma*

Let $O_H = O \cap H$, $O_L = O \setminus H$ and $O'_L = O_L \cap \tilde{\cN}$ (notice that since $H \subseteq \tilde{\cN}$, there is no need to define also an analogous set $O'_H$). Let $O''_L$ denote the set of elements in $O'_L$ that can be added to $S_r$ without violating independence. In other words, $O''_L = \{u \in O'_L \mid S_r + u \in \cI\}$.

\begin{observation} \label{obs:O_double_prime}
For every $u \in O''_L$, $f(u \mid S_r) \leq v_k$.
\end{observation}
\begin{proof}
If $O''_L$ is empty, this is trivial. Thus, assume that $O''_L$ is non-empty. This assumption implies that $S_r$ is not a base, and thus, has a size smaller than $r$. Hence, there exists an index $i \in [r]$ such that $s_i$ is not a real element, and for every element $u \in O''_L$, the fact that $u \not \in H$ despite $S_{i-1} + u \subseteq S_r + u \in \cI$ implies that
\[
	f(u \mid S_r)
	\leq
	f(u \mid S_{i-1})
	\leq
	w_{i}
	=
	v_k
	\enspace.
	\qedhere
\]
\end{proof}

We now need to state the following well-known property of $p$-systems. This property is implied, for example, in~\cite{calinescu2011maximizing}. However, since we are not aware of any place where this property is explicitly stated, we reprove it.
\begin{lemma} \label{lem:p-system_property}
Let $S = \{s_1, s_2, \dotsc, s_t\}$ and $T$ be two independent sets in a $p$-system $(\cN, \cI)$ such that $S$ is a base of $S \cup T$. Then, there exists a partition $T_1, T_2, \dotsc, T_t$ of $T$ such that
\[
    \{s_1, s_2, \dotsc, s_{i - 1}, u\} \in \cI
    \qquad
    \forall\; i \in [t] \text{ and } u \in T_i
    \enspace,
\]
and each set $T_i$ in the partition is of size at most $p$.
\end{lemma}
\begin{proof}
For every $i \in [t]$, we define the set $T_i$ via a downward recursion on $i$. To define $T_i$, assume that $T_{j}$ was already defined for every $i < j \leq t$, and let $B_i \triangleq \{u \in T \setminus \cup_{j = i + 1}^t T_j \mid \{s_1, s_2, \dotsc, s_{i - 1}, u\} \in \cI\}$. If $|B_i| \leq p$, then we define $T_i = B_i$. Otherwise, we set $T_i$ to be an arbitrary subset of $B_i$ of size $p$.

By definition, $T_i$ is of size at most $p$, and furthermore, $\{s_1, s_2, \dotsc, s_{i - 1}, u\} \in \cI$ for every $u \in T_i \subseteq B_i$. Additionally, for every two indexes $1 \leq i < j \leq t$, the sets $T_i$ and $T_{j}$ are disjoint because $T_i \subseteq B_i \subseteq T \setminus T_{j}$. Therefore, to complete the proof of the lemma, it only remains to show that the union of the sets $T_1, T_2, \dotsc, T_t$ is the entire set $T$. We do that by showing via downward induction that for every integer $i \in [t + 1]$, $|T \setminus \cup_{j = i}^t T_j| \leq p(i - 1)$. Notice that, by setting $i = 1$, this inequality indeed proves that the union of the sets $T_1, T_2, \dotsc, T_t$ is the entire set $T$.

To prove the base of the induction, observe that since $S$ is a base of $S \cup T$ and $T$ is an independent subset of this set, it must hold that $p \cdot |S| \geq |T|$. Thus,
\[
    |T \setminus \cup_{j = t + 1}^t T_i|
    =
    |T|
    \leq
    p \cdot |S|
    =
    p \cdot ((t + 1) - 1)
    \enspace.
\]
Assume now that the claim we want to prove by induction holds for $i + 1$ for some $i \in [t]$, and let us prove it for $i$. There are two cases to consider. If $|B_i| \geq p$, then $|T_i| = p$, and thus,
\[
    |T \setminus \cup_{j = i}^t T_i|
    =
    |T \setminus \cup_{j = i + 1}^t T_{i}| - |T_i|
    \leq
    pi - p
    =
    p(i - 1)
    \enspace,
\]
where the first equality uses the disjointness of the sets $T_1, T_2, \dotsc, T_t$, and the inequality employs the induction hypothesis. Consider now the case that $B_i < p$. In this case, $T_i = B_i$, which implies that no element $u \in T \setminus \cup_{j = i}^t T_i$ can be added to $s_1, s_2, \dotsc, s_{i - 1}$. In other words, $\{s_1, s_2, \dotsc, s_{i - 1}\}$ is a base of the set $E_i = \{s_1, s_2, \dotsc, s_{i - 1}\} \cup (T \setminus \cup_{j = i}^t T_i)$. Since $T \setminus \cup_{j = i}^t T_i \subseteq T$ is an independent subset of $E_i$, this implies
\[
    p(i - 1)
    =
    p \cdot |\{s_1, s_2, \dotsc, s_{i - 1}\}|
    \geq
    |T \setminus \cup_{j = i}^t T_i|
    \enspace.
    \qedhere
\]
\end{proof}

We now have all the tools necessary for proving Lemma~\ref{lem:deterministic_lemma} for $p$-systems.
\begin{proof}[Proof of Lemma~\ref{lem:deterministic_lemma} for $p$-systems]
By the definition of $O''_L$, $S_r$ is a base of $S_r \cup (O'_L \setminus O''_L)$, and thus, by Lemma~\ref{lem:p-system_property}, there exist a partition $O_1, O_2, \dotsc, O_r$ of $O'_L \setminus O''_L$ such that $S_{i - 1} + u \in \cI$ for every $i \in [r]$ and $u \in O_i$, and each set $O_i$ in the partition is of size at most $p$. Using this partition and the submodularity and non-negativity of $f$, we get
{\allowdisplaybreaks\begin{align*}
    (1+2\delta) \cdot f(S_r)
    &\geq (1+\delta) \sum_{j\in [r]}f(s_j \mid S_{j - 1}) +\delta \cdot f(S_r)\\
    & = \sum_{j\in [r]}\left((1+\delta)\cdot f(s_j \mid S_{j - 1}) + \frac{\delta}{r} \cdot f(S_r)\right) \\
    & \geq \sum_{j \in [r]}w_j \\
    & \geq p^{-1} \cdot \sum_{j \in [r]} \sum_{o \in O_j} w_j \\
    & \geq p^{-1} \cdot \sum_{j \in [r]}\sum_{o \in O_j} f(o \mid S_{j - 1})\\
    & \geq
    p^{-1} \cdot \sum_{o \in O'_L \setminus O''_L} f(o \mid S_r)\\
		& =
		p^{-1} \cdot \sum_{o \in O'_L} f(o \mid S_r) - p^{-1} \cdot \sum_{o \in O''_L} f(o \mid S_r)\\
    &\geq p^{-1} \cdot f(O'_L \mid S_r) - \frac{r}{p} \cdot v_k\\
		&\geq p^{-1} \cdot f(O'_L \mid S_r) - \frac{\delta}{p} \cdot f(S_r)
    \enspace,
\end{align*}}%
where the second inequality holds by Observation~\ref{obs:w_bounds}; the fourth inequality holds since any element $o \in O_j$ can be added to $S_{j - 1}$, and thus, the fact that it does not belong to $H$ implies that its marginal contribution with respect to $S_{j - 1}$ is at most $w_j$; the penultimate inequality holds by Observation~\ref{obs:O_double_prime} and the fact that $O_L''$, as a subset of the feasible set $O$, is of size at most $r$; and the last inequality follows from the second part Observation~\ref{obs:w_bounds}.

Using the last inequality and the fact that $f(A^*) \geq f(O_H)$ because $O_H = O \cap H$ is a feasible solution (recall that $A^*$ is an independent subset of $S_r \cup H$ maximizing $f$), we now get
\begin{align*}
		[p(1 + 2\delta) + 1 + \delta] \cdot f(S_r) + f(A^*)
    &\geq f(S_r) + f(O'_L \mid S_r) + f(O_H)  \\
    & = f(S_r \cup O'_L) +   f(O_H) \\
    & \geq f(S_r \cup O')
    \enspace,
\end{align*}
where the second inequality uses the submodularity and non-negativity of $f$. When $f$ is linear, $f(O'_L \mid S_r) = f(O'_L)$ since $O'_L \subseteq \tilde{\cN}$ is disjoint from $S_r$. This allows us to get also the inequality
\[
		[p(1 + 2\delta) + \delta] \cdot f(S_r) + f(A^*)
    \geq f(O'_L \mid S_r) + f(O_H)
		= f(O'_L) +   f(O_H) \\
    = f(O')
    \enspace.
		\qedhere
\]
\end{proof}

Lemma~\ref{lem:deterministic_lemma} gives an improved guarantee for $p$-exchange system and matroid $p$-parity constraints. To prove this improved guarantee, we show that these two kinds of systems have the following property.
\begin{definition} \label{def:bilevel}
We say that an independence system $(\cN, \cI)$ is $p$-\emph{bilevel-swappable} if for every two disjoint independent sets $S$ and $T$ and a subset $T_H$ of $T$, there exists a collection $\cC = \{Y_t \subseteq S \mid t \in T \setminus T_H\}$ obeying the three following properties.
\begin{enumerate}[1]\renewcommand{\theenumi}{(\alph{enumi})}
	\item $S \setminus Y_t + t \in \cI$ for every $t \in T \setminus T_H$,\label{item:single_exchange}
	\item every element of $s \in S$ appears in at most $p$ sets in the collection $\cC$, and\label{item:frequency}
	\item $S_L \cup T_H \in \cI$, where $S_L \triangleq \{s \in S : |\{Y_t \in \cC \mid s \in Y_t\}| = p\}$.\label{item:large_exchange}
\end{enumerate}
\end{definition}

Before proving that $p$-exchange systems and matroid $p$-parity are indeed $p$-bilevel-swappable, let us show that this property is strong enough to yield the guarantee of Lemma~\ref{lem:deterministic_lemma} for these kinds of systems.
\begin{lemma} \label{lem:bilevel-swappable_enough}
If $(\cN, \cI)$ is a $p$-bilevel-swappable, then
\[
    (p + O(p\delta)) \cdot f(S_r) + f(A^*) \geq f(O' \cup S_r) \enspace.
\]
When $f$ is linear, we also have
\[
	(p - 1 + O(p\delta)) \cdot f(S_r) + f(A^*) \geq f(O') \enspace.
\]
\end{lemma}
\begin{proof}
Let $\{Y_o \subseteq S_r \mid o \in O'_L \setminus O''_L\}$ be the collection $\cC$ whose existence is guaranteed by the definition of $p$-bilevel-swappable systems for $S = S_r$, $T = O'$ and $T_H = O_H \cup O''_L$ (notice that $S_r$ and $O' \subseteq \tilde{\cN}$ are indeed disjoint). For every element $o \in O'_L \setminus O''_L$, let us denote by $i(o)$ the minimal index in $[r]$ such that $s_{i(o)} \in Y_o$. To see that $i(o)$ is well-defined, notice that $S_r + o \not \in \cI$ by the definition of $O''_L$, but $S_r \setminus Y_o + o \in \cI$, and thus, the set $Y_o$ must be non-empty. By the non-negativity and submodularity of $f$, we now get
\begin{align*}
		f(S_L) + (p - 1 + 2p\delta) \cdot f(S_r)
		&\geq p\cdot \sum_{s_j\in S_L} \Big((1+\delta) \cdot f(s_j \mid S_{j - 1}) + \frac{\delta}{r}\cdot f(S_r)\Big) \\&\mspace{120mu}+ (p - 1)\cdot \sum_{s_j\in S_r \setminus S_L} \mspace{-9mu} \Big((1+\delta) \cdot f(s_j \mid S_{j - 1}) +\frac{\delta}{r}\cdot f(S_r)\Big)\\
    & \geq p\cdot \sum_{s_j\in S_L} w_j + (p - 1)\cdot \sum_{s_j\in S_r \setminus S_L} \mspace{-9mu} w_j \\
		& \geq \sum_{j \in [r]} |\{o \in O'_L \setminus O''_L \mid i(o) = j\}| \cdot w_j\\
    & \geq \sum_{j \in [r]} \sum_{\substack{o \in O'_L \setminus O''_L \\ i(o) = j}} \mspace{-9mu} f(o \mid S_{j-1})\\
    & \geq
    \sum_{o \in O'_L \setminus O''_L} f(o \mid S_r)\\
		& =
		\sum_{o \in O'_L} f(o \mid S_r) - \sum_{o \in O''_L} f(o \mid S_r)\\
    &\geq f(O'_L \mid S_r) - r \cdot v_k\\
		&\geq f(O'_L \mid S_r) - \delta \cdot f(S_r)
    \enspace,
\end{align*}
where the first inequality uses the fact that only elements with a non-negative marginal contributions are chosen on Line~\ref{line:choose} of Algorithm~\ref{alg:filter}; the second inequality holds by Observation~\ref{obs:w_bounds}; the third inequality holds because $|\{o \in O'_L \setminus O''_L \mid i(o) = j\}| \leq |\{Y_o \in \cC \mid s_j \in Y_o\}|$, and the right hand side of this inequality is $p$ for $s_j \in S_L$ and at most $p - 1$ for $s_j \in S_r \setminus S_L$; the fourth inequality holds since for every element $o \in O'_L \setminus O''_L$, the definition of $i(o)$ guarantees that $\{s_1, s_2, \dotsc, s_{i(o) - 1}\} + o \subseteq S_r \setminus Y_o + o \in \cI$, and thus, the fact that $o \not \in H$ implies that $f(o \mid s_1, s_2, \dotsc, s_{i(o) - 1}) \leq w_{i(o)}$; the penultimate inequality holds by Observation~\ref{obs:O_double_prime}; and the last inequality follows from the second part of Observation~\ref{obs:w_bounds}.

Using the last inequality and the fact that $f(A^*) \geq f(S_L \cup O_H)$ because $S_L \cup O_H \subseteq S_L \cup O_H \cup O''_L$ is a feasible solution (recall that we have chosen $T_H = O_H \cup O''_L$ when we invoked the definition of a $p$-bilevel-swapable system), we now get
\begin{align*}
		(p + 2p\delta + \delta) \cdot f(S_r) + f(A^*)
    &\geq f(S_r) + f(O'_L \mid S_r) + f(O_H \mid S_L)  \\
    & \geq f(S_r \cup O'_L) + f(O_H \mid S_r \cup O'_L) \\
    & = f(S_r \cup O')
    \enspace,
\end{align*}
where the second inequality uses the submodularity of $f$. When $f$ is a linear function, since $O'_L, O_H \subseteq \tilde{\cN}$ are both disjoint from $S_L$ and $S_r$, we have $f(O'_L \mid S_r) = f(O'_L)$ and $f(O_H \mid S_L) = f(O_H)$, which allows us to get
\[
		(p - 1 + 2p\delta + \delta) \cdot f(S_r) + f(A^*)
    \geq f(O'_L \mid S_r) + f(O_H \mid S_L)
		= f(O'_L) + f(O_H)
    = f(O')
    \enspace.
		\qedhere
\]
\end{proof}

Given Lemma~\ref{lem:bilevel-swappable_enough}, to complete the proof of Lemma~\ref{lem:deterministic_lemma}, it only remains to show the both $p$-extendible systems and matroid $p$-parity are $p$-bilevel-swappable systems. Lemmata~\ref{lem:extendible_swappable} and~\ref{lem:matroid_partiy_swappable} below prove this. 

\begin{lemma} \label{lem:extendible_swappable}
Every $p$-extendible system $(\cN, \cI)$ is a $p$-bilevel swappable system.
\end{lemma}
\begin{proof}
Let $S$, $T$ and $T_H$ be sets as in the definition of $p$-bilevel swappable. We need to show the existence of a collection $\cC = \{Y_t \subseteq S \mid t \in T \setminus T_H\}$ having all the properties stated in this definition. Since $S$ and $T$ are disjoint, by the definition of a $p$-extendible system, there must exist a collection of sets $\{Y_t \subseteq S \mid t \in T\}$ such that
{\begin{enumerate}[1]\renewcommand{\theenumi}{(E\arabic{enumi})}
	\item $(S \setminus \cup_{t \in T'} Y_t) \cup T' \in \cI$ for every set $T' \subseteq T$, and\label{item:multi-exchange}
	\item every element $s \in S$ appears in at most $p$ sets in the collection $\{Y_t \subseteq S \mid t \in T\}$.\label{item:frequency-extendible}
\end{enumerate}}
Our goal is to show that the sub-collection $\{Y_t \mid t \in T \setminus T_H\}$ has all the properties that $\cC$ should have. We do that by considering each one of these properties separately.

\begin{itemize}
	\item \textbf{Property~\ref{item:single_exchange}.} We need to show that $S \setminus Y_t + t \in \cI$ for every $t \in T \setminus T_H$. This follows by plugging $T' = \{t\}$ into \ref{item:multi-exchange}.
	\item \textbf{Property~\ref{item:frequency}.} We need to show that every element $s \in S$ appears in at most $p$ out of the sets in the collection $\{Y_t \mid t \in T \setminus T_H\}$. This immediately follows from \ref{item:frequency-extendible}.
	\item \textbf{Property~\ref{item:large_exchange}.} We need to show that $S_L \cup T_H \in \cI$. By definition, every element of $S_L$ appears $p$ times in sets of $\{Y_t \subseteq S \mid t \in T \setminus T_H\}$, and thus by \ref{item:frequency-extendible}, cannot appear in any set of $\{Y_t \mid t \in T_H\}$. Using~\ref{item:multi-exchange}, this implies
\[
	S_L \cup T_H
	\subseteq
	(S_r \setminus \cup_{t \in T_H} Y_t) \cup T_H
	\in
	\cI
	\enspace.
	\qedhere
\]
\end{itemize}
\end{proof}

We still need to prove the same result for matrod $p$-parity (Lemma~\ref{lem:matroid_partiy_swappable} below). However, before doing so, we need to introduce some tools.

\begin{lemma}[Due to~\cite{greene1973multiple,woodall1974exchange}] \label{lem:matroid1}
Let $S$ and $T$ be bases of a matroid $\cM$. Then, for every partition $S_1 \cup S_2$ of $S$, there exists a partition $T_1 \cup T_2$ of $T$ such that $S_1 \cup T_2, S_2 \cup T_1\in \cI$.
\end{lemma}
\begin{corollary} \label{cor:matroid1}
Let $S$ and $T$ be independent sets of a matroid $\cM$. Then, for every partition $S_1 \cup S_2$ of $S$, there exists a partition $T_1 \cup T_2$ of $T$ such that $S_1 \cup T_2, S_2 \cup T_1\in \cI$.
\end{corollary}
\begin{proof}
Let $S'$ and $T'$ be any two bases of $\cM$ containing $S$ and $T$, respectively, and let $S'_1 \cup S'_2$ be a partition of $S'$ such that $S_1 \subseteq S'_1$ and $S_2 \subseteq S'_2$. By Lemma~\ref{lem:matroid1}, there exists a partition $T'_1 \cup T'_2$ of $T'$ such that $S'_1 \cup T'_2$ and $S'_2 \cup T'_1$ are both independent. Let us now define $T_1 = T'_1 \cap T$ and $T_2 = T'_2 \cap T$. Note that $T_1$ and $T_2$ form a partition of $T$ because $T_1 \cup T_2 = T' \cap T = T$. Furthermore,
\[
	S_1 \cup T_2
	\subseteq
	S'_1 \cup T'_2
	\in
	\cI
	\qquad
	\text{and}
	\qquad
	S_2 \cup T_1
	\subseteq
	S'_2 \cup T'_1
	\in
	\cI
	\enspace.
	\qedhere
\]
\end{proof}

\begin{lemma}[Proposition~6 of~\cite{lason2015list}] \label{lem:partition}
Let $S$ and $T$ be bases of a matroid $\cM$. Then, for every partition $T_1, T_2, \dotsc, T_\ell$ of $T$, there exists a partition $S_1, S_2, \dotsc, S_\ell$ of $S$ such that $(S \setminus S_i) \cup T_i$ is a base of $\cM$ for every $i \in [\ell]$.
\end{lemma}
\begin{corollary} \label{cor:partition}
Let $S$ and $T$ be independent sets of a matroid $\cM$. Then, for every partition $T_1, T_2, \dotsc, T_\ell$ of $T$, there exist a partition $S_1, S_2, \dotsc, S_\ell$ of $S$ such that $(S \setminus S_i) \cup T_i$ is independent in $\cM$ for every $i \in [\ell]$.
\end{corollary}
\begin{proof}
Let $S'$ and $T'$ be any two bases of $\cM$ containing $S$ and $T$, respectively, and let $T'_1, T'_2, \dotsc, T'_\ell$ be a partition of $T'$ such that $T_i \subseteq T'_i$ for every $i \in [\ell]$. By Lemma~\ref{lem:partition}, there exists a partition $S'_1, S'_2, \dotsc, S'_\ell$ of $S'$ such that $(S' \setminus S'_i) \cup T'_i \in \cI$ for every $i \in [\ell]$. Let us now define $S_i = S'_i \cap S$ for every $i \in [k]$. Note that $S_1, S_2, \dotsc, S_\ell$ is a partition of $S$ because $\cup_{i = 1}^\ell S_i = \cup_{i = 1}^\ell (S'_i \cap S) = S' \cap S = S$. Furthermore, for every $i \in [\ell]$,
\[
	(S \setminus S_i) \cup T_i
	=
	(S \setminus S'_i) \cup T_i
	\subseteq
	(S' \setminus S'_i) \cup T'_i
	\in
	\cI
	\enspace.
	\qedhere
\]
\end{proof}

The final tool we need to introduce is a definition. Given a matroid $\cM = (E, \cJ)$ and a set $S \in \cJ$, the contracted matroid $\cM / S$ is a matroid over the ground set $E \setminus S$ such that a set $T \subseteq E \setminus S$ is independent in $\cM / S$ if $S \cup T \in \cJ$. It it is well known that the contracted matroid is indeed a matroid. As promised, using the this tool and the previous ones, we now prove Lemma~\ref{lem:matroid_partiy_swappable}.

\begin{lemma} \label{lem:matroid_partiy_swappable}
Every matorid $p$-parity $(\cN, \cI)$ is a $p$-bilevel swappable system.
\end{lemma}
\begin{proof}
By the definition of a matroid $p$-parity, there exists a matroid $\cQ = (E, \cJ)$ and a function $v\colon \cN \to 2^E$ mapping every element $u \in \cN$ to a disjoint subset of $E$ of size at most $p$ such that $S \in \cI$ if and only if $v(S) \triangleq \cup_{u \in S} v(u) \in \cJ$. Let $S$, $T$ and $T_H$ be sets as in the definition of $p$-bilevel swappable. Since $S$ and $T$ are independent in $\cM$, the sets $v(S)$ and $v(T)$ are independent in $\cQ$. Combining this observation with the fact that $v(T_H) \subseteq v(T)$, Corollary~\ref{cor:matroid1} implies that there exists a partition of $v(S)$ into two sets $V_L, V_H$ such that both $V_L \cup v(T_H)$ and $V_H \cup v(T \setminus T_H)$ are independent in $\cQ$. In particular, we get that both $V_L$ and $v(T \setminus T_H)$ are independent in the contracted matroid $\cQ/V_H$. Thus, by Corollary~\ref{cor:partition}, there exist disjoint subsets $\{Y'_t \subseteq V_L \mid t \in T \setminus T_H\}$ such that, for every $t \in T \setminus T_H$, $(V_L \setminus Y'_t) \cup v(t)$ is independent in $\cQ/V_H$, or equivalently, $(v(S) \setminus Y'_t) \cup v(t) \in \cJ$.

Let us now define, for every $t \in T \setminus T_H$, the set $Y_t \triangleq \{s \in S \mid v(s) \cap Y'_t \neq \varnothing\}$. To prove that $(\cN, \cI)$ is a $p$-bilevel-swappable system, we need to show that the collection $\{Y_t \subseteq S \mid t \in T \setminus T_H\}$ has all the properties that $\cC$ should have according to Definition~\ref{def:bilevel}. We do that by considering each one of these properties separately.

\begin{itemize}
	\item \textbf{Property~\ref{item:single_exchange}.} We need to show that $S \setminus Y_t + t \in \cI$ for every $t \in T \setminus T_H$. This holds because
	\[
		v(S \setminus Y_t + t)
		=
		\bigg(v(S) \setminus \sum_{\mathclap{\substack{s \in S \\ v(s) \cap Y'_t \neq \varnothing}}} v(s)\bigg) \cup v(t)
		\subseteq
		(v(S) \setminus Y'_t) \cup v(t)
		\in
		\cJ
		\enspace,
	\]
	where the inclusion holds because the sets $\{v(s) \mid s \in S\}$ form a partition of $v(S)$.
	\item \textbf{Property~\ref{item:frequency}.} We need to show that every element $s \in S$ appears in at most $p$ out of the sets in the collection $\{Y_t \mid t \in T \setminus T_H\}$. Recall that $s$ appears in $Y_t$ only if $Y'_t \cap v(s) \neq \varnothing$. Since the sets $Y'_t$ are disjoint, this implies that the number of sets $Y_t$ that contain $s$ is
	\[
		\sum_{t \in T \setminus T_H} \mspace{-9mu} \characteristic[Y'_t \cap v(s) \neq \varnothing]
		\leq
		\sum_{t \in T \setminus T_H} \mspace{-9mu} |Y'_t \cap v(s)|
		\leq
		|v(s)|
		\leq
		p
		\enspace.
	\]

	\item \textbf{Property~\ref{item:large_exchange}.} We need to show that $S_L \cup T_H \in \cI$. By definition, every element $s \in S_L$ appears $p$ times in sets of $\{Y_t \mid t \in T \setminus T_H\}$, which means that there are $p$ sets in $\{Y'_t \mid t \in T \setminus T_H\}$ that have a non-zero intersection with $v(s)$. Since the sets $Y'_t$ are disjoint, this implies that at least $p$ different elements of $v(s)$ must appear in some set $Y'_t$. However, $v(s)$ contains at most $p$ elements, and thus, all the elements of $v(s)$ appear in some set $Y'_t$. Hence, $v(s) \subseteq \cup_{t \in T \setminus T_H} Y'_t \subseteq V_L$, and since this is true for every $s \in S_L$,
	\[
		v(S_L) \cup v(T_H)
		\subseteq
		V_L \cup v(T_H)
		\in
		\cJ
		\enspace.
		\qedhere
	\]
\end{itemize}
\end{proof}

%% file: stream-matroid-alg.tex
\section{Single Pass Semi-streaming Algorithm with Boosting}\label{sec:mainAlg}

In this section, we prove Proposition \ref{prop:main}, which we repeat here for convenience.
\PropMainBoost*

In Section~\ref{sec:offline iid}, we design an offline algorithm for maximizing a monotone submodular function subject to a matroid constraint, which 
is inspired by a recent algorithm due to Ganz et al.~\cite{ganz2026poisson}. Then, in Section \ref{sec:stream-matroid1}, we show that a close variant of this offline algorithm can be implemented in a single-pass under the data stream model, which yields our semi-streaming algorithm, and proves Proposition~\ref{prop:main}.

Both our algorithm and the algorithm of Ganz et al.~\cite{ganz2026poisson} maintain an integral solution $A_{i}\in \cI$, but treat it as a fractional solution $h \cdot \characteristic_{A_i}$ for a height $h$ that grows as the algorithm progresses. The algorithm of~\cite{ganz2026poisson} uses a Poisson clock to determine the number and associated heights of the iterations. In contrast, our algorithm performs a deterministic number $\ell$ of iterations and changes the height deterministically with the progress of these iterations.
Informally, our algorithm gets as input an initial height $h_0 \in (0,1]$ and a value $p'\in (0,1]$.
Each one of the $\ell$ iterations starts with a subset $A_{i-1}$ and a height $h_{i-1}$. During the iteration, the algorithm first updates the height to $h_i$ in a {\bf multiplicative way}, and then
samples a random subset of elements from $\cN$, where the probability of each element to be sampled is $p'$. The algorithm uses the sample to perform at most a single swap replacing a single element $u\in A_{i-1}$ with one of the sampled elements. The replacement chosen to be performed is roughly the one that increases $F(h_i \cdot \characteristic_{A_{i - 1}})$ by the most---recall that $F$ is the multilinear extension of $f$. 

We explore two sets of parameters for our offline algorithm. The first set of parameters results in an offline algorithm that obtains $(\frac{e}{e-1}+\delta)$-approximation, evaluates the multilinear extension $F$ $O((r^2+ r/\delta +n)\cdot \log(1/\delta) )$ times and invoked the independence oracle of the matroid constraint $O(n \log r \log(1/\delta))$ times.
The second set of parameters is attractive for the data stream model and achieves the bounds of Proposition~\ref{prop:main}. Under this set of parameters, the number of elements sampled throughout the execution is no larger than $n$ (in expectation), which intuitively is the reason that it is possible to implement (a close variant of) the algorithm with these parameters as a single-pass semi-streaming algorithm. 

\subsection{Offline Algorithm }\label{sec:offline iid}

Our offline algorithm is given as Algorithm~\ref{alg:iid2}. Notice that the algorithm defines, for every $i \in [\ell]$, the height at the end of iteration $i$ as $h_i \triangleq h \cdot (1+ \frac{p}{r-p})^i$, where $ p\triangleq 1-(1-p')^r$. In the analysis, we also use this definition for other values of $i$. One can observe that, as explained in the intuition above, the height increases in a multiplicative way between iterations because the ratio $h_i / h_{i - 1}$ is $1+ \frac{p}{r-p}$ for all $i$ values. We also would like to note two additional things about the pseudocode of Algorithm~\ref{alg:iid2}. First, the code uses the symbol $\bot$ to denote a dummy element that does not exist in the ground set $\cN$. When the algorithm chooses $\bot$ as $u_i$, it does not remove any element from its solution in iteration $i$. Second, the elements $v_i$ and $u_i$ are not well-defined when $R_i$ is empty. However, this is not an issue because their values do not affect the behavior of the algorithm when $R_i = \varnothing$.

\begin{algorithm}\DontPrintSemicolon
\caption{\textsc{Matroid \textbf{\textsc{Offline}} Algorithm}} \label{alg:iid2}
\KwIn{Matroid $(\cN, \cI)$ of rank $r$, non-negative monotone submodular function $f$, set $A_0\in \cI$.\\
\textbf{Parameters:} $h\in(0,1], p' \in (0, 1], p\triangleq 1-(1-p')^r$, and non-negative integer $\ell \leq  \ln(h^{-1})/\ln(1+\nicefrac{p}{(r-p)})$. \\
}
\smallskip\hrule\smallskip
\For{$i = 1$ \KwTo $\ell$}
{ 
Let $h_i \gets h \cdot (1+ \frac{p}{r-p})^i$.\\
Sample a set $R_i$ containing each element $u\in \cN$, independently, with probability $p'$.\\
Let $(u_i,v_i) \in \arg \max_{u\in A_{i-1} \cup \{\bot\}, v\in R_i, A_{i-1}-u+v\in \cI} \{F(h_i \cdot \characteristic_{A_{i-1}+v})+F(h_i \cdot \characteristic_{A_{i-1}-u})\}$.\label{line:best_swap}\\
\If{$R_i \neq \varnothing$ and $F(h_i \cdot \characteristic_{A_{i-1}+v_i})+F(h_i \cdot \characteristic_{A_{i-1}-u_i})> 2 \cdot F(h_i \cdot \characteristic_{A_{i-1}})$}
    {
     $A_i\gets A_{i-1}-u_i+v_i$.
    }\lElse{$A_i \gets A_{i-1}$.} 
}
\Return $A_{\ell}$.
\end{algorithm}

Algorithm~\ref{alg:iid2} is well-defined only when the heights $h_i$ are values in $[0, 1]$ for every $i$. The following observation confirms that this is indeed the case.

\begin{observation}
By definition, $\ell \leq \ln(h^{-1})/\ln(1+\nicefrac{p}{(r-p)})$. Thus, for every integer $0 \leq i \leq \ell$,
\[
	h_i = h \cdot \bigg(1+\frac{p}{r-p}\bigg)^{i} \leq h \cdot \bigg(1+\frac{p}{r-p}\bigg)^{\ell} \leq 1 \enspace.
\]
\end{observation}

The properties of Algorithm~\ref{alg:iid2} are summarizes by Proposition~\ref{thm:iid-value} below. In this proposition, it is convenient to use the following notation $\Delta h_{i} \triangleq h_i - h_{i-1}$.
\begin{observation}\label{obs:dh}
For every integer $i$,
$\Delta h_i = h_i- h_{i-1} = h_{i} - \frac{h_i}{(1+ \frac{p}{r-p})} = h_i \cdot \frac{p}{r}$.
\end{observation}

\begin{proposition}\label{thm:iid-value}
Algorithm~\ref{alg:iid2} returns a set $A_{\ell}$ such that
\begin{align*}
\bE[F(h_{\ell} \cdot \characteristic_{A_{\ell}})] & \geq  \bigg(1- \frac{1}{\prod_{i=1}^{\ell}(1+\Delta h_i)} \bigg) \cdot f(OPT)+ \frac{1}{\prod_{i=1}^{\ell}(1+\Delta h_i)} \cdot F(h \cdot \characteristic_{A_{0}})\\
& \geq \big(1-\lambda \big)\cdot f(OPT) + \lambda \cdot F(h \cdot \characteristic_{A_{0}})
\end{align*}
for any 
\[
    1 \geq \lambda \geq \exp(h - h_{\ell - 1}) =
	\exp\bigg(h\cdot \bigg(1 - \bigg(1+ \frac{p}{r-p}\bigg)^{\!\ell - 1}\bigg)\bigg) \enspace.
\]
Algorithm~\ref{alg:iid2} performs in expectation at most $O(\ell \cdot (n p'+r))$ queries to the value oracle of the multilinear extension $F$ and $O(\ell n p'\log r)$ queries to the independence oracle of the matroid constraint.
\end{proposition}

Section~\ref{sssc:offline_analysis} below proves Proposition~\ref{thm:iid-value}. Section~\ref{sssc:parameters} then studies the guarantees implied by this proposition for two sets of possible parameters.

\subsubsection{Proof of Proposition~\texorpdfstring{\ref{thm:iid-value}}{\ref*{thm:iid-value}}} \label{sssc:offline_analysis}

In this section, we prove Proposition~\ref{thm:iid-value}. Our first goal is to bound the number of oracle queries necessary for implementing Algorithm~\ref{alg:iid2}. This bound is given by Lemma~\ref{lem:queries} below. The proof of Lemma~\ref{lem:queries} uses the following procedure, which is based on the binary search technique suggested by~\cite{chakrabarty2019faster,nguyen2019note}.
\begin{lemma}[Follows from Lemma 4.2 of~\cite{buchbinder2024deterministic}]\label{lem:independence}
There is an algorithm that gets as input a matroid $\cM=(\cN, \cI)$, an independent set $A\in \cI$ in this matroid, an element $v\in \cN\setminus A$ such that $A + v \not \in \cI$, and a weight $w_u$ for every element $u\in A$. The algorithms finds an element $u\in \arg \max_{u\in A, A+v-u\in\cI}\{w_u\}$
using $O(\log r)$ independence oracle queries. 
\end{lemma}

\begin{lemma} \label{lem:queries}
Algorithm~\ref{alg:iid2} can be implemented so that it performs, in expectation, at most $O(\ell \cdot (n p'+r))$ queries to the value oracle of the multilinear extension $F$ and $O(\ell n p'\log r)$ queries to the independence oracle of the matroid constraint.
\end{lemma}
\begin{proof}
In iteration number $i \in [\ell]$ of Algorithm~\ref{alg:iid2}, the algorithm needs to evaluate the multilinear extension on the vectors $h_i \cdot \characteristic_{A_i - 1}$, $h_i \cdot \characteristic_{A_i - 1} - u$ and $h_i \cdot \characteristic_{A_i - 1} + v$ for every $u \in A_{i - 1}$ and $v \in R_i$. The number of such vectors is $1 + |R_i| + |A_{i - 1}|$, and therefore, the expected number of value oracle queries to $F$ used by Algorithm~\ref{alg:iid2} in all its iterations is
\[
	\bE\bigg[\sum_{i = 1}^\ell \{1 + |R_i| + |A_{i - 1}|\}\bigg]
	\leq
	\ell + \sum_{i = 1}^\ell \bE[|R_i|] + r\ell
	=
	\ell + \ell np' + r\ell
	=
	O(\ell \cdot (np' + r))
	\enspace,
\]
where the inequality holds since $A_{i - 1}$ is an independent set, and thus, of size at most $r$.

Independence oracle queries are necessary only for implementing Line~\ref{line:best_swap} of Algorithm~\ref{alg:iid2}. To implement this line, for every element $v\in R_i$, we need to find an element $u$ maximizing $F(h_i \cdot \characteristic_{A_{i-1}-u})$ among all the elements $u \in A_{i - 1} \cup \{\bot\}$ that obey $A_{i - 1} - u + v \in \cI$. The first step towards finding such a $u$ is a single oracle query used to check whether $A_{i - 1} + v \in \cI$. If that is the case, then $u$ can be chosen as $\bot$ because of the monotonicity of $f$. Otherwise, we can use Lemma~\ref{lem:independence} to find $u$ using $O(\log r)$ independence oracle queries by setting $w_u = F(h_i \cdot \characteristic_{A_{i-1}-u})$ for every $u \in A_{i - 1}$. Therefore, the number of independence oracle queries used by a single iteration of Algorithm~\ref{alg:iid2} is at most $1 + |R_i| \cdot O(\log r)$. Adding this up over all iterations, we get that the expected number of independence oracle queries used by the entire algorithm is at most
\[
	\bE\bigg[\sum_{i = 1}^\ell \{1 + |R_i| \cdot O(\log r)\}\bigg]
	=
	\ell + O(\log r) \cdot \sum_{i = 1}^\ell \bE[|R_i|]
	=
	\ell + O(\log r) \cdot \ell np'
	=
	O(\ell n p'\log r)
	\enspace.
	\qedhere
\]
\end{proof}

Our next goal is to analyze the approximation guarantee of Algorithm~\ref{alg:iid2}. To that end, we need the known properties of the multilinear extension given by the next lemma. One simple way to formally prove these properties is via the work~\cite{bian2017guaranteed}. Bian et al.~\cite{bian2017guaranteed} proved that the multilinear extension of a submodular function has a property known as DR-submodularity. Inequality~\eqref{prop:ML2} then follows immediately from the definition of this property (Definition~1 of~\cite{bian2017guaranteed}), and Inequality~\eqref{prop:ML1} holds since DR-submodular functions are concave along non-negative directions (Proposition~4 of~\cite{bian2017guaranteed}).

\begin{lemma} \label{lem:Mulitilinear_properties}
Let $F\colon [0, 1]^\cN \to \nnR$ be the multilinear extension of a non-negative submodular function $f$. Then,
\begin{align}
F(\vx) - F(\vx-\vy) & \geq \inner{\nabla F(\vx)}{\vy} && \forall\;  \vx, \vy\in [0, 1]^\cN, \vx \geq \vy \label{prop:ML1}\\
 F(\vx + \vz) - F(\vx) & \geq F(\vx + \vy + \vz) -F(\vx + \vy) && \forall\; \vx, \vy, \vz \in [0, 1]^\cN, \vx + \vy + \vz \leq \vone \label{prop:ML2}
\end{align}
\end{lemma}

We also need the following known lemma about matroids.


\begin{lemma}[Proved by~\cite{brualdi1969comments} and can also be found as Corollary~39.12a in~\cite{schrijver2003combinatorial}] \label{le:perfect_matching_two_bases}
Let $A$ and $B$ be two bases of a matroid $\cM = (\cN, \cI)$. Then, there exists a bijection $h\colon A \setminus B \rightarrow B \setminus A$ such that for every $u \in A \setminus B$, $(B - h(u)) + u \in \cI$.
\end{lemma}


Given the above tools, we are now ready to lower bound the increase in the value of the solution maintained by Algorithm~\ref{alg:iid2} in a single iteration.
\begin{lemma}\label{lem:main-iid}
For every integer $0 \leq i \leq \ell - 1$,
\[
\bE[F(h_i \cdot \characteristic_{A_i})]  
  \geq  \frac{\Delta h_{i}}{1+ \Delta h_i} \cdot f(OPT) + \frac{1}{1+ \Delta h_i} \cdot \bE[F(h_{i-1} \cdot \characteristic_{A_{i-1}})]
\]
\end{lemma}

\begin{proof}
By the linearity of expectation,
\begin{multline*}
	\bE[F(h_i \cdot \characteristic_{A_i}) - F(h_{i-1} \cdot \characteristic_{A_{i-1}})]
	=
	\bE\big[F(h_i \cdot \characteristic_{A_i}) - F(h_i \cdot \characteristic_{A_{i-1}})\big] \\ \bE\big[F(h_i \cdot \characteristic_{A_{i-1}}) - F(h_{i-1} \cdot \characteristic_{A_{i-1}})\big] \enspace.
\end{multline*}

Below, we lower bound separately each one of the expectations on right hand side of this inequality. Adding together the two lower bounds then proves the lemma. In the proofs of these lower bounds, we implicitly condition on an arbitrary fixed choice for the randomness used in the first $i - 1$ iterations of Algorithm~\ref{alg:iid2}. By the law of total expectation, lower bounds proved conditioned on an arbitrary such randomness apply also unconditionally.
\paragraph{Analysis of the expectation $\bE\big[F(h_i \cdot \characteristic_{A_{i-1}}) - F(h_{i-1} \cdot \characteristic_{A_{i-1}})\big]$:}
Since we condition on a fixed value for the randomness of the first $i - 1$ iterations, the set $A_{i - 1}$ is deterministic. Thus,
\begin{align*}
	\bE\big[F(h_i \cdot \characteristic_{A_{i-1}}) - F(h_{i-1} \cdot \characteristic_{A_{i-1}})\big]
	&=
    F(h_i \cdot \characteristic_{A_{i-1}})-F(h_{i-1} \cdot \characteristic_{A_{i-1}})\\ & = F(h_i \cdot \characteristic_{A_{i-1}})-F(h_i\cdot \characteristic_{A_{i-1}} -\Delta h_i \cdot \characteristic_{A_{i-1}}) \\
    & \geq \Delta h_i \cdot \inner{\nabla F(h_i \cdot \characteristic_{A_{i-1}})}{ \characteristic_{A_{i-1}}}
\end{align*}
where the inequality follows from Inequality~\eqref{prop:ML1} of Lemma~\ref{lem:Mulitilinear_properties}.
\paragraph{Analysis of the expectation $\bE\big[F(h_i \cdot \characteristic_{A_i}) - F(h_i \cdot \characteristic_{A_{i-1}})\big]$:}
Let us show that it deterministically holds that
\begin{equation} \label{eq:increase1}
	F(h_i \cdot \characteristic_{A_{i}})-F(h_i \cdot \characteristic_{A_{i-1}})
	\geq
	0
	\enspace,
\end{equation}
and furthermore, if $R_i \neq \varnothing$ then,
\begin{equation} \label{eq:increase2}
	F(h_i \cdot \characteristic_{A_{i}})-F(h_i \cdot \characteristic_{A_{i-1}})
	\geq
	F(h_i \cdot \characteristic_{A_{i-1}+v_i}) + F(h_i \cdot \characteristic_{A_{i-1}-u_i}) - 2 F(h_i \cdot \characteristic_{A_{i-1}})
	\enspace.
\end{equation}
When $R_i = \varnothing$, $A_i = A_{i - 1}$, and we are done. Therefore, we may assume that $R_i \neq \varnothing$. Let us consider two cases. If Algorithm~\ref{alg:iid2} sets $A_i$ to $A_{i - 1}$, then Inequality~\eqref{eq:increase1} is trivial, and Inequality~\eqref{eq:increase2} holds because in this case $F(h_i \cdot \characteristic_{A_{i-1}+v_i}) + F(h_i \cdot \characteristic_{A_{i-1}-u_i}) - 2 F(h_i \cdot \characteristic_{A_{i-1}}) \leq 0$. The other option is that $A_i$ is set to $A_{i - 1} - u_i + v_i$. In this case, $F(h_i \cdot \characteristic_{A_{i-1}+v_i}) + F(h_i \cdot \characteristic_{A_{i-1}-u_i}) - 2 F(h_i \cdot \characteristic_{A_{i-1}})$ is positive, and therefore, it is enough to prove Inequality~\eqref{eq:increase2}, as this will imply that Inequality~\eqref{eq:increase1} holds as well.
When $u_i \not \in A_{i - 1}$, Inequality~\eqref{eq:increase2} holds since both its sides are equal, and when $v_i \in A_{i - 1}$, Inequality~\eqref{eq:increase2} follows from the monotonicity of $f$. In the remaining case, i.e., when $v_i \not \in A_{i - 1}$ and $u_i \in A_{i - 1}$, Inequality~\eqref{eq:increase2} follows from Inequality~\eqref{prop:ML2} of Lemma~\ref{lem:Mulitilinear_properties}.

Let $B$ be an arbitrary base of the matroid containing the set $A_{i - 1}$. By the monotonicity of $f$, we may assume that $OPT$ is also a base, and thus, by Lemma~\ref{le:perfect_matching_two_bases}, there exists a bijection  $g\colon OPT \setminus B \rightarrow B \setminus OPT$ such that for every $o \in OPT \setminus B$, $(B - g(o)) + o \in \cI$. We extend $g$ to elements of $OPT \cap B$ by defining $g(o) = o$ for every such element $o$. Notice that this extension preserves the property that $(B - g(o)) + o \in \cI$ for every $o \in OPT$, and thus, if we denote
\[
	g'(o)
	\triangleq
	\begin{cases}
		g(o) & \text{if $g(o) \in A_{i - 1}$} \enspace,\\
		\bot & \text{if $g(o) \not \in A_{i - 1}$} \enspace,
	\end{cases}
\]
then we get also $(A_{i - 1} - g'(o)) + o \in \cI$. Let us now define $\characteristic_{\{\bot\}} \triangleq \varnothing$. Using these definitions and Inequalities~\eqref{eq:increase1} and~\eqref{eq:increase2}, we get
\begin{align}
 \bE[F&(h_i \cdot \characteristic_{A_{i}})-F(h_i \cdot \characteristic_{A_{i-1}})]\nonumber\\
	&\geq \Pr[OPT \cap R_i \neq \varnothing] \cdot \bE[F(h_i \cdot \characteristic_{A_{i}})-F(h_i \cdot \characteristic_{A_{i-1}}) \mid OPT \cap R_i \neq \varnothing]\nonumber\\
 &\geq \Pr[OPT \cap R_i \neq \varnothing] \cdot \bE[\max\{0, F(h_i \cdot \characteristic_{A_{i-1}+v_i}) + F(h_i \cdot \characteristic_{A_{i-1}-u_i}) \nonumber\\&\mspace{350mu} - 2 F(h_i \cdot \characteristic_{A_{i-1}})\} \mid OPT \cap R_i \neq \varnothing] \nonumber\\
 & \geq \frac{p}{r}\cdot \sum_{o\in OPT} \mspace{-9mu}\max\{0, F(h_i \cdot \characteristic_{A_{i-1}+o}) + F(h_i \cdot \characteristic_{A_{i-1}-g'(o)}) - 2 F(h_i \cdot \characteristic_{A_{i-1}})\} \label{ineq-02} \\
    & \geq \frac{ph_i}{r}\cdot \sum_{o\in OPT}\bigg[F(h_i \cdot \characteristic_{A_{i-1} - o} + \characteristic_{o})- F(h_i \cdot \characteristic_{A_{i-1}}) - \inner{\nabla F(h_i \cdot \characteristic_{A_{i-1}})}{ \characteristic_{\{g'(o)\}}}\bigg] \label{ineq-11}\\
    & \geq \frac{ph_i}{r}\cdot \left[f(OPT)- F(h_i \cdot \characteristic_{A_{i-1}})\right] -  \frac{p h_i}{r}\cdot \inner{\nabla F(h_i \cdot \characteristic_{A_{i-1}})}{ \characteristic_{A_{i-1}}} \label{ineq-13}\\
    & = \Delta h_i \cdot \left[f(OPT)- F(h_i \cdot \characteristic_{A_{i-1}})\right] -  \Delta h_i \cdot \inner{\nabla F(h_i \cdot \characteristic_{A_{i-1}})}{ \characteristic_{A_{i-1}}} \label{ineq-12}\\
		&\geq{}
		\Delta h_i \cdot \left[f(OPT)- F(h_i \cdot \characteristic_{A_{i}})\right] -  \Delta h_i \cdot \inner{\nabla F(h_i \cdot \characteristic_{A_{i-1}})}{ \characteristic_{A_{i-1}}} \label{ineq-14}
		\enspace,
\end{align}
where Inequality~\eqref{ineq-13} follows from the monotonicity and submodularity of $f$, Equality~\eqref{ineq-12} holds since $\Delta h_i  = h_i \cdot \frac{p}{r}$ by Observation~\ref{obs:dh}, and Inequality~\eqref{ineq-14} is another application of Inequality~\eqref{eq:increase1}.

The above inequality is the promised lower bound on the expectation $\bE[F(h_i \cdot \characteristic_{A_{i}})-F(h_i \cdot \characteristic_{A_{i-1}})]$. The rest of the proof is devoted to justifying Inequalities~\eqref{ineq-02} and~\eqref{ineq-11} in it. To see why Inequality~\eqref{ineq-02} holds, we need to make two observations. The first observation is that since we assume that $OPT$ is a base, $\Pr[R_i \cap OPT \neq \varnothing]= 1-(1-p')^r= p$. Recall now that Algorithm~\ref{alg:iid2} selects $v_i\in R_i$ and $u_i\in A_{i-1}$ that can be swapped and maximize $F(h_i \cdot \characteristic_{A_{i-1}+v_i}) + F(h_i \cdot \characteristic_{A_{i-1}-u_i})$. This implies that for every $o \in OPT \cap R_i$,
\[
	F(h_i \cdot \characteristic_{A_{i-1}+v_i}) + F(h_i \cdot \characteristic_{A_{i-1}-u_i})
	\geq
	F(h_i \cdot \characteristic_{A_{i-1}+o}) + F(h_i \cdot \characteristic_{A_{i-1}-g'(o)})
	\enspace.
\]
If we now define $o'$ to be a uniformly random element of $OPT \cap R_i$, then by symmetry, when we condition on $OPT \cap R_i \neq \varnothing$, $o'$ has an equal probability to be each one of the elements of $OPT$. Thus,
\begin{align*}
	\bE\big[\max\{0&, F(h_i \cdot \characteristic_{A_{i-1}+v_i}) + F(h_i \cdot \characteristic_{A_{i-1}-u_i}) - 2F(h_i \cdot \characteristic_{A_{i-1}})\} \mid OPT \cap R_i \neq \varnothing\big]\\
	\geq{} &
	\bE\big[\max\{0, F(h_i \cdot \characteristic_{A_{i-1}+o'}) + F(h_i \cdot \characteristic_{A_{i-1}-g'(o')}) - 2F(h_i \cdot \characteristic_{A_{i-1}})\} \mid OPT \cap R_i \neq \varnothing\big]\\
	={} &
	\frac{1}{r} \cdot \sum_{o \in OPT} \mspace{-9mu} \big[\max\{0, F(h_i \cdot \characteristic_{A_{i-1}+o}) + F(h_i \cdot \characteristic_{A_{i-1}-g'(o)}) - 2F(h_i \cdot \characteristic_{A_{i-1}})\}\big]
	\enspace.
\end{align*}

Let us now justify Inequality~\eqref{ineq-11}. If $o \not \in A_i$, then the multilinearity of $F$ implies that
\begin{multline*}
	F(h_i \cdot \characteristic_{A_{i-1}+o}) + F(h_i \cdot \characteristic_{A_{i-1}-g'(o)}) - 2 F(h_i \cdot \characteristic_{A_{i-1}})\\
  =
	h_i \cdot [F(h_i \cdot \characteristic_{A_{i-1} - o} + \characteristic_{o})- F(h_i \cdot \characteristic_{A_{i-1}}) - \inner{\nabla F(h_i \cdot \characteristic_{A_{i-1}})}{ \characteristic_{\{g'(o)\}}}]
	\enspace.
\end{multline*}
If $o \in A_i$, then $g'(o) = g(o) = o$, which implies
\begin{align*}
	h_i \cdot [F(h_i \cdot \characteristic_{A_{i-1} - o} +{}& \characteristic_{o})- F(h_i \cdot \characteristic_{A_{i-1}}) - \inner{\nabla F(h_i \cdot \characteristic_{A_{i-1}})}{ \characteristic_{\{g'(o)\}}}]\\
	={} &
	h_i \cdot [(1 - h_i) \cdot \inner{\nabla F(h_i \cdot \characteristic_{A_{i-1}})}{ \characteristic_{\{o\}}} - \inner{\nabla F(h_i \cdot \characteristic_{A_{i-1}})}{ \characteristic_{\{o\}}}]\\
	={} &
	-h^2_i \cdot \inner{\nabla F(h_i \cdot \characteristic_{A_{i-1}})}{ \characteristic_{\{o\}}}
	\leq
	0
	\enspace,
\end{align*}
where the inequality follows from the monotonicity of $f$. 
\end{proof}

We are now ready to prove Proposition~\ref{thm:iid-value}.

\begin{proof}[Proof of Proposition~\ref{thm:iid-value}]
Lemma~\ref{lem:queries} already shows that Algorithm~\ref{alg:iid2} obeys the bounds stated in the proposition on the expected number of oracle queries. Therefore, we concentrate in this proof on bounding the approximation guarantee of Algorithm~\ref{alg:iid2}. For every $i \in [\ell]$, rearranging the guarantee of Lemma~\ref{lem:main-iid} yields
\[F(OPT)- \bE[F(h_i \cdot \characteristic_{A_i})] \leq \frac{1}{1+\Delta h_i}\cdot \left(f(OPT) - \bE[F(h_{i-1} \cdot \characteristic_{A_{i-1}})]\right)\]
Combining these inequalities for all $i$, we get
\begin{align*}
	F(OPT)- \bE[F(h_\ell \cdot \characteristic_{A_\ell})]
	\leq{} &
	\big(f(OPT) - F(h_0 \cdot \characteristic_{A_{0}})\big)\cdot \frac{1}{\prod_{i=1}^{\ell}(1+\Delta h_i)}\\
	={} &
	\big(f(OPT) - F(h \cdot \characteristic_{A_{0}})\big)\cdot \frac{1}{\prod_{i=1}^{\ell}(1+\Delta h_i)}
	\enspace,
\end{align*}
which is equivalent to the first inequality of Proposition~\ref{thm:iid-value}. Furthermore, as $f(OPT) \geq f(A_0) \geq F(h \cdot \characteristic_{A_{0}})$, we get
\[
	\bE[F(h_\ell \cdot \characteristic_{A_\ell})]
	\geq
	(1 - \lambda) \cdot f(OPT) + \lambda \cdot F(h_0 \cdot \characteristic_{A_{0}})
\]
for any $\lambda$ between $\frac{1}{\prod_{i=1}^{\ell}(1+\Delta h_i)}$ and $1$.

Notice now that since the inequality $\ln(1+x)\geq x-x^2/2$ holds for every $x \geq 0$,
\begin{multline*}
	\frac{1}{\prod_{i=1}^{\ell}(1+\Delta h_i)}
	=
	\exp\bigg(-\sum_{i = 1}^\ell \ln(1 + \Delta h_i)\bigg)
	\leq
	\exp\bigg(-\sum_{i=1}^{\ell}\Delta h_i+\frac{1}{2}\sum_{i=1}^{\ell}(\Delta h_i)^2\bigg)\\
	\leq
	\exp\bigg(h_0 - h_{\ell - 1} - \frac{1}{2} \cdot \Delta h_\ell\bigg)
	\leq
	\exp(h_0 - h_{\ell - 1})
	=
	\exp\bigg(h \cdot \bigg(1 - \bigg(1 + \frac{p}{r - p}\bigg)^{\!\ell - 1}\bigg)\bigg)
	\enspace,
\end{multline*}
where the second inequality holds because the fact that $\sum_{i = 1}^\ell \Delta h_i = h_\ell - h_0 \leq 1$ implies the inequality
$
	\sum_{i=1}^{\ell}(\Delta h_i)^2
	\leq
	\Delta h_\ell \cdot \sum_{i=1}^{\ell} \Delta h_i
	\leq
	\Delta h_\ell
$.
Thus, as is guaranteed by the proposition, it is fine to use any $\lambda$ between the rightmost side of the above inequality and $1$.
\end{proof}

\subsubsection{Choosing the Parameters for Our Algorithm} \label{sssc:parameters}

To finalize the analysis of Algorithm~\ref{alg:iid2}, we need to specify values for its parameters. Recall that these parameters are the initial height $h\in (0,1)$, either a probability $p\in (0,1]$ or a probability $p' \in (0, 1]$ (setting either one of these probabilities yields a value for the other one), and a non-negative  integer $\ell\leq \ln(h^{-1})/\ln(1+\nicefrac{p}{(r-p)})$. The following theorem describes and studies a set of parameters that leads to a roughly $(1 - 1/e)$-approximation offline algorithm. While this algorithm is not a necessary step for deriving our semi-streaming results, it might be of independent interest.
\begin{theorem}[Offline Algorithm]
Let $f\colon 2^\cN \rightarrow \nnR$ be a non-negative monotone submodular function and $\cM = (\cN, \cI)$ be a matroid. For a sufficiently small $\delta>0$, invoking Algorithm~\ref{alg:iid2} with $A_0=\varnothing$, $h=\delta$, $p=\min\{\delta \cdot r, \frac{1}{2}\}$ and $\ell= \lfloor \ln(\delta^{-1})/\ln(1+\nicefrac{p}{(r-p)}) \rfloor$ (the maximal possible value) yields an output set $S\in \cI$ such that 
\[\bE[f(S)] \geq\big(1-1/e- O(\delta)\big)\cdot f(OPT) \enspace.\]
The algorithm performs
$O((r^2+ r/\delta +n)\cdot \log(1/\delta) )$ value oracle queries to the multilinear extension $F$ of $f$ and $O(n \log r \log(1/\delta))$ independence oracle queries to the matroid.
\end{theorem}

\begin{proof}
We begin by bounding the number of oracle queries used by the algorithm. Notice that
\[
	\ln(1+\nicefrac{p}{(r-p)})
	\geq
	\frac{\nicefrac{p}{(r-p)}}{1+\nicefrac{p}{(r-p)}}
	=
	\frac{p}{r}
	=
	\frac{\min\{\delta \cdot r, \frac{1}{2}\}}{r}
	=
	\min\Big\{\delta, \frac{1}{2r}\Big\}	
	\enspace,
\]
and thus, $\ell \leq \ln(\delta^{-1})/\ln(1+\nicefrac{p}{(r-p)})= O(\max\{r, 1/\delta\}\log(\delta^{-1}))$. Similarly,
\[
	p'
	=
	1 - \sqrt[r]{1 - p}
	\leq
	1 - \sqrt[r]{1 - 2p + 2p^2}
	\leq
	1 - e^{-2p/r}
	\leq
	\frac{2p}{r}
	=
	\min\Big\{2\delta, \frac{1}{r}\Big\}
	\enspace.
\]
Given these bounds, by Proposition~\ref{thm:iid-value}, the total number of value oracle queries to $F$ is 
\begin{align*}
	O(\ell \cdot (n p'+r))
	={} &
	O\bigg(\max\{r, 1/\delta\}\log(\delta^{-1}) \cdot (n \cdot \min\{1/r, \delta\} + r)\bigg)\\
	={} &
	O((n + \max\{r, 1/\delta\}\cdot r)\cdot \log(1/\delta))
	=
	O((n + r^2 + r/\delta)\cdot \log(1/\delta))
	\enspace,
\end{align*}
and the total number of independence oracle queries is
\[O(\ell n p'\log r)= O(\max\{r, 1/\delta\} \log(\delta^{-1}) \cdot n \cdot \min\{1/r, \delta\} \cdot \log r) = O(n \log r \log(\delta^{-1}))\enspace.\]

Our next goal is to bound the approximation guarantee of the algorithm. Note that
\begin{align*}
	\exp\bigg(h\cdot \bigg(1 - \bigg(1+ \frac{p}{r-p}\bigg)^{\ell - 1}\bigg)\bigg)
	\leq{} &
	\exp\bigg(\delta \cdot \bigg(1 - \bigg(1+ \frac{p}{r-p}\bigg)^{\ln(\delta^{-1})/\ln(1+\nicefrac{p}{(r-p)}) - 2}\bigg)\bigg)\\
	={} &
	\exp\bigg(\delta - \bigg(1+ \frac{p}{r-p}\bigg)^{-2}\bigg)
	\leq
	\exp\bigg(\delta - 1 + \frac{2p}{r-p}\bigg)\\
	\leq{} &
	\exp\bigg(\delta - 1 + \frac{4p}{r}\bigg)
	\leq
	\exp\bigg(5\delta - 1\bigg)
	\leq
	e^{-1} + O(\delta)
	\enspace,
\end{align*}
where the third inequality holds since $r - p \geq 1 - 1/2 = 1/2$. For a small enough $\delta$, the expression $e^{-1} + O(\delta)$ is upper bounded by $1$, and thus, can be used as $\lambda$ in Proposition~\ref{thm:iid-value}, which implies
\begin{align*}
\bE[f(S)] \geq \bE[F(h_{\ell} \cdot \characteristic_{S})] & \geq  \big(1 - (e^{-1} + O(\delta))\big) \cdot f(OPT) + (e^{-1} + O(\delta)) \cdot f(\varnothing) \\
& \geq \big(1-1/e- O(\delta)\big)\cdot f(OPT) \enspace,
\end{align*}
where the first inequality follows from the monotonicity of $f$ and last inequality follows from the non-negativity of $f$.
%
\end{proof}

Another interesting set of parameters for Algorithm~\ref{alg:iid2} aims to boost $A_0$ using samples (sets $R_i$) whose total size does not exceed, in expectation, the size $n$ of the ground set. The following lemma presents and analyzes this set of parameters. Note that the approximation ratio guaranteed by this lemma is identical to the one guaranteed by Proposition~\ref{prop:main} for a semi-streaming algorithm. In fact, the proof of Proposition~\ref{prop:main} in Section~\ref{sec:stream-matroid1} works by showing that given the set of parameters defined by Lemma~\ref{lem:niid}, Algorithm~\ref{alg:iid2} (with a small modification) can be implemented as a semi-streaming algorithm. Intuitively, this is possible because of the above mentioned bound on the (expected) total size of the samples used.

\begin{lemma}\label{lem:niid}
Let $f\colon 2^\cN \rightarrow \nnR$ be a non-negative monotone submodular function and $\cM = (\cN, \cI)$ be a matroid. For a sufficiently small $\delta>0$ and $\delta' = \delta/9$, invoking Algorithm~\ref{alg:iid2} with $h\in (0, 1/e]$, $p'=\frac{\delta'}{r}$, $\ell= \lfloor \frac{r}{\delta'}\rfloor - 1$ and a set $A \in \cI$ yields an output set $S\in \cI$ such that 
\[
\bE[F(e\cdot h \cdot \characteristic_S )]\ge (1-e^{-h(e-1)})(1 - \delta)\cdot f(OPT)+e^{-h(e-1)}\cdot F(h\cdot \characteristic_{A}) \enspace.
\]
The algorithm performs $O(n + r^2/\delta)$ value oracle queries to the multilinear extension $F$ of $f$ and $O(n \log r)$ independence oracle queries to the matroid.
\end{lemma}

\begin{proof}
We begin the proof by observing that the value specified for $\ell$ above falls within the allowed range. This is true since the fact that $p=1-(1-\frac{\delta'}{r})^r \leq \delta'$ implies that
\[
	\frac{\ln(h^{-1})}{\ln(1+\nicefrac{p}{(r-p)})}
	\geq
	\frac{1}{\ln(1+\nicefrac{\delta'}{(r-\delta')})}
	\geq
	\frac{1}{\nicefrac{\delta'}{(r-\delta')}}
	=
	\frac{r}{\delta'} - 1
	\geq
	\ell
	\enspace.
\]

Let us now bound the number of oracle queries. By Proposition~\ref{thm:iid-value}, the total number of value oracle queries to $F$ is 
\[
	O(\ell \cdot (n p'+r))
	=
	O\bigg(\frac{r}{\delta} \cdot \bigg(n \cdot \frac{\delta}{r} + r\bigg)\bigg)
	=
	O(n + r^2/\delta)
	\enspace,
\]
and the total number of independence oracle queries is
\[O(\ell n p'\log r)= O\big(\frac{r}{\delta} \cdot n \cdot \frac{\delta}{r} \cdot \log r\big) = O(n \log r)\enspace.\]

It remains to bound the approximation guarantee of the algorithm. Towards this goal, notice that
\begin{multline*}
	\bigg(1+ \frac{p}{r-p}\bigg)^{\!\ell - 1} \geq \bigg(1+ \frac{p}{r-p}\bigg)^{\! r/\delta' - 3}
    \geq \bigg(1+ \frac{\delta' - (\delta')^2/2}{r}\bigg)^{\! r/\delta' - 3}\\
		\geq e^{1 - \delta'/2} \cdot \bigg(1+ \frac{\delta' - (\delta'^2)/2}{r}\bigg)^{\!-4}
    \geq e^{1 - \delta'/2} \cdot \bigg(1+ \frac{\delta'}{2}\bigg)^{\!- 4} \geq e(1 - \delta'/2) \cdot (1 - 2 \delta') \geq e- 7\delta'
		\enspace,
\end{multline*}
where the second inequality holds since $r/\delta' \geq 3$ and $p = 1-(1-\frac{\delta'}{r})^r \geq 1-e^{-\delta'} \geq \delta' - (\delta')^2/2$. Using the previous inequality, we can now also get
\begin{align*}
	\exp\bigg(h\cdot \bigg(1 - \bigg(1+ \frac{p}{r-p}\bigg)^{\ell - 1}\bigg)\bigg)
	\leq{} &
	\exp\big(h \cdot (1 - e + 7\delta'))\big)
	\leq
	e^{-h(e - 1)} + 7h\delta'\\
	\leq{} &
	e^{-h(e - 1)} + 9\delta'(1 - e^{-h(e - 1)})
	=
	e^{-h(e - 1)} + \delta(1 - e^{-h(e - 1)})
	\enspace,
\end{align*}
where the second inequality holds for a small enough $\delta$, and the last inequality uses the fact that $1 - e^{-h(e - 1)} \geq h(e - 1)/2$. Since the expression $e^{-h(e - 1)} + \delta(1 - e^{-h(e - 1)})$ is upper bounded by $1$, the last inequality implies that this expression can be used as $\lambda$ in Proposition~\ref{thm:iid-value}. Doing so yields
\begin{align*}
\bE[F&(e \cdot h \cdot \characteristic_S)] \geq \bE[F(h_{\ell} \cdot \characteristic_{S})]\\ \geq{} &  \big(1 - e^{-h(e - 1)} - \delta(1 - e^{-h(e - 1)}))\big) \cdot f(OPT) + (e^{-h(e - 1)} + \delta(1 - e^{-h(e - 1)})) \cdot F(h \cdot \characteristic_A) \\
\geq{} & (1-e^{-h(e-1)})(1 - \delta)\cdot f(OPT)+e^{-h(e-1)}\cdot F(h\cdot \characteristic_{A}) \enspace,
\end{align*}
where the last inequality follows from the non-negativity of $f$, and the first inequality follows from the monotonicity of $f$ since
\[
	h_\ell
	=
	h \cdot \bigg(1 + \frac{p}{r - p}\bigg)^{\!\ell}
	\leq
	h \cdot \bigg(1 + \frac{\delta'}{r - \delta'}\bigg)^{\!r/\delta' - 1}
	\leq
	h \cdot e
	\enspace.
	\qedhere
\]
\end{proof}

\subsection{A Single-Pass Streaming Algorithm}\label{sec:stream-matroid1}

In this section, we prove Proposition~\ref{prop:main}. As mentioned above, we do that by showing that a close variant of Algorithm~\ref{alg:iid2} can be implemented as a semi-streaming algorithm when $p' \cdot \ell \leq 1$ and $\ell$ is not too large. The main obstacle for implementing Algorithm~\ref{alg:iid2} as a semi-streaming algorithm is that computing the sets $R_i$ requires random access to the stream. When the arrival order of the elements is uniformly random, it is natural to try to overcome this difficulty by splitting the stream into consecutive windows, and then treating the elements in window $i$ as the set $R_i$. Unfortunately, this creates correlations between the elements of the different windows, which is problematic. Agrawal et al.~\cite{agrawal2019submodular} found a way to bypass this issue by adding to $R_i$ elements that appeared previously in the algorithm's solution. The ideas of Agrawal et al.~\cite{agrawal2019submodular} were later used and improved in other works~\cite{S20,liu2021cardinality,feldman2026streaming}, and this section is essentially an application of these ideas in the context of Algorithm~\ref{alg:iid2}.

Given a non-negative integer $\ell$ and a probability $p'$ such that $\ell \cdot p' \leq 1$, consider the following two random processes for constructing sets $R_1, R_2, \dotsc, R_{\ell}$.
\begin{itemize}
    \item {\bf Process 1:} For each element $u\in\cN$, independently, assign it to at most one of the sets $R_1, R_2, \dotsc, R_{\ell}$, where $u$ is assigned to each set $R_i$ with probability $p'$.
\item {\bf Process 2:} Throw $n$ balls into $\ell$ bins in the following way. For each ball, with probability $1 - \ell p'$ discard it, and otherwise, throw it into a uniformly random bin. The decision for each ball is made independently. Let $w_i$ denote the number of balls that ended up in bin $i$. Process the elements in $\cN$ in a uniformly random order $\pi$, and set $R_i$ to be the set of elements that appear in the order $\pi$ in positions $(\sum_{j=1}^{i-1}w_j)+1, (\sum_{j=1}^{i-1}w_j)+2, \ldots, \sum_{j=1}^{i}w_j$.
\end{itemize}

We have the following observation.
\begin{observation}
The two above random processes produce the same distribution. In other words, for any $\ell$ disjoint subsets $\cN_1, \cN_2, \ldots, \cN_{\ell}$ of $\cN$, the probability $\Pr[\forall_{i \in [\ell]}\; R_i=\cN_i]$ is exactly the same under both processes. 
\end{observation}
\begin{proof}
Let $R'_i$ and $R''_i$ denote the set $R_i$ produced by Processes $1$ and $2$, respectively. By the construction of $w_i$,
\[
	\Pr[\forall_{i \in [\ell]}\; |R'_i|=|\cN_i|]
	=
	\Pr[\forall_{i \in [\ell]}\; |w_i|=|\cN_i|]
	=
	\Pr[\forall_{i \in [\ell]}\; |R''_i|=|\cN_i|]
	\enspace.
\]
Let us now denote by $C$ the number of ways to choose $\ell$ disjoint subsets $\cN'_1, \cN'_2, \ldots, \cN'_{\ell}$ of $\cN$ such that $|\cN'_i| = |\cN_i|$ for every $i \in [\ell]$. By symmetry, the probability $\Pr[\forall_{i \in [\ell]}\; R_i = \cN'_i]$ is the same for every choice of sets $\cN'_1, \cN'_2, \ldots, \cN'_{\ell}$ counted by $C$, including the choice of $\cN'_i = \cN_i$ for all $i \in [\ell]$. Therefore,
\[
	\Pr[\forall_{i \in [\ell]}\; R'_i=\cN_i] = C^{-1} \cdot \Pr[\forall_{i \in [\ell]}\; |R'_i|=|\cN_i|]
\]
and
\[
	\Pr[\forall_{i \in [\ell]}\;  R''_i=\cN_i] = C^{-1} \cdot \Pr[\forall_{i \in [\ell]}\; |R''_i|=|\cN_i|]
	\enspace.
\]
The observation now follows by combining all the above equalities.
\end{proof}

Using the above equivalence between Processes $1$ and $2$, we can now state our semi-streaming algorithm as Algorithm~\ref{alg:streaming2}. This algorithm differs from Algorithm~\ref{alg:iid2} in three ways. First, it draws the sets $R_1, R_2, \dotsc, R_\ell$ from the distribution created by the above processes. Second, in contrast to Algorithm~\ref{alg:iid2}, which selects $v_i$ from the set $R_i$, Algorithm~\ref{alg:streaming2} selects $v_i$ from an extended set $R_i \cup H_{i - 1}$, where $H_{i - 1}$ is the set of all elements that were added to the algorithm's solution in at least one previous iteration. Finally, when there is a tie between multiple pairs that are candidates to be $(u_i, v_i)$, Algorithm~\ref{alg:iid2} chooses one of them arbitrarily, while Algorithm~\ref{alg:streaming2} uses the arrival order for tie breaking.

\begin{algorithm}[t!]
\setcounter{mpfootnote}{1}
\renewcommand{\thefootnote}{\alph{mpfootnote}}
\DontPrintSemicolon
\caption{\textsc{Matroid \textbf{Single-pass Data Stream} Algorithm}} \label{alg:streaming2}
\KwIn{Matroid $(\cN, \cI)$ of rank $r$, non-negative monotone submodular function $f$, set $A_0\in \cI$.\\
\textbf{Parameters:} $h\in(0,1], p' \in (0, 1], p=1-(1-p')^r$ and non-negative integer $\ell \leq  \ln(h^{-1})/\ln(1+\nicefrac{p}{(r-p)})$ obeying $p' \cdot \ell \leq 1$. \\
}
\smallskip\hrule\smallskip
Let $R_1, R_2, \ldots, R_{\ell}$ be random sets drawn from the distribution produced by Processes $1$ and $2$.\\
$H_0\gets \varnothing$.\\
\For{$i = 1$ \KwTo $\ell$}
{ 
Let $h_i\gets h \cdot (1+ \frac{p}{r-p})^i$.\\
Let $(u_i,v_i) \!\gets\! \arg \max_{\substack{u_i\in A_{i-1} \cup \{\bot\}, v_i\in R_i\cup H_{i-1}\\A_{i-1}-u_i+v_i\in \cI}} \left\{\!F(h_i \cdot \characteristic_{A_{i-1}-u_i})\!+\!F(h_i \cdot \characteristic_{A_{i-1}+v_i})\!\right\}$.\footnotemark \label{line:find_best}\\
\If{$R_i \neq \varnothing$ and $F(h_i \cdot \characteristic_{A_{i-1}-u_i})+F(h_i \cdot \characteristic_{A_{i-1}+v_i})>2 \cdot F(h_i \cdot \characteristic_{A_{i-1}})$}
    {
     $A_i\gets A_{i-1}-u_i+v_i$.\\ $H_i \gets H_{i-1}+v_i$.
    } \lElse{
    $A_i \gets A_{i-1}$, $H_i \gets H_{i-1}$.}
}
\Return $A_{\ell}$.
\begin{minipage}{0.94\textwidth}
\setcounter{mpfootnote}{1}
\footnotetext{The maximum might is not be unique, in which case we use a tie breaking rule based on the arrival order $\pi$ of the elements. To implement the tie breaking rule, the algorithm needs to maintain the sets $A_i$ as ordered sets. Specifically, the order of $A_0$ is arbitrary, if $A_i = A_{i - 1}$, then $A_i$ has the same order as $A_{i - 1}$, and if $A_i = A_{i - 1} - u_i + v_i$, then the elements of $A_{i} - v_i$ keep their relative order as in $A_{i - 1}$, and $v_i$ is the last element of $A_i$. We now write $F(h_i \cdot \characteristic_{A_{i-1}-u})+F(h_i \cdot \characteristic_{A_{i-1}+v}) >_{\pi} F(h_i \cdot \characteristic_{A_{i-1}-u'})+F(h_i \cdot \characteristic_{A_{i-1}+v'})$ if the LHS is strictly larger than the RHS, or the two sides are equal and either (i) $v$ precede $v'$ in $\pi$ or (ii) $v = v'$ and $u$ precede $u'$ in $A_{i - 1}$. Then, Line~\ref{line:find_best} chooses a pair $(u_i, v_i)$ that maximizes $F(h_i \cdot \characteristic_{A_{i-1}-u_i})+F(h_i \cdot \characteristic_{A_{i-1}+v_i})$ according to the comparison operator $>_\pi$.}
\end{minipage}
\end{algorithm}

\begin{observation} \label{obs:implementation}
When $\ell = r \cdot \poly(\log n)$, Algorithm~\ref{alg:streaming2} is a semi-streaming algorithm storing at most $O(r+\ell)$ elements.
\end{observation}
\begin{proof}
Algorithm~\ref{alg:streaming2} requires only a very limited access to the sets $R_1, R_2, \dotsc, R_\ell$. Specifically, it needs to iterate over the elements of $R_1$, then iterate over the elements of $R_2$, etc. If the sets $R_1, R_2, \dotsc, R_\ell$ are constructed using Process $2$, with the random order $\pi$ chosen as the arrival order of the elements in the algorithm's data stream, then to simulate this kind of access we only need to store in memory the currently arrived element and the $\ell$ values $w_1, w_2 \ldots, w_{\ell}$, which require only $O(\log n)$ bits of storage each.

In addition to the above, Algorithm~\ref{alg:streaming2} only needs to store the sets $A_i$ and $H_i$ for a single value of $i$. The set $A_i$ is of size at most $r$ since it is feasible, and the set $H_i$ is of size at most $\ell + r$ since it is a subset of $A_0 \cup \{v_1, v_2, \dotsc, v_\ell\}$. Thus, the entire Algorithm~\ref{alg:streaming2} stores only $O(1 + \ell + r) = O(\ell + r)$ elements and requires only semi-streaming space.
\end{proof}

Our next goal is to analyze the approximation guarantee of Algorithm~\ref{alg:streaming2}. Recall that Algorithm~\ref{alg:streaming2} assigns to each set $A_i$ either the value $A_{i - 1} - u_i + v_i$ or the value $A_{i - 1}$. Let $E \subseteq [\ell]$ denote the set of indexes $i$ for which the first of these options occurred.
\begin{observation} \label{obs:H_content}
For every $i \in [\ell]$, if $i \in E$, then $v_i \not \in A_{i - 1}$. Thus,
\[
	A_i \setminus A_{i - 1}
	=
	\begin{cases}
		\{v_i\} & \text{if $i \in E$} \enspace,\\
		\varnothing & \text{if $i \not \in E$} \enspace,
	\end{cases}
\]
which implies $H_i = \{v_j \mid j \in [i] \cap E\} = \cup_{j = 1}^{i} A_i \setminus A_{i - 1}$.
\end{observation}
\begin{proof}
The fact that $i \in E$ implies that
\[
	F(h_i \cdot \characteristic_{A_{i-1}-u_i})+F(h_i \cdot \characteristic_{A_{i-1}+v_i})>2 \cdot F(h_i \cdot \characteristic_{A_{i-1}})
	\enspace.
\]
If $v_i \in A_{i - 1}$, then this inequality simplifies to $F(h_i \cdot \characteristic_{A_{i-1}-u_i})> F(h_i \cdot \characteristic_{A_{i-1}})$, which contradicts the monotonicity of $f$.
\end{proof}

From this point on, it will be useful to pretend that Algorithm~\ref{alg:streaming2} has an extra pre-processing step that chooses a uniformly random order $\sigma$ over the elements of $\cN$, and then modifies $\pi$ by rearranges the elements of each sample $R_i$ according to the order $\sigma$. Notice that since the elements of each set $R_i$ are already ordered uniformly at random in $\pi$, this pre-processing step does not change the output distribution of Algorithm~\ref{alg:streaming2}.

\begin{lemma}\label{lem:independent_sigma}
    For every $i\in [\ell]$, consider the event that the sets $A_0, A_1, \ldots, A_{i-1}$ take particular values $A'_0, A'_1, \dotsc, A'_{i - 1}$ and the order $\sigma$ takes a particular value $\sigma'$. If this event has a non-zero probability, then conditioned on this event, the membership of each element $u \in \cN$ in $R_i \cup H_{i - 1}$ is independent of the membership of other elements in $R_i \cup H_{i - 1}$ and has a probability at least $p'$.
\end{lemma}
\begin{proof}
For every $j \in [i - 1]$ and element $v \in \cN$, let us define $u(v, j)$ as the element in the set $\arg\max_{u \in A'_{j-1} \cup \{\bot\}, A'_{j-1}-u+v\in \cI} F(h_j \cdot \characteristic_{A'_{j-1}-u})$ that appears first in $A_{i - 1}$. Notice that, by the way Algorithm~\ref{alg:streaming2} orders the sets $A_0, A_1, \dotsc, A_\ell$ (for the sake of the tie breaking rule), the orders of these sets are deterministic given what we condition on, and thus, $u(v, j)$ is a deterministic element. 
Furthermore, by Observation~\ref{obs:H_content}, the conditioning of the lemma also implies that, for every $j \in [i - 1]$, the set $H_j$ takes the deterministic value $H'_j = \cup_{h = 1}^{j - 1} (A'_h \setminus A'_{h - 1})$. A more subtle consequence of the conditioning is that for every two elements $v_1, v_2 \in \cN$, it is possible to determine the order they must have in $\pi$ if they happen to be both in $R_j \cup H_{i - 1}$, despite the fact that the membership of $v_1$ and $v_2$ in $R_j \cup H_{i - 1}$ might not be deterministic given the conditioning.

If $A'_j \setminus A'_{j - 1}$ is non-empty, we denote by $v'_j$ the single element in this set. Furthermore, in the last case, since there is a non-zero probability that $A_{j - 1} = A'_{j - 1}$ and $A_j = A'_j$, it must hold that
\begin{equation} \label{eq:good_j}
	F(h_j \cdot \characteristic_{A'_{j-1}-u(v'_j, j)}) + F(h_j \cdot \characteristic_{A'_{j-1}+v'_j})
	>
	2 \cdot F(h_j \cdot \characteristic_{A'_{j-1}})
\end{equation}
and
\begin{equation} \label{eq:compare_H}
	F(h_j \cdot \characteristic_{A'_{j-1}-u(v'_j, j)}) + F(h_j \cdot \characteristic_{A'_{j-1}+v'_j})
	>_{\pi}
	F(h_j \cdot \characteristic_{A'_{j-1}-u(v, j)}) + F(h_j \cdot \characteristic_{A'_{j-1}+v})
\end{equation}
for every element $v \in H'_{j - 1} - v'_j$.

Consider now a possible assignment $(R'_1, R'_2, \ldots, R'_{\ell})$ to the random sets $R_1, R_2, \ldots, R_{\ell}$ (i.e., the sets $R'_1, R'_2, \ldots, R'_{\ell}$ are disjoint subsets of $\cN$). We say that the assignment $(R'_1, R'_2, \ldots, R'_{\ell})$ is {\em valid} with respect to the fixed sets $A'_0, A'_1, \ldots, A'_{i-1}$ if when executing Algorithm~\ref{alg:streaming2} with $R_j = R'_j$ for every $j \in [\ell]$, it outputs $A_j = A'_j$ for every $j \in [i - 1]$.

\begin{claim} \label{clm:valid}
The assignment
$(R'_1, R'_2, \ldots, R'_{\ell})$ is valid with respect to the sets $A'_0, A'_1, \dotsc, A'_{i - 1}$ if and only if the following two conditions hold.
\begin{itemize}
    \item For every $v \in H'_{i-1}$, if $j \in [i - 1]$ is the lowest index such that $\{v\} = A'_j \setminus A'_{j - 1}$, then $v \in R'_j$.
    \item For every $j \in [i - 1]$ and element $v \in R'_j \setminus H'_{i-1}$, either
		\[
			F(h_j \cdot \characteristic_{A'_{j-1}-u(v, j)})+F(h_j \cdot \characteristic_{A'_{j-1}+v})
			\leq
			2 \cdot F(h_j \cdot \characteristic_{A'_{j-1}})
			\enspace,
		\]
		or the set $A'_j \setminus A'_{j - 1}$ is non-empty and
		\[
		F(h_j \cdot \characteristic_{A'_{j-1}-u(v'_j, j)})+F(h_i \cdot \characteristic_{A'_{j-1}+v'_j})
		>_{\pi}
		F(h_j \cdot \characteristic_{A'_{j-1}-u(v, j)})+F(h_j \cdot \characteristic_{A'_{j-1}+v})
		\enspace.
	\]
\end{itemize}
\end{claim}
\begin{proof}
Let us first prove that the two conditions are necessary. Assume that the first condition is violated for some element $v$ and index $j$. If $A_{j - 1} \neq A'_{j - 1}$, then we are done. Otherwise, because of the violation of the first condition, $v \not \in H'_{j - 1} \cup R'_j = H_{j - 1} \cup R_j$, which implies that $v \neq v_j$ because $v_j$ is an element of $R_j \cup H_{j - 1}$ by definition. Thus, $A_{j} \neq A'_{j}$ since $v \in A'_j$ by definition, but $v \not \in A_j \subseteq A_{j - 1} + v_j = A'_{j - 1} + v_j$. 

Assume now that the second condition is violated for some element $v$ and index $j$. Again, we are done if $A_{j - 1} \neq A'_{j - 1}$, and thus, we only need to consider the case that $A_{j - 1} = A'_{j - 1}$. In this case, the violation of the second condition means that $v$ is a better candidate to be $v_j$ compared to $v'_j$, and therefore, $v_j \neq v'_j$, which implies that $v'_j \not \in A_j \subseteq A_{j - 1} + v_j = A'_{j - 1} + v_j$. However, $v'_j \in A'_j$, and thus, $A_j \neq A'_j$.

Let us now prove that the two conditions suffice, i.e., that $A_j = A'_j$ for every integer $0 \leq j \leq i - 1$ when the conditions hold. We prove that by induction on $j$. For $j = 0$ this is trivial since $A_0$ is a deterministic set. Assume now that $A_{j - 1} = A'_{j - 1}$ for some $1 \leq j \leq i - 1$, and let us prove that $A_j = A'_j$. There are two cases to consider. If $A'_j = A'_{j - 1}$, then since there is a non-zero probability that $A_{j'} = A'_{j'}$ for every $j' \in [i - 1]$, $F(h_j \cdot \characteristic_{A_{j-1}-u(v, j)})+F(h_j \cdot \characteristic_{A_{j-1}+v}) \leq 2 \cdot F(h_j \cdot A_{j - 1})$ for every $v \in H'_{j - 1}$. Together with the second condition, this implies that
\[
	F(h_j \cdot \characteristic_{A_{j-1}-u_j})+F(h_j \cdot \characteristic_{A_{j-1}+v_j})
			\leq
			2 \cdot F(h_j \cdot \characteristic_{A_{j-1}})
	\enspace,
\]
and therefore, Algorithm~\ref{alg:streaming2} sets $A_j = A_{j - 1} = A'_{j - 1} = A'_j$. Consider now the case that $A'_j \neq A'_{j - 1}$. In this case, the first condition implies that $v'_j \in R'_j \cup H'_{j - 1} = R_j \cup H'_{j - 1}$ and $(H'_{i - 1} \setminus H'_{j - 1}) \cap R_j = \varnothing$ , and the second condition implies that $v'_j$ is a better candidate to be $v_j$ compared to any element of $R_j \setminus H'_{i - 1}$. Combined with the fact that, by Inequality~\eqref{eq:compare_H}, $v'_j$ is a better candidate compared to any element of $H'_{j - 1}$, we get that $v_j = v'_j$. This completes the proof since, due to Inequality~\eqref{eq:good_j}, the equality $v_j = v'_j$ implies $A_j = A_{j - 1} - u(v_j, j) + v_j = A'_{j - 1} - u(v'_j, j) + v'_j = A'_j$.
\end{proof}

Recall that the relative order according to $\pi$ of every two elements $v_1, v_2 \in R_j \cup H_{i - 1}$ is deterministic given our conditioning. This implies that for every element $v \in \cN \setminus H'_{i - 1}$, there exists a set $I_v \subseteq [i - 1]$ such that the second condition of Claim~\ref{clm:valid} is unsatisfied with respect to $v$ if and only if $v \in R'_j$ for some $j \in I_v$. In other words, the condition is satisfied with respect to $v$ when $v\in R'_j$ for some $j\in [\ell]\setminus I_v$, or $v$ does not belong to any set $R'_j$. Similarly, for every element $v \in H'_{i - 1}$, there exists a single index $j_v \in [i - 1]$ such that the first condition of Claim~\ref{clm:valid} is satisfied with respect to $v$ if and only if $v \in R_{j_v}$. For consistency with the elements of $\cN \setminus H'_{i - 1}$, let us denote $I_v \triangleq [\ell] - j_v$. Claim~\ref{clm:valid} shows that the event that $A_j = A'_j$ for every $j \in [i - 1]$ is equivalent to the event that $v \not \in R_j$ for any $v \in \cN$ and $j \in I_v$. Therefore,
\begin{multline*}
	\Pr[R_i = R'_i \mid \forall_{j \in [i - 1]}\; A_j = A'_j]
	=
	\Pr[R_i = R'_i \mid \forall_{v \in \cN, j\in I_v}\; v \not \in R_j]\\
	=
	\prod_{v \in R'_i} \Pr[v \in R_i \mid \forall_{j\in I_v}\; v \not \in R_j] \cdot \prod_{v \not \in R'_i} \Big(1 - \Pr[v \in R_i \mid \forall_{j\in I_v}\; v \not \in R_j]\Big)
	\enspace.
\end{multline*}
Notice that the value $q_v \triangleq \Pr[v \in R_i \mid \forall_{j\in I_v}\; v \not \in R_j]$ is deterministic for every $v \in \cN$, and thus, the last equality proves that, given our conditioning, $R_i$ independently contains every element $v \in \cN$ with probability $q_v$. As a consequence, we get that the membership of every element $v \in \cN$ in $R_i \cup H_{i - 1}$, under our conditioning, is also independent because $H_{i - 1}$ is a deterministic set under this conditioning. Let us now lower bound the probability that each element $v \in \cN$ belongs to $R_i \cup H_{i - 1}$ under our conditioning. If $v \in H'_{i - 1}$, then 
\[
	\Pr[v \in R_i \cup H_{i - 1} \mid \forall_{j \in [i - 1]}\; A_j = A'_j]
	=
	\Pr[v \in R_i \cup H'_{i - 1} \mid \forall_{j \in [i - 1]}\; A_j = A'_j]
	=
	1
	\enspace.
\]
Otherwise, if $v \in \cN \setminus H'_{i - 1}$, then $I_v \subseteq [i - 1]$, which implies $i \not \in I_v$, and therefore,
\begin{multline*}
	\Pr[v \in R_i \cup H_{i - 1} \mid \forall_{j \in [i - 1]}\; A_j = A'_j]
	=
	\Pr[v \in R_i \mid \forall_{j \in [i - 1]}\; A_j = A'_j]\\
	=
	\Pr[v \in R_i \mid \forall_{u \in \cN, j\in I_u}\; u \not \in R_j]
	=
	\Pr[v \in R_i \mid \forall_{j\in I_v}\; v \not \in R_j]
	=
	\frac{p'}{1 - p'|I_v|}
	\geq
	p'
	\enspace.
\end{multline*}
This completes the proof of the lemma.
\end{proof}

\begin{corollary}\label{cor:independent}
    For every $i\in [\ell]$, consider the event that the sets $A_0, A_1, \ldots, A_{i-1}$ take particular values $A'_0, A'_1, \dotsc, A'_{i - 1}$. If this event has a non-zero probability, then conditioned on this event, there exists a random subset $R'_i \subseteq R_i \cup H_{i - 1}$ such that each element $u \in \cN$ belongs to $R'_i$ with probability $p'$, independently.
\end{corollary}
\begin{proof}
For every order $\sigma'$ over $\cN$, let us denote by $\cE_{\sigma'}$ the event that $\sigma = \sigma'$ and $A_j = A'_j$ for every $j \in [i - 1]$. Lemma~\ref{lem:independent_sigma} shows that if the event $\cE_{\sigma'}$ has a non-zero probability, then conditioned on this event, each element $u \in \cN$ belongs to $R_i \cup H_{i - 1}$ with an independent probability $p_{u, \sigma'} \geq p'$. Therefore, if we denote by $R'_i$ a subset of $R_i \cup H_{i - 1}$ that contains every element $u \in \cN$ with probability $p' / p_{u, \sigma'}$, independently, then conditioned on $\cE_{\sigma'}$, the set $R'_i$ independently contains every element of $\cN$ with probability $p'$.

To prove the corollary, we need to show that $R'_i$ has the same distribution also when conditioned on the weaker event that $A_j = A'_j$ for every $j \in [i - 1]$. To see that, notice that by the law of total probability, for every $S \subseteq \cN$,
\begin{align*}
	\Pr[R'_i = S \mid \forall_{j \in [i - 1]}\; A_j = A'_j]
	={} &
	\sum_{\sigma' : \Pr[\cE_{\sigma'}] > 0} \mspace{-18mu} \Pr[\cE_{\sigma'}] \cdot \Pr[R'_i = S \mid \cE_{\sigma'}]\\
	={} &
	\sum_{\sigma' : \Pr[\cE_{\sigma'}] > 0} \mspace{-18mu} \Pr[\cE_{\sigma'}] \cdot (p')^{|S|}(1 - p')^{n - |S|}
	=
	(p')^{|S|}(1 - p')^{n - |S|}
	\enspace.
	\qedhere
\end{align*}
\end{proof}

We are now ready to prove Proposition~\ref{prop:main}.

\begin{proof}[Proof of Proposition~\ref{prop:main}]
Our first goal is to show that Lemma~\ref{lem:main-iid} applies also to Algorithm~\ref{alg:streaming2}. The proof of this lemma conditions on all the randomness of the algorithm prior to iteration $i$, and assumes that after this conditioning the set $R_i$ has the following two properties.
\begin{itemize}
	\item It contains each element of $\cN$ with probability $p'$, independently.
	\item Every element of $R_i$ is considered as a possible option for $v_i$
\end{itemize}
To make the proof apply to Algorithm~\ref{alg:streaming2}, we change the conditioning to be on the values of the set $A_0, \ldots, A_{i-1}$, and we replace $R_i$ in the proof with the set $R'_i$ from Corollary~\ref{cor:independent}. One can verify that $R'_i$ indeed has the above two properties given this conditioning, and thus, after these modifications the resulting proof indeed shows that Lemma~\ref{lem:main-iid} applies also to Algorithm~\ref{alg:streaming2}.

Using Lemma~\ref{lem:main-iid}, we can now repeat the relevant part of the proof of Proposition~\ref{thm:iid-value}, and get that the approximation guarantee of this proposition applies also to Algorithm~\ref{alg:streaming2}. Thus, by setting the parameters of Algorithm~\ref{alg:streaming2} as in Lemma~\ref{lem:niid}, i.e., $p'=\frac{\delta'}{r}$ and $\ell= \lfloor \frac{r}{\delta'}\rfloor - 1$ for $\delta' = \delta/9$, we get for Algorithm~\ref{alg:streaming2} the approximation guarantee stated in Lemma~\ref{lem:niid}, which is the exact same approximation guarantee stated in Proposition~\ref{prop:main}. To complete the proof of this proposition, it only remains to notice that, according to Observation~\ref{obs:implementation}, the number of elements Algorithm~\ref{alg:streaming2} stores given the above values for the parameters is
\[
	O(r + \ell)
	=
	O(r + r/\delta)
	=
	O(r/\delta)
	\enspace.
	\qedhere
\]
\end{proof}

%% file: k-System-Hardness.tex
\section{\texorpdfstring{$p$}{p}-Systems Hardness} \label{sec:hardness}

In this section, we prove Theorem~\ref{thm:hardness}, which we repeat below for convenience. A non-encoding algorithm is an algorithm that obeys the following properties.
\begin{itemize}
	\item If an element $u$ is explicitly stored in the memory of the algorithm after the arrival of an element $v \neq u$, then $u$ must have also been explicitly stored in the memory of the algorithm before the arrival of $v$.
	\item Only elements that are explicitly stored in the memory of the algorithm when it terminates can appear in its output set. Similarly, when the algorithm queries the independence oracle on a set, this set can contain only elements that are explicitly stored in the algorithm's memory at the time of the query.
\end{itemize}

\thmHardness*

The remainder of this section is devoted to the proof of Theorem~\ref{thm:hardness}. Section~\ref{ssc:construction} describes and analyzes a $p$-system $(\cN, \cI)$ and a linear function $f$ that underlay the proof. Section~\ref{ssc:hardness}, then, leverages these $p$-system and function to prove Theorem~\ref{thm:hardness}.

\subsection{\texorpdfstring{$p$}{p}-System and Linear Function Construction} \label{ssc:construction}

In this section, we describe the $p$-system $(\cN, \cI)$ and linear function $f$ that we later use to prove Theorem~\ref{thm:hardness} in Section~\ref{ssc:hardness}. Our construction has four positive integer parameters $k$, $\ell$, $r$ and $\alpha \geq 2$ (whose values are determined later). For every integer $1 \leq i \leq \ell$, let us define two sets
\[
	A_i
	=
	\{a_{i, j} \mid 1 \leq j \leq \alpha^i \cdot p^2\}
	\qquad
	\text{and}
	\qquad
	B_i
	=
	\{b_{i, j} \mid 1 \leq j \leq \alpha^i \cdot rp^2\}
	\enspace.
\]
The ground set of the $p$-system we construct is simply $ \cN = \bigcup_{i = 1}^\ell (A_i \cup B_i)$, and a set $S$ is independent in the $p$-system (i.e., belongs to $\cI$) if for every $1 \leq i \leq \ell$ it holds that
\begin{equation} \label{eq:upper_bound_construction}
	|S \cap (A_i \cup B_i)|
	\leq
	\begin{cases}
	0 & \text{if $\exists_{i < i' \leq \ell} |S \cap B_{i'}| > \alpha^{i'}$} \enspace,\\
	\alpha^i \cdot k & \text{if $\forall_{i < i' \leq \ell} |S \cap B_{i'}| \leq \alpha^{i'}$ and $S \cap B_i \neq \varnothing$} \enspace,\\
	\alpha^i \cdot k^2 & \text{otherwise} \enspace.
	\end{cases}
\end{equation}
The linear function $f\colon 2^\cN \to \nnR$ we construct gives a weight of $\alpha^{-i}$ for every element of $A_i \cup B_i$. In other words, for every set $S \subseteq \cN$,
\[
	f(S)
	=
	\sum_{i \in [\ell]} \alpha^{-i} \cdot |S \cap (A_i \cup B_i)|
	\enspace.
\]

The following lemma shows that the above constructed $(\cN, \cI)$ is indeed a $p$-system for an appropriate choice of values for the parameters.

\begin{lemma} \label{lem:k-system}
If $p \geq k + k^2 / (\alpha - 1)$, then the pair $(\cN, \cI)$ is a $p$-system. In particular, when $\alpha = k^2 + 1$, this pair is a $p$-system when $p \geq k + 1$.
\end{lemma}
\begin{proof}
Since the right hand side of Inequality~\eqref{eq:upper_bound_construction} is a down-monotone function of $S$ (i.e., $g(S) \leq g(S')$ whenever $S \supseteq S'$), the set $\cI$ of independent sets is down-closed. Furthermore, since this right hand side is always non-negative, the empty set is guaranteed to be independent. Thus, to prove the lemma, it suffices to show that $(\cN, \cI)$ obeys the ``ratio of base sizes'' property of $p$-systems. In other words, for every subset $\cN'$ of $\cN$ and an independent set $T \subseteq \cN'$, we need to show that every base of $\cN'$ is of size at least $|T| / p$. Assume towards a contradiction that this is not the case, \ie, that there exists a base $C$ of $\cN'$ whose size is less than $|T| / p$. Since $p > k$, for such a base there must exist an integer $1 \leq i \leq \ell$ such that $|C \cap (A_i \cup B_i)| < |T \cap (A_i \cup B_i)| / p \leq |T \cap (A_i \cup B_i)| / k$. We assume from now on that $i$ is the highest integer of this kind. It is also important to note that the inequality $|C \cap (A_i \cup B_i)| < |T \cap (A_i \cup B_i)| / k$ implies that the set $\Delta_i \triangleq (T \setminus C) \cap (A_i \cup B_i)$ is not empty. However, since $C$ is a base, none of the elements of $\Delta_i$ can be added to $C$ without violating independence.

We now need to consider a few cases. The first case is the case in which there does not exist an integer $i \leq i' \leq \ell$ such that $|C \cap B_{i'}| \geq \alpha^{i'}$, and in particular, $|C \cap B_i| < \alpha^i$. Combined with the observation that $|C \cap (A_i \cup B_i)| < |T \cap (A_i \cup B_i)| / k \leq \alpha^i \cdot k$, we get that in this case any single element of $A_i \cup B_i$ can be added to $C$ without violating independence. However, this contradicts the fact that $\Delta_i$ is non-empty.

The second case is the case in which the only integer $i \leq i' \leq \ell$ such that $|C \cap B_{i'}| \geq \alpha^{i'}$ is $i' = i$, which in particular implies that $C \cap (A_i \cup B_i)$ is non-empty. In this case, the inequality $|C \cap (A_i \cup B_i)| < |T \cap (A_i \cup B_i)| / k \leq \alpha^i \cdot k$ only guarantees that any single element of $A_i$ can be added to $C$ without violating independence, but this is no longer necessarily true for elements of $B_i$ (because adding them might affect the right hand side of Inequality~\eqref{eq:upper_bound_construction} for other $i$ values). Since no element of $\Delta_i$ can be added to $C$, this implies $\Delta_i \subseteq B_i$. Recalling now that $\Delta_i$ is a non-empty subset of $T$, we get that $T \cap B_i \neq \varnothing$, and plugging this observation into Inequality~\eqref{eq:upper_bound_construction} yields that the independent set $T$ must obey $|T \cap (A_i \cup B_i)| / k \leq \alpha^i \leq |C \cap B_i| \leq |C \cap (A_i \cup B_i)|$---the second inequality follows from the definition of the case. However, this contradicts the definition of $i$ .

It remains to consider the case in which there exists an integer $i < i' \leq \ell$ such that $|C \cap B_{i'}| \geq \alpha^{i'}$; and therefore, for every $\eps \in (0, 1)$,
\begin{align*}
	|C|
	\geq{} &
	\eps \cdot |C \cap (A_{i'} \cup B_{i'})| + (1 - \eps) \sum_{i'' = i'}^\ell |C \cap (A_{i''} \cup B_{i''})|
	\geq
	\eps \cdot \alpha^{i'} + \frac{1 - \eps}{k} \cdot \sum_{i'' = i'}^\ell |T \cap (A_{i''} \cup B_{i''})|\\
	>{} &
	\frac{\eps(\alpha - 1)}{k^2} \cdot \sum_{i'' = 1}^{i' - 1} \alpha^{i''} k^2 + \frac{1 - \eps}{k} \cdot \sum_{i'' = i'}^\ell |T \cap (A_{i''} \cup B_{i''})|\\
	\geq{} &
	\frac{\min\{\eps(\alpha - 1) / k, 1 - \eps\}}{k} \cdot \sum_{i'' = 1}^\ell |T \cap (A_{i''} \cup B_{i''})|
	=
	\frac{\min\{\eps(\alpha - 1) / k, 1 - \eps\}}{k} \cdot |T|
	\enspace,
\end{align*}
where the second inequality follows from the definitions of $i$ and $i'$. Setting now $\eps = \frac{1}{(\alpha - 1)/k + 1}$, we get
\[
	|C|
	\geq
	\frac{(\alpha - 1)/k}{k[(\alpha - 1)/k + 1]} \cdot |T|
	=
	\frac{|T|}{k + k^2/(\alpha - 1)} \cdot |T|
	\geq
	\frac{|T|}{p}
	\enspace,
\]
which contradicts the definition of $C$, and thus, completes the proof of the lemma.
\end{proof}

Next, we prove that $\cI$ contains a large value independent set, but no independent set in this system is too large in size.

\begin{lemma} \label{lem:opt_value}
The maximum size $|S|$ and value $f(S)$ of an independent set $S \in \cI$ are $\Theta(k^2\alpha^{\ell})$ and $\ell \cdot k^2$, respectively.
\end{lemma}
\begin{proof}
Consider the set $S = \bigcup_{i = 1}^\ell A_i$. One can verify that $S$ is independent, and its size is
\[
	|S|
	=
	\sum_{i = 1}^\ell |A_i|
	=
	\sum_{i = 1}^\ell (\alpha^i \cdot k^2)
	=
	\alpha k^2 \cdot \frac{\alpha^\ell - 1}{\alpha - 1}
	=
	\Theta(1) \cdot k^2(\alpha^\ell - 1)
	=
	\Theta(k^2\alpha^\ell)
	\enspace,
\]
where the third equality holds since $\alpha \geq 2$. Additionally, the value of $S$ is
\[
	f(S)
	=
	\sum_{i = 1}^\ell |A_i| \cdot \alpha^{-i}
	=
	\sum_{i = 1}^\ell (\alpha^i \cdot k^2) \cdot \alpha^{-i}
	=
	\sum_{i = 1}^\ell k^2
	=
	\ell \cdot k^2
	\enspace.
\]

Consider now an arbitrary independent set $S' \in \cI$. For every integer $1 \leq i \leq \ell$, the size of $S' \cap (A_i \cup B_i)$ must be upper bounded by $\alpha^i \cdot k^2 = |S \cap (A_i \cup B_i)|$ due to Inequality~\eqref{eq:upper_bound_construction}. Thus, both the size and value of $S'$ cannot be larger than the size and value of $S$, respectively.
\end{proof}

In the next section, we show that using an adversarial arrival order, one can force a data stream algorithm with limited memory to produce an independent set of $(\cN, \cI)$ that has a small expected intersection with $A_i$ for every $1 \leq i \leq \ell$. The next lemma shows that such an independent set has a low value.

\begin{lemma} \label{lem:S_bound}
Let $S$ be an independent set of $(\cN, \cI)$. Then, $f(S)$ is upper bounded by $k + \ell + \sum_{i = 1}^\ell |S \cap A_i| \cdot \alpha^{-i}$.
\end{lemma}
\begin{proof}
Let $1 \leq i' \leq \ell$ be the maximum integer for which $|S \cap B_{i'}| > \alpha^{i'}$. If there is no such integer, let $i' = 1$. The independence of $S$ guarantees that $S \cap (A_i \cup B_i) = \varnothing$ for every $1 \leq i < i'$. Hence, we can upper bound $f(S)$ by
\begin{align*}
	\sum_{i = 1}^\ell |S \cap (A_i \cup B_i)| \cdot \alpha^{-i}
	={} &
	\sum_{i = i'}^{\ell} |S \cap A_i| \cdot \alpha^{-i} + \sum_{i = i' + 1}^{\ell} |S \cap B_i| \cdot \alpha^{-i} + |S \cap B_{i'}| \cdot \alpha^{-i'}\\
	\leq{} &
	\sum_{i = 1}^\ell |S \cap A_i| \cdot \alpha^{-i} + \sum_{i = i' + 1}^{\ell} \alpha^i \cdot \alpha^{-i} + (\alpha^{i'}\cdot k) \cdot \alpha^{-i'}\\
	={} &
	\sum_{i = 1}^\ell |S \cap A_i| \cdot \alpha^{-i} + (\ell - i') + k
	\leq
	\sum_{i = 1}^\ell |S \cap A_i| \cdot \alpha^{-i} + \ell + k
	\enspace,
\end{align*}
where the first inequality is based on the choice of $i'$ and the observation that Inequality~\eqref{eq:upper_bound_construction} implies that $|S \cap B_i| \leq \alpha^i \cdot k$ for every independent set $S$ and integer $1 \leq i \leq \ell$.
\end{proof}

\subsection{Hardness of the Constructed Independence System} \label{ssc:hardness}


In this section, we prove Theorem~\ref{thm:hardness} using the $p$-system $(\cN, \cI)$ and linear function $f$ defined in Section~\ref{ssc:construction}. As was mentioned before Lemma~\ref{lem:S_bound}, to prove Theorem~\ref{thm:hardness}, we need to show that a data stream algorithm with a bounded memory cannot produce an independent set of $(\cN, \cI)$ containing many elements of $\bigcup_{i = 1}^\ell A_i$. The first part of this section is devoted to proving this.

Let $ALG$ be an arbitrary deterministic non-encoding data stream algorithm for maximizing $f$ subject to the constraint $(\cN, \cI)$. We assume that the order in which the elements of $\cN$ are fed to $ALG$ is chosen as follows. First, the elements of $A_1 \cup B_1$ are fed in a uniformly random order, then the elements of $A_2 \cup B_2$ are fed in a uniformly random order, and so on, until finally, the elements of $A_\ell \cup B_\ell$ are fed to $ALG$ in a uniformly random order. Let $M$ be the maximum number of elements that $ALG$ might store in its memory given this distribution over arrival orders. Additionally, for every $1 \leq i \leq \ell$, let $\cE_i$ be the event that at no point during the execution of $ALG$ more than $\alpha^i$ elements of $A_i$ reside within $ALG$'s memory.
\begin{lemma} \label{lem:single_level_good}
For every $1 \leq i \leq \ell$, $\Pr[\cE_i] \geq 1 - k^2M/r$.
\end{lemma}
\begin{proof}
Let $H_i$ be an event fixing an arbitrary order for the elements of $\bigcup_{i' = 1}^{i - 1} (A_{i'} \cup B_{i'})$. Below, we prove the lemma conditioned on the event $H_i$. Note that this suffices to prove the lemma by the law of total probability. Furthermore, for simplicity of notation, we omit the conditioning on $H_i$ in this proof. In other words, all the expectations, probabilities and values in this proof are implicitly conditions on $H_i$. 

For every integer $1 \leq j \leq |A_i \cup B_i|$, let us denote by $u_j$ the $j$-th element of $A_i \cup B_i$ to arrive and by $M_{i, j}$ the set of elements explicitly stored in the memory of $ALG$ immediately after the arrival of $u_j$. Assume now that we modify the independence oracle of $(\cN, \cI)$ in such a way that it treats elements of $A_i$ like they are elements of $B_i$. Given this modified oracle, since all the elements of $A_i \cup B_i$ are now indistinguishable from the point of view of $ALG$, the algorithm can choose the set of elements that end up in $M_{i, j}$ only based on the positions of these elements in the arrival order. In other words, there exists a set $H_{i, j}$ containing $|M_{i, j}|$ indexes out of $1, 2, \dotsc, |A_i \cup B_i|$ (note that $|M_{i, j}|$ is deterministic in this setting) such that $M_{i, j} = \{u_k \mid k \in H_{i, j}\}$.

Let us consider now the behavior of $ALG$ with the original independence oracle. Since the elements of $\bigcup_{i' = i + 1}^\ell (A_{i'} \cup B_{i'})$ arrive only after the arrival of the elements of $A_i \cup B_i$, the original independence oracle behaves during the arrival of the elements of $A_i \cup B_i$ just like the modified independence oracle as long as the memory of $ALG$ does not explicitly store at least $\alpha^i k + 1 \geq \alpha^i + 1$ elements of $A_i$. Thus, as long as this event does not happen up until the arrival of $u_j$ (including), the set $M_{i, j}$ is still given by $\{u_k \mid k \in H_{i, j}\}$. Consequently, since $ALG$ is a non-encoding algorithm, if its memory ever explicitly stores at least $\alpha^i + 1$ elements of $A_i$, then the first time in which this happens must be when some element $u_j \in A_i$ arrives, and the set $\{u_k \mid k \in H_{i, j - 1}\}$ already contains $\alpha^i$ elements of $A_i$ and is of size $|H_{i, j - 1}| = |M_{i, j - 1}| \leq M$. To upper bound the probability that such a thing happens, let us denote the elements of $A_i$ by $a_1, a_2, \dotsc, a_{|A_i|}$ in an arbitrary (but fixed) order; and for every $1 \leq h \leq |A_i|$, let $\cE_{i, h}$ be the event that the set $H_{i, j - 1}$ is of size at most $M$ and $|\{u_k \mid k \in H_{i, j - 1}\} \cap A_i| = \alpha^i$, where $j$ is the value obeying $u_j = a_h$. Given the above discussion, $\cE_i$ happens whenever none of the events $\cE_{i, h}$ happens, and thus,
\begin{equation} \label{eq:event_sum}
	\Pr[\cE_i]
	\geq
	1 - \Pr\left[\bigcup_{h = 1}^{|A_i|} \cE_{i, h}\right]
	\geq
	1 - \sum_{h = 1}^{|A_i|} \Pr[\cE_{i, h}]
	\enspace,
\end{equation}
where the second inequality follows from union bound.

Fix now some integers $1 \leq h \leq |A_i|$ and $1 \leq j \leq |A_i \cup B_i|$, and let us upper bound the probability $\Pr[\cE_{i, h} \mid a_h = u_j]$. If $|H_{i, j - 1}| \leq M$, then
\begin{multline*}
	\Pr[\cE_{i, h} \mid a_h = u_j]
	=
	\Pr[|A_i \cap \{u_k \mid k \in H_{i, j - 1}\}| = \alpha^i \mid a_h = u_j]\\
	\leq
	\frac{\bE[|A_i \cap \{u_k \mid k \in H_{i, j - 1}\}| \mid a_h = u_j]}{\alpha^i}
	=
	\frac{\frac{|A_i| - 1}{|A_i \cup B_i| - 1} \cdot |H_{i, j - 1}|}{\alpha^i}
	\leq
	\frac{(|A_i| - 1) \cdot M}{\alpha^i(|A_i \cup B_i| - 1)}
	\enspace,
\end{multline*}
where the first inequality follows from Markov's inequality, and the second equality holds since, given $a_h = u_j$, every element of $A_i$ other than $a_h$ has a chance of $|H_{i, j - 1}| / [|A_i \cup B_i| - 1]$ to appear in one of the positions of $H_{i, j - 1}$. If $|H_{i, j - 1}| > M$, then we still get $\Pr[\cE_{i, h} \mid a_h = u_j] = 0 \leq \frac{(|A_i| - 1) \cdot M}{\alpha^i(|A_i \cup B_i| - 1)}$ by the definition of $\cE_{i, h}$. Thus, by the law of total expectation,
\begin{align*}
	\Pr[\cE_{i, h}]
	={} &
	\sum_{j = 1}^{|A_i \cup B_i|} \mspace{-9mu} \Pr[a_h = u_j] \cdot \Pr[\cE_{i, h} \mid a_h = u_j]
	=
	\frac{\sum_{j = 1}^{|A_i \cup B_i|} \Pr[\cE_{i, h} \mid a_h = u_j]}{|A_i \cup B_i|}\\
	\leq{} &
	\frac{\sum_{j = 1}^{|A_i \cup B_i|} \frac{(|A_i| - 1) \cdot M}{\alpha^i(|A_i \cup B_i| - 1)}}{|A_i \cup B_i|}
	=
	\frac{(|A_i| - 1) \cdot M}{\alpha^i \cdot (|A_i| + |B_i| - 1)}
	=
	\frac{(|A_i| - 1) \cdot M}{\alpha^i \cdot [(r + 1)|A_i| - 1]}
	\leq
	\frac{M}{\alpha^i r}
	\enspace.
\end{align*}
The lemma now follows by plugging the last bound into Inequality~\eqref{eq:event_sum} and recalling that $|A_i| = \alpha^i k^2$.
\end{proof}

Using the last lemma, we can now bound the expected value of the independent set produced by $ALG$.
\begin{lemma}
Let $S$ denote the independent set outputted by $ALG$. Then, $\bE[f(S)] \leq k + 2\ell + k^4M\ell^2 / r$.
\end{lemma}
\begin{proof}
Let $\cE$ be the intersection of the events $\{\cE_i \mid 1 \leq i \leq \ell\}$. By the union bound and Lemma~\ref{lem:single_level_good}, the event $\cE$ happens with probability at least $1 - k^2M\ell / r$. Given the event $\cE$, we are guaranteed that $ALG$'s memory never contains more than $\alpha^i$ elements of $A_i$ for any $1 \leq i \leq \ell$, and thus, the set $S$ it produces also cannot contain more than this many elements of $A_i$ (since $ALG$ is a non-encoding algorithm). Thus, Lemma~\ref{lem:S_bound} guarantees that given $\cE$, we have
\[
	f(S)
	\leq
	k + \ell + \sum_{i = 1}^\ell |S \cap A_i| \cdot \alpha^{-i}
	\leq
	k + \ell + \sum_{i = 1}^\ell \alpha^i \cdot \alpha^{-i}
	=
	k + 2\ell
	\enspace.
\]
Recalling that Lemma~\ref{lem:opt_value} shows that no independent set (such as $S$) can have a value larger than $\ell \cdot k^2$, we also get
\begin{align*}
	\bE[f(S)]
	={} &
	\Pr[\cE] \cdot \bE[f(S) \mid \cE] + \Pr[\bar{\cE}] \cdot \bE[f(S) \mid \bar{\cE}]
	\leq
	\bE[f(S) \mid \cE] + \Pr[\bar{\cE}] \cdot \bE[f(S) \mid \bar{\cE}]\\
	\leq{} &
	(k + 2\ell) + (k^2M\ell / r) \cdot (\ell \cdot k^2)
	=
	k + 2\ell + k^4M\ell^2 / r
	\enspace.
	\qedhere
\end{align*}
\end{proof}

Since $ALG$ is an arbitrary deterministic non-encoding data stream algorithm, and the distribution of the arrival order of the elements of $\cN$ is independent of the identity of $ALG$, we get the following corollary by Yao's principle. 
\begin{corollary} \label{cor:hardness_raw}
For every (possibly randomized) non-encoding data stream algorithm $ALG$ using a memory of at most $M$ elements, there exists an order over the elements of $\cN$ such that feeding these elements to $ALG$ in this order makes $ALG$ produce a set of value at most $k + 2\ell + k^4M\ell^2 / r$ in expectation. 
\end{corollary}

The last component that we need before getting to the proof of Theorem~\ref{thm:hardness} itself is the following observation, which relates the size $n$ of the ground set $\cN$ to the values of the parameters of our constructed independence system.
\begin{observation} \label{obs:ground_size}
It holds that $n = \Theta(k^2\alpha^\ell r)$.
\end{observation}
\begin{proof}
By the definition of the ground set $\cN$,
\begin{align*}
	n
	={} &
	|\cN|
	=
	\sum_{i = 1}^\ell (|A_i| + |B_i|)
	=
	\sum_{i = 1}^\ell \alpha^i k^2(1 + r)\\
	={} &
	\alpha k^2(1 + r) \cdot \frac{\alpha^\ell - 1}{\alpha - 1}
	=
	\Theta(1) \cdot k^2(1+r)(\alpha^\ell - 1)
	=
	\Theta(k^2\alpha^\ell r)
	\enspace,
\end{align*}
where the penultimate equality holds since $\alpha \geq 2$.
\end{proof}

We are now ready to prove Theorem~\ref{thm:hardness}.
\begin{proof}[Proof of Theorem~\ref{thm:hardness}]
Since the guarantee of the theorem is asymptotic, we may assume $p \geq 2$. Thus, 
in this proof, we may set the values of the parameters $k$, $\alpha$ and $\ell$ to $p - 1$, $k^2 + 1$ and $k$, respectively. Recall that, by Lemma~\ref{lem:k-system}, these choices of values for $k$ and $\alpha$ guarantee that $(\cN, \cI)$ is a $p$-system.

To prove the first part of the theorem, we need to show that every non-encoding data stream algorithm $ALG$ with a memory of $O(n / p^{2p + 5}) = O(n / p^{2k + 7})$ achieves an approximation ratio no better than $\Omega(p^2)$ when the elements of $(\cN, \cI)$ are fed to it in some algorithm specific order.
By Corollary~\ref{cor:hardness_raw}, there is such an algorithm specific arrival order for which $ALG$ outputs a set whose expected value is only $k + 2\ell + k^4M\ell^2 / r = k + 2\ell + k^4 \ell^2 \cdot O(n / p^{2k + 7}) / r$. Since Lemma~\ref{lem:opt_value} shows that the maximum value of an independent set of $(\cN, \cI)$ is $\ell \cdot k^2$, the approximation ratio of $ALG$ for this arrival order is no better than
\begin{align*}
	\frac{\ell \cdot k^2}{k + 2\ell + k^4 \ell^2 \cdot O(n / p^{2k + 7}) / r}
	={} &
	\frac{\ell \cdot k^2}{k + 2\ell + k^4 \ell^2 \cdot O(k^2\alpha^\ell r / p^{2k + 7}) / r}\\
	={} &
	\frac{k^3}{3k + k^6 \cdot O(k^2p^{2k} r / p^{2k + 7}) / r}
	=
	\frac{k^2}{3 + O(1)}
	=
	\Omega(k^2)
	=
	\Omega(p^2)
	\enspace,
\end{align*}
where the first equality follows from Observation~\ref{obs:ground_size}, and the second equality follows from our choice of value for $\ell$ and the observation that $\alpha = k^2 + 1 \leq (k + 1)^2 = p^2$.

To prove the second part of the theorem, we need to show that the $p$-system $(\cN, \cI)$ has the extra properties described in this part, namely, that (i) the largest independent set in it is of size $O(p^{2p})$, and (ii) the size $n$ of $\cN$ is larger than the any given function of $p$. Property~(i) holds since Lemma~\ref{lem:opt_value} shows that the maximum size of an independent set of $(\cN, \cI)$ is $O(k^2\alpha^\ell) = O(k^2p^{2k}) = O(p^{2k + 2}) = O(p^{2p})$. To get Property~(ii), we observe that, by Observation~\ref{obs:ground_size}, $n$ can be made arbitrarily large by choosing the parameter $r$ to be large enough. Since there is no constraint on how large the parameter $r$ can be compared to the other parameters, we can always choose it to be large enough so that $n$ is larger than any given function of $p$.
\end{proof}

